\documentclass[11pt]{article}

\usepackage{newa4}      % DIN A4 Format
\usepackage{latexsym}   % providing 11 symbols of LaTeX 2.09
\usepackage{amsmath}    % \text,\boldsymbol,\pmb,\DeclareMathOperator
\usepackage{amsthm}
\usepackage{amssymb}    % symbols of the amsfonts package
\usepackage{amsfonts}   % AMS fonts: \mathbb,\mathfrak

\newtheorem{lemma}{Lemma}[section]
\newtheorem{defin}[lemma]{Definition}
\newenvironment{definition}{\begin{defin}\em}{\end{defin}}
\newtheorem{proposition}[lemma]{Proposition}
\newtheorem{theorem}[lemma]{Theorem}
\newtheorem{corollary}[lemma]{Corollary}
\newcounter{claim}[lemma]
\newenvironment{claim}{\begin{list}{}{\listparindent\parindent%
                                      \leftmargin0cm\parsep\parskip}%
                                      \item[] \refstepcounter{claim}%
                                      {\em Claim \arabic{claim} }}%
                                      {\end{list}}

\newcommand{\apmap}[3]{#1 \colon #2 \Vdash #3}
\newcommand{\dda}{\mathord{\mbox{\makebox[0pt][l]{\raisebox{-.4ex}
                           {$\downarrow$}}$\downarrow$\,}}}
\newcommand{\down}{\mathclose\downarrow}                           
\newcommand{\fun}[3]{#1 \colon #2 \rightarrow #3}                           
\newcommand{\set}[2]{\{\,#1 \mid #2 \,\}}
\newcommand{\fcbael}[2]{\pair{\mathfrak{#1} | \mathfrak{#2}}} 
\newcommand{\fcbaset}[2]{\pair{\mathcal{#1} | \mathcal{#2}}} 
\newcommand{\fsubset}{\subseteq_\mathrm{fin}}
\newcommand{\low}{\mathopen\downarrow\,}
\newcommand{\pair}[1]{\langle #1 \rangle}

\def\ac{\mathop{\mathstrut\rm AC}}
\def\am{\mathop{\mathstrut\rm AM}}
\def\ap{\mathop{\mathstrut\rm AP}}
\def\Cl{\mathop{\mathstrut\rm CL}}
\def\con{\mathop{\mathstrut\rm Con}\nolimits}
\def\Con{\mathop{\mathstrut\rm CON}\nolimits}
\def\ds{\mathop{\mathstrut\rm DS}}
\def\Ev{\mathop{\mathstrut\rm EV}}
\def\Id{\mathop{\mathstrut\rm Id}}
\def\fct{\mathop{\mathstrut\rm fct}}
\def\pr{\mathop{\mathstrut\rm pr}\nolimits}
\def\rs{\mathop{\mathstrut\rm RS}}
\def\sp{\mathop{\mathstrut\rm SP}}
\def\st{\mathop{\mathstrut\rm st}}
\def\St{\mathop{\mathstrut\rm ST}}
\def\ucl{\overline{\Cl}}

\begin{document}

\title{Information Systems with Witnesses:\\ The Function Space Construction
\thanks{The research leading to these results has received funding from the People Programme (Marie Curie Actions) of the European Union's Seventh Framework Programme FP7/2007-2013/ under REA grant agreement no. PIRSES-GA-2013-612638-CORCON.}}
\author{Dieter Spreen\\
Department of Mathematics, University of Siegen\\
 57068 Siegen, Germany\\
and\\
Department of Decision Sciences, University of South Africa\\
P.O. Box 392, 0003 Pretoria, South Africa}
\date{}

\maketitle

\begin{abstract}
Information systems with witnesses have been introduced in Ref.~\cite{sp21} as a logic-style representation of L-domains: The category of such information systems with approximable mappings as morphisms is equivalent to the category of L-domains with Scott continuous functions, which is known to be Cartesian closed. In the present paper a direct proof of the Cartesian closure of the category of information systems with witnesses and approximable mapppings is given. As is shown, the collection of approximable mappings between two information systems with witnesses comes with a natural information system structure.
\end{abstract}

\section{Introduction}\label{sec-intro}

In a recent paper (Ref.~\cite{sp21}), the author introduced information systems with witnesses as a logic-style representation of  L-domains. L-domains have independently been introduced  by Coquand (cf.\ Ref.~\cite{co89}) and Jung (cf.\ Refs.~\cite{ju89,ju90}) and are known to form one of the two maximal Cartesian closed full subcategories of the continuous domains. They generalise the bounded-complete ones: whereas in a bounded-complete domain every bounded subset has a global least upper bound, in an L-domain such sets may have different local least upper bounds depending on the upper bounds they have.

The idea to represent classes of domains via logical calculi goes back to Dana Scott's seminal 1982 paper (cf.\ Ref.~\cite{sco82}), in which he  introduced information systems to capture the bounded-complete algebraic domains. An information system consists of a set of tokens to be thought of as atomic statements about a computational process, a consistency predicate telling us which finite sets of such statements contain consistent information, and an entailment relation saying what atomic statements are entailed by which consistent sets of these. Theories  of such a logic, also called states, i.e.\ finitely consistent and entailment-closed sets of atomic statements, form a  domain with respect to set inclusion. A state represents consistent information. So, any finite collection of substates must contain consistent information as well, and this fact is witnessed by any of  its upper bounds.  

Whereas in Scott's approach the consistency witnesses are hidden, in the new approach they are made explicit. This allows to consider the more general situation in which a finite set of tokens may have different consistency witnesses, and the result of entailment may depend on them.  As was shown in Ref.~\cite{sp21}, the theories, or states, of such a more general information system form an L-domain, and, up to isomorphism, each L-domain can be obtained in this way. Moreover, there is an equivalence between the categories of information systems with witnesses and L-domains.

The category of information systems has approximable mappings as morphisms. These are relations between the consistent subsets of one information system and the tokens of another. Entailment is a particular approximable mapping.

As  mentioned earlier, the category of L-domains and Scott continuous functions is Cartesian closed. Because of the equivalence of this category with the category of information systems with witnesses we know that the latter one is Cartesian closed as well. However, this means that in concrete situations we have to pass back and forth between information systems and domains in order to construct the exponent of two information systems with witnesses. In this paper we present a direct proof of the Cartesian closure of the category of information systems with witnesses. In particular, we present a  construction of the exponent. This will be an information system with witnesses the states of which are exactly the approximable mappings between the information systems under consideration.  

Whereas for Scott's information systems capturing the bounded-complete algebraic domains, the function space construction is straightforward and well understood, the situation is rather intricate in the present case. This due to the fact that consistency is only locally defined and we have to deal with consistency witnesses explicitly. Moreover, as is known from Hoofman's work (cf.\ Ref.~\cite{ho93}), in the continuous case entailment is required to allow for interpolation.

The paper is organized as follows: Section~\ref{sec-dom} contains basic definitions and results from domain theory. In Section~\ref{sec-infosys} relevant definitions and facts about  information systems with witnesses  are recalled from Ref.~\cite{sp21}. Approximable mappings between such information systems are considered in Section~\ref{sec-am}. Section~\ref{sec-fctsp} is concerned with the function space construction: an information system with witnesses is presented the states of which are exactly the approximable mappings between two given systems.  
As is shown in Section~\ref{sec-cc}, this  information system is indeed the exponent of the two given ones.
The paper closes with the discussion of an application of the results obtained in proof assistants and program extraction.

\section{Domains: basic definitions and results}\label{sec-dom}

For any set $A$, we write $X \fsubset A$ to mean that $X$ is finite subset of $A$. The collection of all subsets of $A$ will be denoted by $\mathcal{P}(A)$ and that of all finite subsets by $\mathcal{P}_f(A)$. Moreover, for two sets $A_1$ and $A_2$, we let $\pr_1$ and $\pr_2$, respectively, be the canonical projections of $A_1 \times A_2$ onto the first and second component. For $\nu, \mu \in \{ 1, 2 \}$, set $\pr_{\nu, \mu} = \pr_\nu \circ \pr_\mu$.

Let $(D, \sqsubseteq)$ be a poset. $D$ is \emph{pointed} if it contains a least element $\bot$. For an element $x \in D$, $\low x$ denotes the principal ideal generated by $x$, i.e., $\low x = \set{y \in D}{y \sqsubseteq x}$. A subset $S$ of $D$ is called \emph{consistent} if it has an upper bound. $S$ is \emph{directed}, if it is nonempty and every pair of elements in $S$ has an upper bound in $S$. $D$ is a \emph{directed-complete partial order} (\emph{dcpo}), if every directed subset $S$ of $D$ has a least upper bound $\bigsqcup S$ in $D$, and $D$ is \emph{bounded-complete} if every consistent subset has a least upper bound.

Assume that $x, y$ are elements of a poset $D$. Then $x$ is said to \emph{approximate} $y$, written $x \ll y$, if for any directed subset $S$ of $D$ the least upper bound of which exists in $D$, the relation $y \sqsubseteq \bigsqcup S$ always implies the existence of some $u \in S$ with $x \sqsubseteq u$. Moreover, $x$ is \emph{compact} if $x \ll x$. A subset $B$ of $D$ is a \emph{basis} of $D$, if for each $x \in D$ the set $\dda\!_B x = \set{u \in B}{u \ll x}$ contains a directed subset with least upper bound $x$. Note that the set of all compact elements of $D$ is included in every basis of $D$.  A directed-complete partial order $D$ is said to be \emph{continuous} (or a \emph{domain}) if it has a basis and it is called \emph{algebraic} (or an \emph{algebraic domain}) if its compact elements form a basis. A pointed bounded-complete domain is called \emph{bc-domain}. Standard references for domain theory and its applications are \cite{gs, gu, aj, dom, ac, gie}.

\begin{lemma}\label{lem-preordprop}
In a poset $D$ the following statements hold for all $x, y, z \in D$: \begin{enumerate}
\item\label{lem-preordprop-0} The approximation relation $\ll$ is transitive.
\item\label{lem-preordprop-1} $x \ll y \Longrightarrow x \sqsubseteq y$.
\item\label{lem-preordprop-2} $x \ll y \sqsubseteq z \Longrightarrow x \ll z$.
\item\label{lem-preordprop-4} If $D$ has a least element $\bot$, then $\bot \ll x$.
\item\label{lem-preordprop-3} If $F \subseteq \low x \cap \low y$ such that the least upper bounds $\bigsqcup^x F$ and $\bigsqcup^y F$, respectively, exist relative to $\low x$ and $\low y$, then 
\[
x, y \sqsubseteq z \Longrightarrow \bigsqcup\nolimits^x F = \bigsqcup\nolimits^y F.
\]
\item\label{lem_preordprop-5} If $D$ is a continuous domain with basis $B$, and $M \fsubset D$, then
\[
M \ll x \Longrightarrow (\exists v \in B) M \ll v \ll x,
\]
where $M \ll x$ means that $m \ll x$, for any $m \in M$.

\end{enumerate}
\end{lemma}
Property~\ref{lem_preordprop-5} is known as the \emph{interpolation law}.

\begin{definition}
Let $D$ and $D'$ be posets. A function $\fun{f}{D}{D'}$ is \emph{Scott continuous} if it is monotone and for any directed subset $S$ of $D$ with existing least upper bound,
\[
\bigsqcup f(S) = f(\bigsqcup S).
\]
\end{definition}

With respect to the pointwise order the set $[D \to D']$ of all Scott continuous functions between two dcpo's $D$ and $D'$ is a dcpo again. Observe that it need not be continuous even if $D$ and $D'$ are. This is the case, however, if $D'$ is an L-domain (cf.\ Ref.~\cite{aj}).

\begin{definition}
A pointed\footnote{Note that in Ref.~\cite{gie} pointedness is not required.} domain $D$ is an \emph{L-domain}, if each pair $x, y \in D$ bounded above by $z \in D$ has a least upper bound $x \sqcup^z y$ in $\low z$.
\end{definition}

Obviously, every bc-domain is an L-domain.
As has been shown by Jung (cf.\ Refs.~\cite{ju89,ju90}), the category $\mathbf{L}$ of L-domains is one of the two maximal Cartesian closed full subcategories of the category $\mathbf{CONT_\perp}$ of pointed domains and Scott continuous maps. The same holds for the category $\mathbf{aL}$ of algebraic L-domains with respect to the category $\mathbf{ALG_\perp}$ of pointed algebraic domains.  The one-point domain is the terminal object in these categories and the categorical product $D \times E$  of two domains $D$ and $E$ is the Cartesian product of the underlying sets  ordered coordinatewise.

For domains $D$ and $D'$ and basic elements $d \in D$ and $d' \in D'$ the \emph{single-step function} $\fun{(d \searrow d')}{D}{D'}$ is defined by
\[
(d \searrow d')(x) = \begin{cases}
					d' & \text{if $d \ll x$,}\\
					\bot' & \text{otherwise.}
				\end{cases}
\]				
As is well known, every Scott continuous function $\fun{f}{D}{D'}$ is the least upper-bound of all single-step functions approximating it:
\[
f = \bigsqcup \set{(d \searrow d')}{d' \ll f(d)}.
\]
In general, however, the set of these single-step functions is not directed. A way to get out of this problem is to require the existence of joins of bounded finite collections of single-step functions. Such joins are called \emph{step functions}. 

If $D'$ is bounded-complete, the pointwise least upper bound $\bigsqcup_{\nu = 1}^n (d_\nu \searrow d'_\nu)$ exists, if the set $\set{d'_\nu}{1 \le \nu \le n \wedge d_\nu \ll x}$ is bounded for all $x \in D$. The cost of generalising this to the case of L-domains is at least the burden of bookkeeping where least upper bounds are taken. In particular, if $(d_\nu \searrow d'_\nu)$, $\nu = 1, \ldots, n$, are single-step functions below $f$, then their least upper bound in $\down f$, written $\bigsqcup_{1 \le \nu \le n}^f (d_\nu \searrow d'_\nu)$, is given by
\[
\bigsqcup_{1 \le \nu \le n}\nolimits^f (d_\nu \searrow d'_\nu)(x) = \bigsqcup_{\nu :\, d_\nu \ll x}\nolimits^{f(x)} d'_\nu.
\]

\section{Information systems with witnesses}\label{sec-infosys}

In this section, the ideas outlined in the introduction are made precise:  An information system with witnesses consists of a set $A$ of tokens, a consistency predicate $\Con$ classifying which finite sets of tokens are consistent with which token as witness, and an entailment relation between pairs of consistent sets and associated witnesses on the one side, and arbitrary tokens on the other.  The conditions that have to be satisfied are grouped. There are requirements which consistency predicate and entailment relation have to meet in which the consistency witness is kept fixed. They are well known from Scott's information systems and Hoofman's extension of this notion to the continuous case.  In addition, we find conditions that specify the interplay between consistency witnesses. 

Note that we sometimes write $X \in \Con(i)$ instead of $(i, X) \in \Con$. Proofs of the results can be found in Ref.~\cite{sp21}.

\begin{definition}\label{dn-infsys}
Let $A$ be a set, $\Delta \in A$, $\Con \subseteq A \times \mathcal{P}_f(A)$, and $\mbox{$\vdash$} \subseteq \Con \times A$. Then $(A, \Con, \vdash, \Delta)$ is an \emph{information system with witnesses} if the following conditions hold, for all $i, j, a \in A$ and all finite subsets $X, Y$ of $A$:
\begin{enumerate}
\item\label{dn-infsys-1}
$\{i\} \in \Con(i)$

\item\label{dn-infsys-2}
$Y \subseteq X \wedge X \in \Con(i) \Rightarrow Y\in \Con(i)$

\item\label{dn-infsys-3}
$(i, \emptyset) \vdash \Delta$

\item\label{dn-infsys-4}
$X \in \Con(i) \wedge (i, X) \vdash Y \Rightarrow Y \in \Con(i)$

\item\label{dn-infsys-5}
$X, Y \in \Con(i) \wedge X \subseteq Y \wedge (i, X) \vdash a \Rightarrow (i, Y) \vdash a$

\item\label{dn-infsys-6}
$X \in \Con(i) \wedge (i, X) \vdash Y \wedge (i, Y) \vdash a \Rightarrow (i, X) \vdash a$

\item\label{dn-infsys-10}
$(i,X) \vdash a \Rightarrow (\exists Z \in \Con(i)) (i,X)\vdash Z  \wedge (i, Z) \vdash a$

\item\label{dn-infsys-11}
$(i,X) \vdash Y \Rightarrow (\exists e \in A) (i, X)\vdash e  \wedge Y \in \Con(e)$

\item\label{dn-infsys-7}
$\{i\} \in \Con(j) \Rightarrow\Con(i) \subseteq \Con(j)$

\item\label{dn-infsys-8}
$\{i\}\in \Con(j) \wedge X \in \Con(i) \wedge(i,X) \vdash a \Rightarrow (j,X) \vdash a$

\item\label{dn-infsys-9}
$\{i\}\in\Con(j) \wedge X\in\Con(i) \wedge (j, X) \vdash a \Rightarrow (i,X) \vdash a$

\end{enumerate}
\end{definition}

All requirements are very natural:  Each token witnesses its own consistency (\ref{dn-infsys-1}). If the consistency of some set is witnessed by $i$, the same holds for all of its subsets (\ref{dn-infsys-2}).  $\Delta$ is entailed by any set of information and each of its witnesses, i.e., it represents global truth~(\ref{dn-infsys-3}). By Condition~(\ref{dn-infsys-4})  entailment  preserves consistency. If a set $X$ with consistency witness $i$ entails $a$, so does any bigger set with the same witness (\ref{dn-infsys-5}). For fixed witness, entailment is transitive~(\ref{dn-infsys-6}).  Consistency and entailment are preserved when moving from a witness $i$ to another one $j$ with respect to which $i$ is consistent  (\ref{dn-infsys-7}, \ref{dn-infsys-8}). Moreover, entailment is conservative in that case: what is entailed with respect to witness $j$ from a set consistent with respect to $i$ is already entailed with respect to $i$ (\ref{dn-infsys-9}). Conditions~(\ref{dn-infsys-10}) and (\ref{dn-infsys-11}) are both interpolation properties. They can be combined into one, called \emph{Global Interpolation Property}.

\begin{lemma}\label{lem-globint}
Let $A$ be a set, $\Delta \in A$, $\Con \subseteq A \times \mathcal{P}_f(A)$, and $\mbox{$\vdash$} \subseteq \Con \times A$ such that Axioms~\ref{dn-infsys}(\ref{dn-infsys-4}, \ref{dn-infsys-5}, \ref{dn-infsys-7}-\ref{dn-infsys-9}) are satisfied. Then Axioms~\ref{dn-infsys}(\ref{dn-infsys-10}, \ref{dn-infsys-11}) hold if, and only if, for all $i \in A$, $X \in \con_i$ and $F \fsubset A$,
\begin{equation}\label{eq-gip}
(i, X) \vdash F \Rightarrow (\exists (j, Y) \in \Con) (i, X)  \vdash (j, Y) \wedge (j, Y) \vdash F, \tag{GIP}
\end{equation}
where $(i, X) \vdash (j, Y)$ means that $(i, X) \vdash j$ and $(i, X) \vdash Y$.
\end{lemma} 

The next result extends Condition~\ref{dn-infsys}(\ref{dn-infsys-10}).

\begin{lemma}\label{lem-setax7}
Let $A$ be a set, $\Delta \in A$, $\Con \subseteq A \times \mathcal{P}_f(A)$, and $\mbox{$\vdash$} \subseteq \Con \times A$ such that Axioms~\ref{dn-infsys}(\ref{dn-infsys-4}, \ref{dn-infsys-5}, \ref{dn-infsys-10}) are satisfied. Then the following rule holds, for all $a \in A$, $F \fsubset A$ and $(i, X) \in \Con$,
\[
(i,X) \vdash a \Rightarrow (\exists Z \in \Con(i)) (i,X)\vdash Z  \wedge (i, Z) \vdash F.
\]
\end{lemma}
\begin{proof}
Let $b \in F$. By Axiom~\ref{dn-infsys}(\ref{dn-infsys-10}) there is some $Z_b \in \Con(i)$ with $(i, X) \vdash Z_b$ and $(i, Z_b) \vdash b$. Set $Z = \bigcup \{\, Z_b \mid b \in F \,\}$. Then $(i, X) \vdash Z$. Hence, $Z \in \Con(i)$, by Condition~\ref{dn-infsys}(\ref{dn-infsys-4}). Because of Axiom~\ref{dn-infsys}(\ref{dn-infsys-5}) we therefore have that $(i, Z) \vdash F$.
\end{proof}

Sometimes a stronger requirement than \ref{dn-infsys}(\ref{dn-infsys-6}) is needed which reverses Condition~\ref{dn-infsys}(\ref{dn-infsys-10}).

\begin{lemma}\label{lem-strong6}
Let $(A, \Con, \vdash, \Delta)$ be an information system with witnesses. Then the following rule holds, for all $a \in A$ and $(i, X),(j, Y) \in \Con$,
\[
(i,X) \vdash (j, Y)  \wedge (j, Y) \vdash a \Rightarrow (i,X) \vdash a.
\]
\end{lemma}

To relate information systems to domains, the notion of state is required.

\begin{definition}\label{dn-st}
Let $(A, \Con, \vdash, \Delta)$ be an information system with witnesses. A subset $x$ of $A$ is a \emph{state} of $(A, \Con, \vdash, \Delta)$ if the following three conditions hold:
\begin{enumerate}
\item\label{dn-st-1}
$(\forall F \fsubset x) (\exists i \in x) F \in \Con(i)$

\item\label{dn-st-2}
$(\forall i \in x)(\forall X \fsubset x) (\forall a \in A) [X \in \Con(i) \wedge (i,X) \vdash a \Rightarrow a \in x]$

\item\label{dn-st-3}
$(\forall a \in x) (\exists i \in x) (\exists X \fsubset x) X \in \Con(i) \wedge (i,X) \vdash a.$ 

\end{enumerate}
\end{definition}

As follows from the definition, states are subsets of tokens that are \emph{finitely consistent} (\ref{dn-st-1}) and \emph{closed under entailment} (\ref{dn-st-2}). Furthermore, each token in a state is \emph{derivable} (\ref{dn-st-3}), i.e.\ for each token the state contains a consistent set and its witness entailing the token.

By Condition~\ref{dn-st}(\ref{dn-st-1}) states are never empty: Choose $F$ to be the empty set. Then the state contains some $i$ with $\emptyset \in \Con(i)$.

Note that Conditions~(\ref{dn-st-1}, \ref{dn-st-3}) in Definition~\ref{dn-st} can be replaced by a single requirement.

\begin{proposition}\label{pn-stsing}
Let $(A, \Con, \vdash, \Delta)$ be an information system with witnesses and $x$ be a subset of $A$. Then Conditions~\ref{dn-st}(\ref{dn-st-1}) and (\ref{dn-st-3}) together are equivalent to the following statement:
\begin{equation}\tag{ST}\label{st}
(\forall F \fsubset x) (\exists i \in x) (\exists X \fsubset x) X \in \Con(i) \wedge (i,X) \vdash F.
\end{equation}
\end{proposition}

With respect to set inclusion the states of $A$ form a  directed-complete partially ordered set, denoted by $|A|$. Moreover, the consistent subsets of $A$ generate a canonical basis of $|A|$. For $(i,X) \in \Con$ let
\[
[X]_i = \set{a \in A}{(i,X) \vdash a}.
\]
Then $[X]_i$ is a state and for every $z \in |A|$, the set of all $[X]_i$ with  $\{i\} \cup X \subseteq z$ is directed and $z$ is its union.

This result allows characterising the approximation relation on $A$ in terms of the entailment relation. The characterisation nicely reflects the intuition that $x \ll y$, if $x$ is covered by a ``finite part'' of $y$.

\begin{proposition}\label{lem-app}
For $x, y \in |A|$,
\[
x \ll y \Longleftrightarrow (\exists (i,V) \in \Con) \{i\} \cup V \subseteq y \wedge (i,V) \vdash x.
\]
\end{proposition}

Because of Axioms~\ref{dn-infsys}(\ref{dn-infsys-1}, \ref{dn-infsys-2}) we have that $\emptyset \in \Con(i)$, for all $i \in A$. Moreover, with Axioms~\ref{dn-infsys}(\ref{dn-infsys-3}, \ref{dn-infsys-4}), it follows that $\{\Delta\} \in \Con(j)$,  for all $j \in A$. As is easily verified, $[\emptyset]_i = [\{ \Delta \}]_j$, for all $i, j \in A$, and $[\emptyset]_\Delta \subseteq x$, for all $x \in |A|$.

Local least upper bounds can be computed in a similar way as directed least upper bounds. Let $x, y, z \in |A|$ so that $x, y \subseteq z$. Then 
\[
\bigcup\set{[Z]_k}{(k, Z) \in \Con \wedge\, k \in z \wedge Z \fsubset x \cup y}
\]
is the least upper bound of $x$ and $y$ in $\low z$.

\begin{theorem}\label{tm-ldom}
Let $(A, \Con, \vdash, \Delta)$ be an information system with witnesses. Then $\mathcal{L}(A) = (|A|, \subseteq, [\emptyset]_\Delta)$ is an L-domain with basis $\widehat{\Con} = \set{[X]_i}{ (i,X) \in \Con}.$
\end{theorem}

Let us see next when $\mathcal{L}(A)$ is algebraic.

\begin{definition}\label{dn-refl}
Let $(A, \Con, \vdash, \Delta)$ be an information system with witnesses. An element $(j,V) \in \Con$ is called \emph{reflexive} if $(j,V) \vdash (j, V)$.
\end{definition}

 Obviously, $[V]_j$ is compact, for every reflexive  $(j,V) \in \Con$. We denote the subset of reflexive elements of $\Con$ by $\Con_\mathrm{refl}$.

\begin{theorem}\label{tm-alg}
Let $(A, \Con, \vdash, \Delta)$ be an information system with witnesses. Then $\mathcal{L}(A)$ is algebraic if, and only if, the information system $A$ satisfies Condition~(\ref{alg})  saying that for all $(i,X) \in \Con$ and $F \fsubset A$, 
\begin{equation}\tag{ALG}\label{alg}
(i, X) \vdash F  \Rightarrow (\exists (j,V) \in \Con_\mathrm{refl}) (i, X) \vdash  (j, V) \wedge (j, V) \vdash F.
\end{equation}
\end{theorem}

Condition (\ref{alg})  is a global interpolation requirement. Similarly to  Lemma~\ref{lem-globint} it is equivalent to a local condition.

\begin{lemma}\label{lem-salg}
Condition~(\ref{alg}) holds if, and only if, the following Condition~(\ref{salg}) is satisfied for all $(i, X) \in \Con$ and $a \in A$,
\begin{equation} \tag{SALG}\label{salg}
(i, X) \vdash a \Rightarrow (\exists Z \in \Con(i)) (i, X) \vdash Z \wedge (i, Z) \vdash Z \wedge (i, Z) \vdash a.
\end{equation}
\end{lemma}

In Scott's information systems a finite set of tokens is consistent, if it has a consistency witness, independently of which token this might be. This provides us with a condition forcing an information system with witnesses to generate a bounded-complete domain. 

\begin{theorem}\label{tm-bc}
Let $(A, \Con, \vdash, \Delta)$ be an information system with witnesses. Then $\mathcal{L}(A)$ is bounded-complete, and hence a bc-domain, if the information system $A$ satisfies Condition~(\ref{bc}) saying that for all $X \fsubset A$ and $i, j \in A$,
\begin{equation}\tag{BC}\label{bc}
(i, X), (j, X) \in \Con \Rightarrow (\forall a \in A)[(i, X) \vdash a \Leftrightarrow (j, X) \vdash a].
\end{equation}
\end{theorem}

It is unknown whether the requirement on $A$ is also necessary. 

So far, we have seen that information systems with witness generate L-domains. But the converse holds as well: Every L-domain defines a canonical information system with witnesses such that the L-domain generated by it is isomorphic to the given domain. 

Let $(D, \sqsubseteq)$ be an L-domain with basis $B$ and least element $\bot$. Set  $\mathcal{I}(D) = (B, \Con, \vdash, \bot)$ with
\[
 \Con = \set{(i, X)}{i \in B \wedge X \fsubset \low i \cap B} \,\,\text{and}\,\, (i, X) \vdash a \Longleftrightarrow a \ll \bigsqcup\nolimits^i X.
\]

\begin{theorem}\label{tm-dominsys}
Let $D$ be an L-domain. Then $\mathcal{I}(D)$ is an information system with witnesses such that $D$ and $\mathcal{L}(\mathcal{I}(D))$ are isomorphic. In addition,
\begin{enumerate}
\item\label{tm-dominsys-1}
$D$ is algebraic if, and only if, the information system $\mathcal{I}(D)$ satisfies Condition~(\ref{alg}).
\item\label{tm-dominsys-2}
$D$ is bounded-complete if, and only if, Condition~(\ref{bc}) holds in $\mathcal{I}(D)$. 
\end{enumerate}
\end{theorem}

\section{Approximable mappings}\label{sec-am}

In the next step we want to consider the appropriate morphisms between information systems with witnesses.  They will be relations between the consistent sets and their consistency witnesses of one information system with witnesses and the tokens of another, just as the entailment relations.

\begin{definition}\label{dn-am}
An \emph{approximable mapping} $H$ between information systems with witnesses $(A, \Con, \vdash, \Delta)$ and $(A', \Con', \vdash', \Delta')$, written $\apmap{H}{A}{A'}$, is a relation between $\Con$ and $A'$ satisfying the following nine conditions, for all $i, j \in A$, $X, X' \fsubset A$, $b, k \in A'$ and $Y, F \fsubset A'$ with $X \in \Con(i)$ and $Y \in \Con'(k)$:
\begin{enumerate}
\item\label{dn-am-6}
$(\Delta, \emptyset) H \Delta'$.

\item\label{dn-am-2}
$X' \in \Con(i) \wedge X \subseteq X' \wedge (i, X) H b \Rightarrow (i, X') H b$

\item\label{dn-am-3}
$(i, X) \vdash X' \wedge (i, X') H b \Rightarrow (i, X) H b$

\item\label{dn-am-5-1}
$(i, X) H b \Rightarrow (\exists U \in \Con(i)) (i, X) \vdash U \wedge (i, U) H b$

\item\label{dn-am-1}
$(i, X) H (k, Y) \wedge (k, Y) \vdash' b \Rightarrow (i, X) H b$

\item\label{dn-am-5-2}
$(i, X) H b \Rightarrow (\exists (d, V) \in \Con') (i, X) H (d, V)  \wedge (d, V) \vdash' b$

\item\label{dn-am-8}
$(i, X) H F \Rightarrow (\exists e \in A') (i, X) H e \wedge F \in \Con'(e)$.

\item\label{dn-am-4}
$\{ i \} \in \Con(j)  \wedge (i, X) H b \Rightarrow (j, X) H b$

\item\label{dn-am-7}
$\{ i \} \in \Con(j) \wedge (j, X) H b \Rightarrow (i, X) H b$

\end{enumerate}
Here, $(i, X) H Y$ means that $(i, X) H c$, for all $c \in Y$,  and $(i, X) H (k, Y)$ that $(i, X) H k$ as well as $(i, X) H Y$.
\end{definition}

Note that because of Condition~\ref{dn-infsys}(\ref{dn-infsys-5}), Condition~\ref{dn-am}(\ref{dn-am-2}) follows from Conditions~\ref{dn-am}(\ref{dn-am-3}, \ref{dn-am-5-1}).
The Left Interpolation Rule~\ref{dn-am}(\ref{dn-am-5-1}) together with the Conservativity Requirement~\ref{dn-am}(\ref{dn-am-7}) can be replaced by one rule.

\begin{lemma}\label{pn-amleft}
Let $(A, \Con, \vdash, \Delta)$ and $(A', \Con', \vdash', \Delta')$ be information systems with witnesses. Then, for any $H \subseteq \Con \times A'$ satisfying Conditions~\ref{dn-am}(\ref{dn-am-2}, \ref{dn-am-3}, \ref{dn-am-4}), any $(i, X) \in \Con$, and all $F \fsubset A'$, Conditions~\ref{dn-am}(\ref{dn-am-5-1}, \ref{dn-am-7}) together are equivalent to the following requirement:
\begin{equation}\label{eq-amleft}
(i, X) H F \Rightarrow (\exists (c, U) \in \Con) (i, X) \vdash (c, U) \wedge (c, U) H F.
\end{equation}
\end{lemma}
\begin{proof}
Assume that Requirement~(\ref{eq-amleft}) holds. Then Condition~\ref{dn-am}(\ref{dn-am-5-1}) follows as a special case because of Condition~\ref{dn-am}(\ref{dn-am-4}). For Condition~\ref{dn-am}(\ref{dn-am-7}) suppose that $\{ i \} \in \Con(j)$ and $(j, X) H b$. With ~Requirement~(\ref{eq-amleft}) we obtain that there is some $(c, U) \in \Con$ so that $(j, X) \vdash (c, U)$ and $(c, U) H b$. Since $X \in \Con(i)$, it follows with Axiom~\ref{dn-infsys}(\ref{dn-infsys-9}) that $(i, X) \vdash (c, U)$. Thus, $\{ c \} \in \Con(i)$ and therefore $(i, U) H b$, by Condition~\ref{dn-am}(\ref{dn-am-4}). With Condition~\ref{dn-am}(\ref{dn-am-3}) we finally obtain that $(i, X) H b$.

Now, conversely, suppose that Conditions \ref{dn-am}(\ref{dn-am-5-1}) and \ref{dn-am}(\ref{dn-am-7}) hold. Moreover, assume that $(i, X) H F$. Then $(i, X) H b$, for all $b \in F$. Thus, there are $U_b \in \Con(i)$ with $(i, X) \vdash U_b$ and $(i, U_b) H b$. Set $U = \set{u_b}{b \in F}$. It follows that $(i, X) \vdash U$. Hence, $U \in \Con(i)$. With Condition~\ref{dn-am}(\ref{dn-am-2}) we therefore have that $(i, U) H F$. By Axiom~\ref{dn-infsys}(\ref{dn-infsys-11}) we moreover obtain from $(i, X) \vdash U$ that there is some $c \in A$ with $(i, X) \vdash c$ and $U \in \Con(c)$. Thus, $\{ c \} \in \Con(i)$. Consequently, $(c, U) H F$, because of Condition~\ref{dn-am}(\ref{dn-am-7}).
\end{proof}

Similarly, the Right Interpolation Rule~\ref{dn-am}(\ref{dn-am-5-2}) and the Witness Generation Rule~\ref{dn-am}(\ref{dn-am-8}) can be combined into one rule.

\begin{lemma}\label{pn-amright}
Let $(A, \Con, \vdash, \Delta)$ and $(A', \Con', \vdash', \Delta')$ be information systems with witnesses. Then, for any $H \subseteq \Con \times A'$, $(i, X) \in \Con$, and $F \fsubset A'$, Conditions~\ref{dn-am}(\ref{dn-am-5-2}, \ref{dn-am-8})  together are equivalent to the following requirement:
\begin{equation}\label{eq-amright}
(i, X) H F \Rightarrow (\exists (e, V) \in \Con') (i, X) H (e, V) \wedge (e, V) \vdash' F.
\end{equation}
\end{lemma}
\begin{proof}
Assume that Requirement~(\ref{eq-amright}) holds. Then Condition~\ref{dn-am}(\ref{dn-am-5-2}) follows as a special case. For Condition~\ref{dn-am}(\ref{dn-am-8}) suppose that $(i, X) H F$. Then there is some $(e, V) \in \Con')$ so that $(i, X) H (e, V)$ and $(e, V) \vdash' F$. It follows that $(i, X) H e$ and $F \in \Con'(e)$.

Next, conversely, suppose that Conditions~\ref{dn-am}(\ref{dn-am-5-2}) and \ref{dn-am}(\ref{dn-am-8}) hold. Moreover, assume that $(i, X) H F$. Then, for all $b \in F$, there exist $(d_b, V_b) \in \Con'$ such that $(i, X) H (d_b, V_b)$ and $(d_b, V_b) \vdash' b$. Let $V = \set{V_b \cup \{  d_b \}}{b \in F}$. Since we have that $(i, X) H V$, it follows with Requirement~\ref{dn-am}(\ref{dn-am-8}) that there is some $e \in A'$ with $(i, X) H e$ and $V \in \Con'(e)$.  Hence, $\{ d_b \} \in \Con'(e)$ and thus $(e, V_b) \vdash' b$, from which we obtain that $(e, V) \vdash' F$. In addition, we have that $(i, X) H (e, V)$.
\end{proof}

Finally, the extended left and right Interpolation Rules~(\ref{eq-amleft}) and (\ref{eq-amright}) together can be exchanged for one rule.

\begin{lemma}\label{pn-amint}
Let $(A, \Con, \vdash, \Delta)$ and $(A', \Con', \vdash', \Delta')$ be information systems with witnesses. Then, for any $H \subseteq \Con \times A'$, $(i, X) \in \Con$, and $F \fsubset A'$, Conditions~(\ref{eq-amleft}) and (\ref{eq-amright})  together are equivalent to the following rule:
\begin{multline}\label{eq-amint}
%\begin{split}
(i, X) H F \Rightarrow (\exists (c, U) \in \Con) (\exists (e, V) \in \Con')\\
  (i, X) \vdash (c, U) \wedge (c, U) H (e, V) \wedge (e, V) \vdash' F. 
%\end{split}
\end{multline}
\end{lemma}
For a proof see Ref.~\cite{sp21}.

Similarly to Lemma~\ref{lem-strong6} a strengthening of Axiom~\ref{dn-am}(\ref{dn-am-3}) can be derived. It reverses the implication in Condition~\ref{dn-am}(\ref{dn-am-5-1}).

\begin{lemma}\label{lem-amstrong3}
Let $H$ be an approximable mapping between information systems $A$ and $A'$ with witnesses. Then for all $(i, X), (j, Y) \in \Con$ and $b \in A'$,
\[
(i, X) \vdash (j, Y) \wedge (j, Y) H b \Rightarrow (i, X) H b.
\]
\end{lemma}

As has already been mentioned, entailment relations are special approximable mappings. For $(i, X) \in \Con$ and $a \in A$, set $(i, X) \Id_A a$ if $(i, X) \vdash a$. Then $\apmap{\Id}{A}{A}$ such that for all $\apmap{H}{A}{A'}$, $H \circ \Id_{A'} = H = \Id_A \circ H$, where for approximable mappings $\apmap{H}{A}{A'}$ and $\apmap{G}{A'}{A''}$ their composition $\apmap{H \circ G}{A}{A''}$ is defined by
\[
(i, X) (H \circ G) c \Longleftrightarrow (\exists (j, Y) \in \Con') (i, X) H (j, Y) \wedge (j, Y) G c.
\]

Let $\mathbf{ISW}$ be the category of information systems with witnesses and approximable mappings and $\mathbf{aISW}$, $\mathbf{bcISW}$, and $\mathbf{abcISW}$, respectively, be the full categories of information systems with witnesses that satisfy Condition~(\ref{alg}), Condition~(\ref{bc}) or both of them.

\begin{theorem}\label{tm-eqlisw}
The category $\mathbf{ISW}$ of information systems with witnesses and approximable mappings is equivalent to the category $\mathbf{L}$ of L-domains and Scott continuous functions.
\end{theorem}

\begin{corollary}\label{cor-eqabc}
The categories $\mathbf{aISW}$, $\mathbf{bcISW}$ and $\mathbf{abcISW}$, respectively, of information systems with witnesses satisfying Conditions~(\ref{alg}), (\ref{bc}), or both of them, and approximable mappings are equivalent to the categories $\mathbf{aL}$, $\mathbf{BC}$ and $\mathbf{aBC}$ of algebraic L-domain, bc-domains and algebraic bc-domains with Scott continuous functions.
\end{corollary}

\section{The function space construction}\label{sec-fctsp}

As mentioned earlier, the categories $\mathbf{L}$ and $\mathbf{aL}$ are Cartesian closed. The same is true for $\mathbf{BC}$ and $\mathbf{aBC}$. Because of the equivalence of theses categories with $\mathbf{ISW}$, $\mathbf{aISW}$, $\mathbf{bcISW}$ and $\mathbf{abcISW}$, respectively, we know that the latter categories are Cartesian closed as well. In this and the next section we present a direct proof of the Cartesian closure of $\mathbf{ISW}$ and its just mentioned full subcategories. To this end we first show that the collection of approximable mappings between two information systems with witnesses comes itself with a natural information system structure. We start with some preliminary definitions and then discuss what will be the tokens of this information system.

\begin{definition}\label{dn-witeq}
Let $(A, \Con, \vdash, \Delta)$ be an information system with witnesses and $X \fsubset A$. Two witnesses $i, j \in A$ are called \emph{$X$-equivalent}, written $i \sim j \,[X]$, if there are $n \in \omega$ and $k_1, \ldots, k_n, a_0, \ldots, a_n \in A$ with $a_0 = i$, $a_n = j$, $X \in \Con(a_{\nu -1}) \cap \Con(a_\nu)$ and $\{ a_{\nu - 1} \}, \{ a_\nu \} \in \Con(k_\nu)$, for $\nu = 1, \ldots, n$.
\end{definition}

If the information system is generated from an L-domain as in Theorem~\ref{tm-dominsys}, the witnesses $i$ and $j$ are basic elements and $X$ is a finite subset of such, then  $i \sim j \,[X]$ implies that the suprema of $X$ with respect to $i$ and $j$, respectively, coincide.

\begin{lemma}\label{lem-witeqprop}
For any $X \fsubset A$, $\sim [X]$ is a partial equivalence relation. Moreover, the following five statements hold, for all $i, j, k \in A$ and $U \fsubset A$:
\begin{enumerate}
\item\label{lem-witeqprop-0} 
$i \sim \Delta \,[\emptyset]$

\item\label{lem-witeqprop-1}
$i \sim j \,[X] \Longrightarrow (\forall a \in A) [(i, X)\vdash a \Leftrightarrow (j, X) 
\vdash a]$

\item\label{lem-witeqprop-2}
$i \sim j \,[X] \wedge U \subseteq X \Longrightarrow i \sim j \,[U]$

\item\label{lem-witeqprop-3}
$i \sim j \,[X] \wedge (i, X) \vdash U \Longrightarrow i \sim j \,[U]$

\item\label{lem-witeqprop-4}
$i \sim j \,[X] \wedge (j, X) \vdash (k, U) \Longrightarrow i \sim k \,[U]$.
\end{enumerate}
\end{lemma}
\begin{proof}
(\ref{lem-witeqprop-0}) is obvious.

For the remaining statements let $i \sim j \,[X]$. Then there are $n \in \omega$ and $k_1$, \ldots, $k_n$, $a_0$, \ldots, $a_n \in A$ such that $a_0 = i$, $a_n = j$, $X \in \Con(a_{\nu -1}) \cap \Con(a_\nu)$ and $\{ a_{\nu - 1} \}, \{ a_\nu \} \in \Con(k_\nu)$, for $\nu = 1, \ldots, n$.

(\ref{lem-witeqprop-1}) With Conditions~\ref{dn-infsys}(\ref{dn-infsys-8}, \ref{dn-infsys-9}) it follows that 
\[(a_{\nu -1},X) \vdash a \Longleftrightarrow (k_\nu, X) \vdash a \Longleftrightarrow (a_\nu, X) \vdash a,
\]
and consequently that $(i, X) \vdash a$, exactly if $(j, X) \vdash a$, for all $a \in A$.

(\ref{lem-witeqprop-3}) Let $1 \le \nu \le n$. Then we have that $a_\nu \sim i \,[X]$ and hence, by Statement~(\ref{lem-witeqprop-1}), that $(a_\nu, X) \vdash U$. Thus $U \in \Con(a_\nu)$. This shows that for all $\nu = 1, \ldots, n$, $U \in \Con(a_{\nu -1})\cap \Con(a_\nu)$ Since also $\{ a_{\nu -1}\}, \{ a_\nu\} \in \Con(k_\nu)$, we obtain that $i \sim j \,[U]$.

(\ref{lem-witeqprop-2}) follows similarly.

(\ref{lem-witeqprop-4}) We have that $(j, X) \vdash U$ and hence that $i \sim j \,[U]$, by Rule~(\ref{lem-witeqprop-2}). Thus there are $m \in \omega$ and $r_1, \ldots, r_m, b_0, \ldots, b_m$ with $b_0 = i$, $b_m = j$, $U \in \Con(b_{\mu -1}) \cap \Con(b_\mu)$ and $\{ b_{\mu -1} \}, \{ b_\mu \} \in \Con(r_\mu)$, with $\mu = 1, \ldots, m$. By assumption $(j, X) \vdash k$. Hence, $\{ k \} \in \Con(j)$. As $\{ j \} \in \Con(r_m)$, we obtain that $\{ k \} \in \Con(r_m)$. Now, set $b_m = k$. Then it follows that $i \sim k \,[U]$.
\end{proof}

\begin{lemma} \label{lem-witeqam}
Let $(A', \Con', \vdash', \Delta')$ be a further information system with witnesses and $H$ be an approximable mapping between $A$ and $A'$. Moreover, let  $X \fsubset A$, $i, j \in A$, and $b \in A'$. Then
\[
i \sim j \,[X] \wedge (i, X) H b  \Longrightarrow (j, X) H b.
\]
\end{lemma}
\begin{proof}
Let $i \sim j \,[X]$ and $(i, X) H b$. By Lemma~\ref{pn-amint} there is some $(c, U) \in \Con$ such that $(i, X) \vdash (c, U)$ and $(c, U) H b$. Because of Lemma~\ref{lem-witeqprop}(\ref{lem-witeqprop-1}) it follows that also $(j, X) \vdash (c, U)$. With  Lemma~\ref{lem-amstrong3} we therefore obtain that $(j, X) H b$.
\end{proof}

For $\mathfrak{X} \fsubset \Con$ and $(a, S) \in \Con$ define
\[
\Cl((a, S), \mathfrak{X}) = \set{(i, X) \in \mathfrak{X}}{(\exists j \in A) i \sim j \,[X] \wedge (a, S) \vdash (j, X)}.
\]
Then $\Cl((a, S), \mathfrak{X})$ is finite. Let in addition
\begin{equation*}
\ucl((a, S), \mathfrak{X}) = \\
 \bigcup \set{\{ j \} \cup X}{(\exists i \in A) (i, X) \in \mathfrak{X} \wedge j \sim i \,[X] \wedge (a, S) \vdash (j, X)}.
\end{equation*}

\begin{lemma}\label{lem-clprop}
Let $(a, S), (c, T) \in \Con$ such that $(a, S) \vdash T$ and $a \sim c \,[T]$. Then 
\[
\Cl((c, T), \mathfrak{X}) \subseteq \Cl((a, S), \mathfrak{X}).
\]
\end{lemma}
\begin{proof}
Let $(i, X) \in \mathfrak{X}$ and $i \sim j \,[X]$ so that $(c, T) \vdash (j, X)$. Then also $(a, T) \vdash (j, X)$ and hence, $(a, S) \vdash (j, X)$. 
\end{proof}

Next, set
\[
W(\mathfrak{X}) = \set{a \in A}{\bigcup \pr_2(\mathfrak{X}) \in \Con(a) \wedge (\forall (i, X) \in \mathfrak{X}) a \sim i \,[X]},
\]
and let $(A', \Con', \vdash', \Delta')$ be a further information system with witnesses.

\begin{definition}
For $\mathfrak{V} \fsubset \Con \times A'$, $\mathfrak{U} \subseteq \mathfrak{V}$ and $a \in W(\pr_1(\mathfrak{V}))$,  $(a, \mathfrak{U})$ is \emph{$\mathfrak{V}$-maximal} if the following condition holds, for all $((i, X), j) \in \mathfrak{V}$:
\[
X \subseteq \bigcup \pr_{2, 1}(\mathfrak{U}) \wedge a \sim i \,[X] \Longrightarrow ((i, X), j) \in \mathfrak{U}.
\]
\end{definition}

\begin{lemma}\label{lem-max}
Let $(a, \mathfrak{U})$ be $\mathfrak{V}$-maximal and $c \in W(\pr_1(\mathfrak{U}))$ with $c \sim a \,[\bigcup \pr_{2, 1}(\mathfrak{U})]$. Then $(c, \mathfrak{U})$ is $\mathfrak{V}$-maximal as well.
\end{lemma}
\begin{proof}
Let $((i, X), j) \in \mathfrak{V}$ with $X \subseteq \bigcup \pr_{2, 1}(\mathfrak{U})$ with $c \sim i \,[X]$. Then $a \sim c \,[X]$, by Lemma~\ref{lem-witeqprop}(\ref{lem-witeqprop-2}), and hence $a \sim i \,[X]$. Therefore, $((i, X), j) \in \mathfrak{U}$, as $(a, \mathfrak{U})$ is $\mathfrak{V}$-maximal.
\end{proof}

 As in the case of Scott's information systems, the tokens of the function space $A \rightarrow A'$ will be finite subsets  of $\Con \times A'$. Let $\mathfrak{V} = \{ ((i_1, X_1), c_1), \ldots, ((i_n, X_n), c_n)\}$ be such a set and $(a, S) \in \Con$. Assume that $J \subseteq \{ 1, \ldots, n \}$ with $(a, S) \vdash (i_\nu, X_\nu)$, exactly if $\nu \in J$. Then we need a witness for the set $\set{c_\nu}{\nu \in J}$. Moreover, if $J_\mu \subseteq \{ 1, \ldots, n \}$, for $ \mu = 1, \ldots, m$, with $(a, S) \vdash (i_\nu, X_\nu)$, for $\nu \in J_\mu$ and $1 \le \mu \le m$, and $t_\mu$ is a witness for $\set{c_\nu}{\nu \in J_\mu}$, then we also need a witness for $\set{t_\mu}{1 \le \mu \le m}$.

\begin{definition}\label{dn-ass}
$\mathfrak{W} \subseteq \Con \times A'$ is an \emph{associate} of $\mathfrak{V}$ if for each $J \subseteq \{ 1, \ldots, n \}$, $R(\set{(i_\nu, X_\nu)}{\nu \in J})$ is a system of representatives of $W(\{\,(i_\nu, X_\nu) \mid \nu \in J\,\})$ with respect to $\sim [\bigcup\nolimits_{\nu \in J} X_\nu]$ so that $R(\emptyset) = \{ \Delta \}$ and $R(\{ (i_\nu, X_\nu \}) = \{ i_\nu \}$,  and there exists an increasing chain $(\mathfrak{W}^{(\kappa)})_{\kappa \in \omega}$ of subsets of $\Con \times A'$ such that 
\begin{enumerate}

\item\label{dn-ass-1} $\mathfrak{W} = \bigcup_{\kappa \in \omega} \mathfrak{W}^{(\kappa)}$,

\item\label{dn-ass-2}
$\mathfrak{W}^{(0)} = \begin{cases}
					\{ ((\Delta, \emptyset), \Delta') \} & \text{if $(\Delta, \emptyset)$ is $\mathfrak{V}$-maximal,} \\
					\emptyset & \text{otherwise,}
			      \end{cases}	
$

\item\label{dn-ass-3} and for all $\kappa \ge 1$, 
\[
\mathfrak{W}^{(\kappa)} = \mathfrak{W}^{(\kappa-1)} \cup \widehat{\mathfrak{W}}^{(\kappa)},
\]
 where $\widehat{\mathfrak{W}}^{(\kappa)}$ satisfies the following conditions:
\begin{enumerate}
\item\label{dn-ass-3-1}
For all $J \subseteq \{ 1, \ldots, n \}$ with $\| J \| = \kappa$ and all $e \in R(\{\, (i_\nu, X_\nu) \mid \nu \in J \,\})$  for which $(e, \{\, ((i_\nu, X_\nu), c_\nu) \mid \nu \in J \,\})$ is $\mathfrak{V}$-maximal,  there is exactly one $j \in A'$ so that
\begin{enumerate}

\item \label{dn-ass-3-1-1}
$ \{\, t \in A' \mid (\exists a \in A) (\exists K \subseteq \{ 1, \ldots, n \}) \bigcup_{\rho \in K} X_\rho \subseteq \bigcup_{\nu \in J} X_\nu \wedge a \sim e \,[\bigcup_{\rho \in K} X_\rho] \wedge \mbox{}\\
\mbox{}\hspace{1cm}   [(\| K \| = 1 \wedge ((a, \bigcup_{\rho \in K} X_\rho), t) \in \mathfrak{V}) \vee 
 (a \in R(\{\, (i_\rho, X_\rho) \mid \rho \in K \,\}) \wedge \mbox{} \\
\mbox{}\hspace{1cm}  ((a, \bigcup_{\rho \in K} X_\rho), t) \in \mathfrak{W}^{(\kappa-1)})] \,\} \in \Con'(j),
$

\item\label{dn-ass-3-1-2} $((a, \bigcup_{\nu \in J} X_\nu), j) \in \widehat{\mathfrak{W}}^{(\kappa)}$.

\end{enumerate}

\item\label{dn-ass-3-2}
For all $((a, T), b) \in \widehat{\mathfrak{W}}^{(\kappa)}$ there are  $J \subseteq \{ 1, \ldots n \}$ and $a \in R(\{\, (i_\nu, X_\nu) \mid \nu \in J \,\})$ so that $\| J \| = \kappa$, $(a, \{\, ((i_\nu, X_\nu), c_\nu) \mid \nu \in J \,\})$ is $\mathfrak{V}$-maximal,  $T = \bigcup_{\nu \in J} X_\nu$, and $b$ satisfies Condition~(\ref{dn-ass-3-1-1}).
\end{enumerate}			      
\end{enumerate}
\end{definition}

Let $J \subseteq \{ 1, \ldots, n \}$ and $a \in W(\set{(i_\nu, X_\nu)}{\nu \in J})$ so that $(a, \{\, ((i_\nu, X_\nu), c_\nu) \mid \nu \in J \,\})$ is $\mathfrak{V}$-maximal. If $J = \emptyset$, then both sets $\set{(i_\nu, X_\nu)}{\nu \in J}$ and $\set{c_\nu}{\nu \in J}$ are empty as well. Moreover, $\Delta \in W(\emptyset)$ and $\emptyset \in \Con'(\Delta')$. This explains Condition~\ref{dn-ass}(\ref{dn-ass-2}). Note that $(\Delta, \emptyset)$ is not $\mathfrak{V}$-maximal, if for some $1 \le \kappa \le n$, $X_\kappa = \emptyset$.

For the larger cardinalities it follows from Condition~\ref{dn-ass}(\ref{dn-ass-3-1-1}) that the second components of pairs in $\mathfrak{W}$ are the consistency witnesses for the sets of second components in $\mathfrak{V}$ we were looking for above.

Note in Condition~\ref{dn-ass}(\ref{dn-ass-3-1-1}) that if $J \subseteq \{ 1, \ldots, n \}$ with $\|J\| \ge 1$ and 
\[
((a_J, \bigcup_{\nu \in J} X_\nu), t_J) \in \mathfrak{W},
\]
for some $a_j \in A$ and $t_J \in A'$, then $(a_J, \{\, ((i_\nu, X_\nu), c_\nu) \mid \nu \in J \,\})$ must be $\mathfrak{V}$-maximal, by Condition~\ref{dn-ass}(\ref{dn-ass-3-2}). Thus, if for some $K \subseteq \{ 1, \ldots, n \}$ one has that $\bigcup_{\kappa \in K} X_\kappa \subseteq \bigcup_{\nu \in J} X_\nu$, then $K \subseteq J$.

Let $\ac(\mathfrak{V})$ be the set of all associates of $\mathfrak{V}$. We write $\fcbael{W}{V}$ with $\mathfrak{V} \fsubset \Con \times A'$ to mean that $\mathfrak{W} \in \ac(\mathfrak{V})$. Similarly, for sets $\mathcal{V} \fsubset \mathcal{P}_f(\Con \times A')$ and $\mathcal{W} \fsubset \mathcal{P}(\Con\times A')$ of equal cardinality, say $\mathcal{V} = \{ \mathfrak{V}_1, \ldots, \mathfrak{V}_m \}$ and $\mathcal{W} = \{ \mathfrak{W}_1, \ldots, \mathfrak{W}_m \}$, such that $\mathfrak{W}_\nu \in \ac(\mathfrak{V}_\nu)$, for $\nu = 1, \ldots, m$, we write $\fcbaset{W}{V}$ for the set $\set{\pair{\mathfrak{W}_\nu | \mathfrak{V}_\nu}}{1 \le \nu \le m}$. In this sense we also say that $\fcbaset{W}{V} \fsubset \mathcal{P}((\Con \times A')^2)$.

\begin{lemma}\label{lem-singext}
For $(c, U) \in \Con$ and $Z \fsubset A'$ let $\mathfrak{D} = \set{ ((c, U), d) }{d \in Z}$. Moreover, let $b \in A'$ with $Z \cup \{ \Delta' \} \in \Con'(b)$, if  $U$ is not empty, and $Z \in \Con'(b)$, otherwise, and define
\[
\mathfrak{E} = \begin{cases}
           \{ ((\Delta, \emptyset), \Delta'), ((c, U), b) \} & \text{if $U \not = \emptyset$,} \\
           \{ ((c, U), b) \} & \text{otherwise.}
\end{cases}           
\] 
Then $\mathfrak{E} \in \ac(\mathfrak{D})$.
\end{lemma}

The proof is a straightforward exercise. Note that for any nonempty subset $V \subseteq Z$, $(c, \{\, ((c, U), d) \mid d \in V \,\})$ is $\mathfrak{D}$-maximal, exactly if $V = Z$. If $U$ is empty, this is even true for any subset $V$.

For $(a, S) \in \Con$,
\[
\sp((a, S), \mathfrak{V}) = \set{((i, X), j) \in \mathfrak{V}}{a \sim i \,[X] \wedge X \subseteq S},
\]
is the \emph{$\mathfrak{V}$-spectrum} of $(a, S)$. Set $\ds((a, S), \mathfrak{V}) = \pr_1(\sp((a, S), \mathfrak{V})$ and $\rs((a, S), \mathfrak{V}) = \pr_2(\sp((a, S), \mathfrak{V})$. 

Assume again that $\mathfrak{V} = \{ ((i_1, X_1), c_1), \ldots, ((i_n, X_n), c_n)\}$ and let $\mathfrak{W} \in \ac(\mathfrak{V})$.  If $((c, T), d) \in \mathfrak{W}$, if follows from Condition~\ref{dn-ass}(\ref{dn-ass-3-2}) that there is exactly one subset $J \subseteq \{ 1, \ldots, n \}$ such that $(c, \{\, ((i_\nu, X_\nu), c_\nu) \mid \nu \in J\,\})$ is $\mathfrak{V}$-maximal, $c \in W(\{\, (i_\nu, X_\nu) \mid \nu \in J \,\})$, and $T = \bigcup_{\nu \in J} X_\nu$. Then $\ds((c, T), \mathfrak{V}) = \{\, (i_\nu, X_\nu) \mid \nu \in J \,\}$ and $\rs((c, T), \mathfrak{V}) = \{\, c_\nu \mid \nu \in J \,\}$.

Define
\[
\ap((a, S), \mathfrak{V}) = \set{c \in A'}{(\exists (i, X) \in \Cl((a, S), \pr_1(\mathfrak{V}))) ((i, X), c) \in \mathfrak{V}}.
\]

As a consequence of Lemma~\ref{lem-clprop} we have for $(c, T) \in \Con$ with $(a, S)\vdash T$ and $a \sim c \,[T]$ that $\ap((c, T), \mathfrak{V}) \subseteq \ap((a, S), \mathfrak{V})$. 

Let  $M \subseteq \{ 1, \ldots, n \}$ such that 
\[
\Cl((a, S), \pr_1(\mathfrak{V})) = \set{(i_\kappa, X_\kappa)}{\kappa \in M}.
\]
Then $(a, \{\, ((i_\kappa, X_\kappa), c_\kappa) \mid \kappa \in M \,\})$ is $\mathfrak{V}$-maximal.
It follows that 
\[
\bigcup \pr_2(\Cl((a, S), \pr_1(\mathfrak{V}))) = \bigcup\nolimits_{\kappa \in M} X_\kappa \text{ and } \ap((a, S), \mathfrak{V}) = \set{c_\kappa}{\kappa \in M}.
\]
Moreover, there exists exactly one pair $(e,j)$ with $e \in W(\set{(i_\kappa, X_\kappa)}{\kappa \in M})$ and $j \in A'$ so that
\[
a \sim e \,[\bigcup_{\kappa \in M} X_\kappa], \,\,  \set{c_\kappa}{\kappa \in M} \in \Con'(j), \,\, \text{and} \,\,
((e, \bigcup_{\kappa \in M} X_\kappa), j) \in \mathfrak{W}.
\]
Set 
\[
\mathfrak{W}(a, S) = ((e, \bigcup_{\kappa \in M} X_\kappa), j).
\]
Then we also write
\begin{gather*}
\mathfrak{W}_1(a, S) = (e, \bigcup_{\kappa \in M} X_\kappa), \,\, \mathfrak{W}_{1,1}(a, S) = e, \,\, \mathfrak{W}_{1,2}(a, S) =  \bigcup_{\kappa \in M} X_\kappa, 
  \intertext{and} \mathfrak{W}_2(a, S) = j.
\end{gather*}

The next lemma is a consequence of Lemma~\ref{lem-clprop} and Condition~\ref{dn-ass}(\ref{dn-ass-3-1-1})

\begin{lemma}\label{lem-jextprop}
Let $(a, S), (c, T) \in \Con$. Then the following statements hold:
\begin{enumerate}
\item \label{lem-jextprop-1}
If $(a, S) \vdash T$ and $a \sim c \,[T]$. Then $\{ \mathfrak{W}_2(c, T) \} \in \Con'(\mathfrak{W}_2(a, S))$.

\item \label{lem-jextprop-2}
If $T \subseteq S$ and $a \sim c \,[T]$. Then $\{ \mathfrak{W}_2(c, T) \} \in \Con'(\mathfrak{W}_2(a, S))$.

\item \label{lem-jextprop-3}
If $(a, S) \vdash (c, T)$ and $(c, T) \vdash \ucl((a, S), \pr_1(\mathfrak{V})))$. Then 
\[
\Cl((a, S), \pr_1(\mathfrak{V})) = \Cl((c, T), \pr_1(\mathfrak{V}))
\quad \text{and} \quad
\mathfrak{W}_2(c, T) = \mathfrak{W}_2(a, S).
\]

\end{enumerate}
\end{lemma}

\begin{definition}\label{dn-fctsp}
Let $(A, \Con, \vdash, \Delta)$ and $(A', \Con', \vdash', \Delta')$ be information systems with witnesses. Define
\begin{enumerate}
\item\label{dn-fctsp-1} 
$A \rightarrow A' = \set{\fcbael{W}{V} \in \mathcal{P}(\Con \times A') \times \mathcal{P}_f(\Con \times A')}{\mathfrak{W} \in \ac(\mathfrak{V})} $

\item \label{dn-fctsp-2} 
For $\fcbael{W}{V} \in A \rightarrow A'$ and $\fcbaset{W}{V} \fsubset A \rightarrow A'$, 
$
\fcbaset{W}{V} \in \Con_\rightarrow(\fcbael{W}{V})
$ 
if for all $(a, S) \in \Con$,
\begin{enumerate}

\item\label{dn-fctsp-2-1} $
(\forall  \pair{\mathfrak{G} | \mathfrak{F}} \in \pair{\mathcal{W} | \mathcal{V}}) (\exists k \in A')\,
\mathfrak{G}_2(a, S) \sim k \,[\ap((a, S), \mathfrak{F})] \wedge \{ k \} \in \Con'(\mathfrak{W}_2(a, S)) 
$

\item\label{dn-fctsp-2-2}  $
(\forall r \in A') (r \sim \mathfrak{W}_2(a, S) \, [\ap((a, S), \mathfrak{V})] \Rightarrow 
r \sim \mathfrak{W}_2(a, S) \, [\ap((a, S), \bigcup \mathcal{V})] )
$.
\end{enumerate}

\item \label{dn-fctsp-3} 
For $\pair{\mathfrak{W} | \mathfrak{V}}, \pair{\mathfrak{B} | \mathfrak{A}} \in A \rightarrow A'$ and $\pair{\mathcal{W} | \mathcal{V}} \in \Con_\rightarrow(\pair{\mathfrak{W} | \mathfrak{V}})$,
$
(\pair{\mathfrak{W} | \mathfrak{V}}, \pair{\mathcal{W} | \mathcal{V}}) \vdash_\rightarrow \pair{\mathfrak{B} | \mathfrak{A}} 
$ if
\begin{enumerate}
\item \label{dn-fctsp-3-1}
$
(\forall ((e, Y), f) \in \mathfrak{A}) (\mathfrak{W}_2(e, Y), \ap((e, Y), \bigcup \mathcal{V})) \vdash' f 
$
\item \label{dn-fctsp-3-2}
$  
(\forall ((a, Z), b) \in \mathfrak{B}) (\exists k \in A')\,  b \sim k \,[\rs((a, Z), \mathfrak{A})] \wedge \mbox{}   \\
\mbox{}\hspace{.7em}(\mathfrak{W}_2(a, Z), \bigcup \{\,\ap((e, Y), \bigcup \mathcal{V}) \mid (e, Y) \in \ds((a, Z), \mathfrak{A}) \,\})  \vdash' 
(k, \rs((a, Z), \mathfrak{A}))
$.
\end{enumerate}

\item \label{dn-fctsp-4}  $\Delta_\rightarrow = \pair{\{ ((\Delta, \emptyset), \Delta') \} | \emptyset }$.

\end{enumerate}
\end{definition}

In Condition~\ref{dn-fctsp}(\ref{dn-fctsp-3-1}) one needs that 
$
\ap((e, Y), \bigcup \mathcal{V}) \in \Con'(\mathfrak{W}_2(e, Y)).
$
This is an easy consequence of Condition~\ref{dn-fctsp}(\ref{dn-fctsp-2-2}): choose $r = \mathfrak{W}_2(e, Y)$ and note that by definition of $\mathfrak{W}_2(e, Y)$,
$
\ap((e, Y), \mathfrak{V}) \in \Con'(\mathfrak{W}_2(e, Y)).
$
 
For Condition~\ref{dn-fctsp}(\ref{dn-fctsp-3-2}) observe that 
\[
\bigcup \{\,\ap((e, Y), \bigcup \mathcal{V}) \mid (e, Y) \in \ds((a, Z), \mathfrak{A}) \,\} \subseteq \ap((a, Z), \bigcup \mathcal{V}).
\]
Because of Condition~\ref{dn-infsys}(\ref{dn-infsys-2}) we therefore have that  
\[
\bigcup \{\,\ap((e, Y), \bigcup \mathcal{V}) \mid (e, Y) \in \ds((a, Z), \mathfrak{A}) \,\} \in \Con'(\mathfrak{W}_2(a, Z)).
\]

\begin{lemma}\label{lem-extvdash}
Let $\pair{\mathcal{W} | \mathcal{V}} \in \Con_\rightarrow(\pair{\mathfrak{W} | \mathfrak{V}})$ and $\pair{\mathcal{G} | \mathcal{F}} \fsubset A \rightarrow A'$ such that Condition~\ref{dn-fctsp}(\ref{dn-fctsp-3-1}) holds for all $\fcbael{G}{F} \in \fcbaset{G}{F}$. Then for all $(a, S) \in \Con$,
\begin{equation*}
(\mathfrak{W}_2(a, S), \bigcup \{\,\ap((c, T), \bigcup \mathcal{V}) \mid (c, T) \in \Cl((a, S), \pr_1(\bigcup\mathcal{F})) \,\}) \vdash' 
\ap((a, S), \bigcup\mathcal{F}).
\end{equation*}
\end{lemma}
\begin{proof}
Let $b \in \ap((a, S), \bigcup\mathcal{F})$. Then there are $(j, U) \in \Con$ and $j' \in A$ such that $((j, U), b) \in \bigcup\mathcal{F}$, $j \sim j' \,[U]$, and $(a, S) \vdash (j', U)$. As a consequence of our assumption, we have that $(\mathfrak{W}_2(j, U), \ap((j, U), \bigcup \mathcal{V})) \vdash' b$. Then $\mathfrak{W}_2(j, U) = \mathfrak{W}_2(j', U)$ and $\{ \mathfrak{W}_2(j, U) \} \in \Con'(\mathfrak{W}_2(a, S))$. It follows that also $(\mathfrak{W}_2(a, S), \ap((j, U), \bigcup \mathcal{V})) \vdash' b$. Since moreover $\ap((j, U), \bigcup \mathcal{V}) \subseteq \bigcup \{\,\ap((c, T), \bigcup \mathcal{V}) \mid (c, T) \in \Cl((a, S), \pr_1(\bigcup\mathcal{F})) \,\}$, we finally obtain that $(\mathfrak{W}_2(a, S), \bigcup \{\,\ap((c, T), \bigcup \mathcal{V}) \mid (c, T) \in \Cl((a, S), \pr_1(\bigcup\mathcal{F})) \,\}) \vdash' b$.
\end{proof}

\begin{lemma}\label{lem-extwit}
Let $\pair{\mathcal{W} | \mathcal{V}} \in \Con_\rightarrow(\pair{\mathfrak{W} | \mathfrak{V}})$ and $\pair{\mathfrak{B} | \mathfrak{A}} \in A \rightarrow A'$ so that Condition~\ref{dn-fctsp}(\ref{dn-fctsp-3-2}) holds. Then for all $(a, S) \in \Con$, $\mathfrak{B}_2(a, S) \sim \mathfrak{W}_2(a, S) \,[\ap((a, S), \mathfrak{A})]$.
\end{lemma}
\begin{proof}
By definition of $\mathfrak{B}(a, S)$,
\[
\ap((a, S), \mathfrak{A}) = \rs(\mathfrak{B}_1(a, S), \mathfrak{A}) \,\, \text{and} \,\, (\mathfrak{B}_1(a, S), \mathfrak{B}_2(a, S)) \in \mathfrak{B}.
\]
So, Condition~\ref{dn-fctsp}(\ref{dn-fctsp-3-2}) implies that $\mathfrak{B}_2(a, S) \sim \mathfrak{W}_2(a, S) \,[\ap((a, S), \mathfrak{A})]$.
\end{proof}

\begin{proposition}\label{pn-fctinfsys}
Let $(A, \Con, \vdash, \Delta)$ and $(A', \Con', \vdash', \Delta')$ be information systems with witnesses. Then $(A \rightarrow A', \Con_\rightarrow, \vdash_\rightarrow, \Delta_\rightarrow)$ is an information system with witnesses as well.
\end{proposition}

In the subsequent lemmas we verify the conditions in Definition~\ref{dn-infsys}.

\begin{lemma}\label{lem-fc1}
$(A \rightarrow A', \Con_\rightarrow, \vdash_\rightarrow, \Delta_\rightarrow)$ satisfies Condition~\ref{dn-infsys}(\ref{dn-infsys-1}).
\end{lemma}
\begin{proof}
We have to show that $\{ \pair{\mathfrak{W} | \mathfrak{V}} \} \in \Con_\rightarrow(\pair{\mathfrak{W} | \mathfrak{V}})$. Let $(a, S) \in \Con$. Condition~\ref{dn-fctsp}(\ref{dn-fctsp-2-2}) is vacuously true and for Requirement~\ref{dn-fctsp}(\ref{dn-fctsp-2-1}) choose $k = \mathfrak{W}_2(a, S)$.
\end{proof}

\begin{lemma}\label{lem-fc2}
$(A \rightarrow A', \Con_\rightarrow, \vdash_\rightarrow, \Delta_\rightarrow)$ satisfies Condition~\ref{dn-infsys}(\ref{dn-infsys-2}).
\end{lemma}
\begin{proof}
Let $\pair{\mathcal{W} | \mathcal{V}} \in \Con_\rightarrow(\pair{\mathfrak{W} | \mathfrak{V}})$ and $\pair{\mathcal{G} | \mathcal{F}} \subseteq \pair{\mathcal{W} | \mathcal{V}}$. We need to verify that $\pair{\mathcal{G} | \mathcal{F}} \in \Con_\rightarrow(\pair{\mathfrak{W} | \mathfrak{V}})$: Condition~\ref{dn-fctsp}(\ref{dn-fctsp-2-1}) holds trivially. For Requirement~\ref{dn-fctsp}(\ref{dn-fctsp-2-2}) apply \linebreak Lemma~\ref{lem-witeqprop}(\ref{lem-witeqprop-2}).
\end{proof}

\begin{lemma}\label{lem-fc3}
$(A \rightarrow A', \Con_\rightarrow, \vdash_\rightarrow, \Delta_\rightarrow)$ satisfies Condition~\ref{dn-infsys}(\ref{dn-infsys-3}).
\end{lemma}
\begin{proof}
We have to prove that $(\pair{\mathfrak{W} | \mathfrak{V}}, \emptyset) \vdash_\rightarrow \Delta_\rightarrow$. Condition~\ref{dn-fctsp}(\ref{dn-fctsp-3-1}) is vacuously satisfied and for Condition~\ref{dn-fctsp}(\ref{dn-fctsp-3-2}) choose $k = \Delta'$. 
\end{proof}

\begin{lemma}\label{lem-fc4}
$(A \rightarrow A', \Con_\rightarrow, \vdash_\rightarrow, \Delta_\rightarrow)$ satisfies Condition~\ref{dn-infsys}(\ref{dn-infsys-4}).
\end{lemma}
\begin{proof}
Let $\pair{\mathcal{W} | \mathcal{V}} \in \Con_\rightarrow(\pair{\mathfrak{W} | \mathfrak{V}})$ and $\pair{\mathcal{G} | \mathcal{F}} \fsubset A \rightarrow A'$ with $(\pair{\mathfrak{W} | \mathfrak{V}}, \pair{\mathcal{W} | \mathcal{V}}) \vdash_\rightarrow \pair{\mathcal{G} | \mathcal{F}}$. We must verify  that $\pair{\mathcal{G} | \mathcal{F}} \in \Con_\rightarrow(\pair{\mathfrak{W} | \mathfrak{V}})$.

Let to this end $(a, S) \in \Con$ and $\pair{\mathfrak{B} | \mathfrak{A}} \in \pair{\mathcal{G} | \mathcal{F}}$. As $(\pair{\mathfrak{W} | \mathfrak{V}}, \pair{\mathcal{W} | \mathcal{V}}) \vdash_\rightarrow \pair{\mathcal{G} | \mathcal{F}}$, there is some $k \in A'$ so that $\mathfrak{B}_2(a, S) \sim k \,[\ap((a, S), \mathfrak{A})]$ and
\begin{equation*}
(\mathfrak{W}_2(a, S), \bigcup \{\, \ap((c, T), \bigcup \mathcal{V}) | (c, T) \in \Cl((a, S), \pr_1(\mathfrak{A})) \,\} \vdash'   (k, \ap((a, S), \mathfrak{A})).
\end{equation*}
Thus, $\{ k \} \in \Con'(\mathfrak{W}_2(a, S))$, which proves Condition~\ref{dn-fctsp}(\ref{dn-fctsp-2-1}).

For Condition~\ref{dn-fctsp}(\ref{dn-fctsp-2-2}) let $r \in A'$ with $r \sim \mathfrak{W}_2(a, S) \,[\ap((a, S), \mathfrak{V})]$. Then 
\begin{equation}\label{eq-cond5}
r \sim \mathfrak{W}_2(a, S) \,[\ap((a, S), \bigcup \mathcal{V})],
\end{equation}
 as $\pair{\mathcal{W} | \mathcal{V}} \in \Con_\rightarrow(\pair{\mathfrak{W} | \mathfrak{V}})$.
With Lemma~\ref{lem-extvdash} we moreover have that 
\[
(\mathfrak{W}_2(a, S), \ap((a, S), \bigcup \mathcal{V})) \vdash' \ap((a, S), \bigcup \mathcal{F}).
\] 
By Condition~\ref{dn-infsys}(\ref{dn-infsys-10}) there is thus some $k \in A'$ with 
\[
(\mathfrak{W}_2(a, S), \ap((a, S), \bigcup \mathcal{V})) \vdash' (k, \ap((a, S), \bigcup \mathcal{F})).
\] 
Because of Statement~(\ref{eq-cond5}) we obtain with Lemma~\ref{lem-witeqprop}(\ref{lem-witeqprop-4}) that 
\[
k \sim \mathfrak{W}_2(a, S) \,[\ap((a, S), \bigcup \mathcal{F})],
\]
 and similarly that $k \sim r \,[\ap((a, S), \bigcup \mathcal{F})]$. Therefore, 
\[
r \sim \mathfrak{W}_2(a, S) \,[\ap((a, S), \bigcup \mathcal{F})]. \qedhere
\]
\end{proof}

Requirement~\ref{dn-infsys}(\ref{dn-infsys-5}) is obvious, because for $\pair{\mathcal{W} | \mathcal{V}} \subseteq \pair{\mathcal{G} | \mathcal{F}}$ we have that $\ap((a, s), \bigcup \mathcal{V}) \subseteq \ap((a, s), \bigcup \mathcal{F})$.

\begin{lemma}\label{lem-fc6}
$(A \rightarrow A', \Con_\rightarrow, \vdash_\rightarrow, \Delta_\rightarrow)$ satisfies Condition~\ref{dn-infsys}(\ref{dn-infsys-6}).
\end{lemma}
\begin{proof}
Let $\pair{\mathcal{W} | \mathcal{V}} \in \Con_\rightarrow(\pair{\mathfrak{W} | \mathfrak{V}})$, $\pair{\mathcal{G} | \mathcal{F}} \fsubset A \rightarrow A'$, and $\pair{\mathfrak{B} | \mathfrak{A}} \in A \rightarrow A'$ with $(\pair{\mathfrak{W} | \mathfrak{V}}, \pair{\mathcal{W} | \mathcal{V}}) \vdash_\rightarrow \pair{\mathcal{G} | \mathcal{F}}$ and $(\pair{\mathfrak{W} | \mathfrak{V}}, \pair{\mathcal{G} | \mathcal{F}}) \vdash_\rightarrow \pair{\mathfrak{B} | \mathfrak{A}}$. We need to show that 
\[
(\pair{\mathfrak{W} | \mathfrak{V}}, \pair{\mathcal{W} | \mathcal{V}}) \vdash_\rightarrow \pair{\mathfrak{B} | \mathfrak{A}}.
\]

The first condition to be verified is a consequence of Lemma~\ref{lem-extvdash}. For the second condition let $((a, Z), b) \in \mathfrak{B}$. By our assumption there is some $k \in A'$ so that 
\[
k \sim b \,[\rs((a, Z), \mathfrak{A})]
\]
and 
\begin{equation*}
(\mathfrak{W}_2(a, Z), \bigcup \{\, \ap((e, Y), \bigcup \mathcal{F}) \mid (e, Y) \in \ds((a, Z), \mathfrak{A}) \,\}) \vdash' 
 (k, \rs((a, Z), \mathfrak{A})).
\end{equation*}
 Moreover, with Lemma~\ref{lem-extvdash}, we obtain that 
\begin{multline*}
(\mathfrak{W}_2(a, Z), \bigcup \{\, \ap((c, T), \bigcup \mathcal{V}) \mid \\
(\exists (e, Y) \in \ds((a, Z), \mathfrak{A})) 
(c, T) \in \Cl((e, Y), \pr_1(\bigcup \mathcal{F})) \,\} ) \vdash' \\
 \bigcup \{\, \ap((e, Y), \bigcup \mathcal{F}) \mid (e, Y) \in \ds((a, Z), \mathfrak{A}) \,\}.
\end{multline*}
Since 
\begin{multline*}
\bigcup \{\, \ap((c, T), \bigcup \mathcal{V}) \mid \\
 (\exists (e, Y) \in \ds((a, Z), \mathfrak{A})) (c, T) \in \Cl((e, Y), \pr_1(\bigcup \mathcal{F})) \,\} \subseteq \\
 \bigcup \{\, \ap((e, Y), \bigcup \mathcal{V}) \mid (e, Y) \in \ds((a, Z), \mathfrak{A}) \,\},
\end{multline*}
it follows that 
\begin{equation*}
(\mathfrak{W}_2(a, Z), \bigcup \{\, \ap((e, Y), \bigcup \mathcal{V}) \mid (e, Y) \in \ds((a, Z), \mathfrak{A}) \,\} ) \vdash' \\
 (k, \rs((a, Z), \mathfrak{A})),
\end{equation*}
as was to be demonstrated.
\end{proof}

\begin{lemma}\label{lem-fc10}
$(A \rightarrow A', \Con_\rightarrow, \vdash_\rightarrow, \Delta_\rightarrow)$ satisfies Condition~\ref{dn-infsys}(\ref{dn-infsys-10}).
\end{lemma}
\begin{proof}
Let $\pair{\mathfrak{W} | \mathfrak{V}}, \pair{\mathfrak{G} | \mathfrak{F}}  \in A \rightarrow A'$ and $\pair{\mathcal{W} | \mathcal{V}} \in \Con_\rightarrow(\pair{\mathfrak{W} | \mathfrak{V}})$ with 
\[
(\pair{\mathfrak{W} | \mathfrak{V}}, \pair{\mathcal{W} | \mathcal{V}}) \vdash_\rightarrow \pair{\mathfrak{G} | \mathfrak{F}}.
\]
We have to show that there exists $\pair{\mathcal{E} | \mathcal{D}} \fsubset A \rightarrow A'$ so that 
\begin{equation}\label{eq-fc10-cond}
(\pair{\mathfrak{W} | \mathfrak{V}}, \pair{\mathcal{W} | \mathcal{V}}) \vdash_\rightarrow \pair{\mathcal{E} | \mathcal{D}} \,\, \text{and} \,\, (\pair{\mathfrak{W} | \mathfrak{V}}, \pair{\mathcal{E} | \mathcal{D}}\}) \vdash_\rightarrow \pair{\mathfrak{G} | \mathfrak{F}}.
\end{equation}

If $\mathfrak{F}$ is empty, set $\pair{\mathcal{E} | \mathcal{D}} = \{ \Delta_\rightarrow \}$. Otherwise, let 
\[
\mathfrak{F} = \{ ((a_1, T_1), d_1), \ldots, ((a_m, T_m), d_m) \}.
\]
 By assumption we know that for $1 \le \nu \le m$,
\begin{equation}\label{eq-fc10-1}
(\mathfrak{W}_2(a_\nu, T_\nu), \ap((a_\nu, T_\nu), \bigcup \mathcal{V})) \vdash' d_\nu.
\end{equation}
Because of the Global Interpolation Property~(\ref{eq-gip}) there is therefore some $(k_\nu, Z_\nu) \in \Con'$ with
\begin{gather}
(\mathfrak{W}_2(a_\nu, T_\nu), \ap((a_\nu, T_\nu), \bigcup \mathcal{V})) \vdash' (k_\nu, Z_\nu) \label{eq-fc10-2} \\
(k_\nu, Z_\nu)\vdash' d_\nu.\label{eq-fc10-3}
\end{gather}

Since
\begin{equation}\label{eq-fc10-4}
(a_\nu, T_\nu) \vdash \ucl((a_\nu, T_\nu), \pr_1(\bigcup \mathcal{V})),
\end{equation}
we similarly obtain some $(e_\nu, U_\nu) \in \Con$ such that 
\begin{gather}
(a_\nu, T_\nu) \vdash (e_\nu, U_\nu) \label{eq-fc10-5} \\
(e_\nu, U_\nu) \vdash  \ucl((a_\nu, T_\nu), \pr_1(\bigcup \mathcal{V})). \label{eq-fc10-6}
\end{gather}

With Lemma~\ref{lem-jextprop}(\ref{lem-jextprop-3}) it follows that 
\begin{equation}\label{eq-fc10-6+}
\ap((a_\nu, T_\nu), \bigcup\mathcal{V}) = \ap((e_\nu, U_\nu), \bigcup\mathcal{V})\,\, \text{and}\,\,
\mathfrak{W}_2(a_\nu, T_\nu) = \mathfrak{W}_2(e_\nu, U_\nu).
\end{equation}

Let $\mathfrak{D}_\nu = \set{((e_\nu, U_\nu), j)}{j \in Z_\nu}$ and set 
\[
\mathfrak{E}_\nu =  
                \begin{cases}
                            \{ ((\Delta, \emptyset), \Delta'), ((e_\nu, U_\nu), k_\nu) \} & \text{if $U_\nu \not= \emptyset$,} \\
                            \{ ((e_\nu, U_\nu), k_\nu) \} & \text{otherwise.}
                \end{cases}            
\]
Then $\mathfrak{E}_\nu \in \ac(\mathfrak{D}_\nu)$ by Lemma~\ref{lem-singext}. It remains to show that $\pair{\mathcal{E} | \mathcal{D}}$ with $\mathcal{E} = \{ \mathfrak{E}_1, \ldots, \mathfrak{E}_m \}$ and $\mathcal{D} = \{ \mathfrak{D}_1. \ldots, \mathfrak{D}_m \}$ meets Requirements~(\ref{eq-fc10-cond}).

As a consequence of Statements~(\ref{eq-fc10-2}) and (\ref{eq-fc10-6+}) we gain that 
\begin{equation}\label{eq-k-nu}
(\mathfrak{W}_2(e_\nu, U_\nu), \ap((e_\nu, U_\nu), \bigcup \mathcal{V})) \vdash' (k_\nu, Z_\nu).
\end{equation}

For the verification of Condition~\ref{dn-fctsp}(\ref{dn-fctsp-3-2}), let $((a, Z), b) \in \mathfrak{E}_\nu$ and set $k = b$. If 
\[
((a, Z), b) = ((\Delta, \emptyset), \Delta'),
\]
 it follows that $U_\nu$ is not empty and therefore $\ds((\Delta, \emptyset), \mathfrak{D}_\nu))$ and $\rs((\Delta, \emptyset), \mathfrak{D}_\nu)$ are both empty. In case that $((a, Z), b) = ((e_\nu, U_\nu), k_\nu)$, we have that $\ds((e_\nu, U_\nu), \mathfrak{D}_\nu) = \{ (e_\nu, U_\nu) \}$ and $\rs((e_\nu, U_\nu), \mathfrak{D}_\nu) = Z_\nu$.  Thus in both cases, 
\begin{equation*}
(\mathfrak{W}_2(a, Z), \bigcup \{\,\ap((c, T), \bigcup \mathcal{V}) \mid (c, T) \in \ds((a, Z), \mathfrak{D}_\nu) \,\}) 
 \vdash' 
 (k, \rs((a, Z), \mathfrak{D}_\nu)),
\end{equation*}
in the first case because of Axiom~\ref{dn-infsys}(\ref{dn-infsys-3}), and in the other one as a consequence of Statements~(\ref{eq-fc10-2}, \ref{eq-fc10-6}) as well as Axiom~\ref{dn-infsys}(\ref{dn-infsys-9}).
This shows that $(\pair{\mathfrak{W} | \mathfrak{V}}, \pair{\mathcal{W} | \mathcal{V}}) \vdash_\rightarrow \pair{\mathcal{E} | \mathcal{D}}$. 

It remains to show that $(\pair{\mathfrak{W} | \mathfrak{V}}, \pair{\mathcal{E} | \mathcal{D}}) \vdash_\rightarrow \pair{\mathfrak{G} | \mathfrak{F}}$. 
Note that by Statement~(\ref{eq-fc10-2}) and Axiom~\ref{dn-infsys}(\ref{dn-infsys-4}), $\{ k_\nu \} \in \Con(\mathfrak{W}_2(a_\nu, T_\nu))$. Moreover, 
\[
Z_\nu \subseteq \ap((a_\nu, T_\nu), \mathfrak{D}_\nu) \subseteq \ap((a_\nu, T_\nu), \bigcup \mathcal{D}).
\]

With Lemma~\ref{lem-fc4} it follows from what we have just seen that $\pair{\mathcal{E} | \mathcal{D}} \in \Con_\rightarrow(\pair{\mathfrak{W} | \mathfrak{V}})$. Thus, $\ap((a_\nu, T_\nu), \bigcup \mathcal{D}) \in \Con'(\mathfrak{W}_2(a_\nu, T_\nu))$. As a consequence of Statement~(\ref{eq-fc10-3}) we now obtain that 
\begin{equation}\label{eq-fc10-cond2}
(\mathfrak{W}_2(a_\nu, T_\nu), \ap((a_\nu, T_\nu), \bigcup \mathcal{D})) \vdash' d_\nu, 
\end{equation}
for all $1 \le \nu \le m$.

For the second condition let $((a, Z), b) \in \mathfrak{G}$. Since $(\pair{\mathfrak{W} | \mathfrak{V}}, \pair{\mathcal{W} | \mathcal{V}}) \vdash_\rightarrow \pair{\mathfrak{G} | \mathfrak{F}}$, there is some $k \in A'$ with $k \sim b \,[\rs((a, Z), \mathfrak{F})]$ and $\{ k \} \in \Con'(\mathfrak{W}_2(a, Z))$.  With Statement~(\ref{eq-fc10-cond2}) it follows that 
\[
(\mathfrak{W}_2(a, Z),  \bigcup \{\,\ap((c, T), \bigcup \mathcal{D}) \mid (c, T) \in \ds((a, Z), \mathfrak{F}) \,\}) 
\vdash' \rs((a, Z), \mathfrak{F}).
\]
Because of Axiom~\ref{dn-infsys}(\ref{dn-infsys-11}) there is thus some $r \in A'$ such that 
\[
(\mathfrak{W}_2(a, Z),  \bigcup \{\,\ap((c, T), \bigcup \mathcal{D}) \mid (c, T) \in \ds((a, Z), \mathfrak{F}) \,\}) \vdash' r
\]
and $\rs((a, Z), \mathfrak{F}) \in \Con'(r)$. Hence, $\{ r \} \in \Con'(\mathfrak{W}_2(a, Z))$, from which it ensues that $r \sim k \,[\rs((a, Z), \mathfrak{F})]$. So, $r \sim b \,[\rs((a, Z), \mathfrak{F})]$ and 
\begin{equation*}
(\mathfrak{W}_2(a, Z),  \bigcup \{\,\ap((c, T), \bigcup \mathcal{D}) \mid (c, T) \in \ds((a, Z), \mathfrak{F}) \,\}) 
\vdash' \\ 
(r, \rs((a, Z), \mathfrak{F})).    \qedhere \hspace{-1em}
\end{equation*}
\end{proof}

\begin{lemma}\label{lem-fc11}
$(A \rightarrow A', \Con_\rightarrow, \vdash_\rightarrow, \Delta_\rightarrow)$ satisfies Condition~\ref{dn-infsys}(\ref{dn-infsys-11}).
\end{lemma}
\begin{proof}
Let $\pair{\mathfrak{W} | \mathfrak{V}}  \in A \rightarrow A'$, $\pair{\mathcal{G} | \mathcal{F}} \fsubset A \rightarrow A'$, and $\pair{\mathcal{W} | \mathcal{V}} \in \Con_\rightarrow(\pair{\mathfrak{W} | \mathfrak{V}})$ with 
\[
(\pair{\mathfrak{W} | \mathfrak{V}}, \pair{\mathcal{W} | \mathcal{V}}) \vdash_\rightarrow \pair{\mathcal{G} | \mathcal{F}}.
\]
We will construct some $\pair{\mathfrak{B} | \mathfrak{A}} \in A \rightarrow A'$ so that 
\begin{equation}\label{eq-fc11-cond}
(\pair{\mathfrak{W} | \mathfrak{V}}, \pair{\mathcal{W} | \mathcal{V}}) \vdash_\rightarrow \pair{\mathfrak{B} | \mathfrak{A}} \,\, \text{and} \,\, \pair{\mathcal{G} | \mathcal{F}} \in  \Con_\rightarrow(\pair{\mathfrak{B} | \mathfrak{A}}).
\end{equation}

If $\bigcup\mathcal{F}$ is empty, set $\pair{\mathfrak{B} | \mathfrak{A}} = \Delta_\rightarrow$. Let us now assume that $\bigcup\mathcal{F}$ is non-empty. In particular, let $\bigcup\mathcal{F} = \{ ((a_1, T_1), d_1), \ldots, ((a_m, T_m), d_m) \}$ and $1 \le \nu \le m$. As in the proof of Lemma~\ref{lem-fc10} there exists $Z_\nu \in \Con'(\mathfrak{W}_2(a_\nu, T_\nu))$ with
\begin{gather}
(\mathfrak{W}_2(a_\nu, T_\nu), \ap((a_\nu, T_\nu), \bigcup \mathcal{V})) \vdash' Z_\nu \label{eq-fc11-2} \\
(\mathfrak{W}_2(a_\nu, T_\nu), Z_\nu)\vdash' d_\nu\label{eq-fc11-3}
\end{gather}
and $(e_\nu, U_\nu) \in \Con$ such that 
\begin{gather}
(a_\nu, T_\nu) \vdash (e_\nu, U_\nu) \label{eq-fc11-5} \\
(e_\nu, U_\nu) \vdash  \ucl((a_\nu, T_\nu), \pr_1(\bigcup \mathcal{V}))). \label{eq-fc11-6}
\end{gather}

Set 
\[
\mathfrak{A} = \set{((e_\nu, U_\nu), j)}{j \in Z_\nu \wedge 1 \le \nu \le m}.
\]

Then it follows as in the derivation of Lemma~\ref{lem-fc10} that 
\begin{equation}\label{eq-k-nu-2}
(\mathfrak{W}_2(e_\nu, U_\nu), \ap((e_\nu, U_\nu), \bigcup \mathcal{V})) \vdash' Z_\nu.
\end{equation}

We will now construct some $\mathfrak{B} \in \ac({\mathfrak{A})}$ satisfying the Requirements~(\ref{eq-fc11-cond}). For each $J \subseteq \{ 1, \ldots, m \}$, let $R(\set{(e_\nu, U_\nu)}{\nu \in J})$ be a system of representatives of $W(\{\,(e_\nu, U_\nu) \mid \nu \in J\,\})$ with respect to $\sim [\bigcup\nolimits_{\nu \in J} U_\nu]$ so that $R(\emptyset) = \{ \Delta \}$ und $R(\{ (e_\nu, U_\nu)\}) = \{ e_\nu\}$, for $1 \le \nu \le m$. By Axiom~\ref{dn-fctsp}(\ref{dn-fctsp-3-2}) there is some $k_{a, Z, \pair{\mathfrak{G} |\mathfrak{F}}} \in A'$, for each $\pair{\mathfrak{G} | \mathfrak{F}} \in \pair{\mathcal{G} | \mathcal{F}}$ and $((a, Z), b) \in  \mathfrak{G}$, with $k_{a, Z, \pair{\mathfrak{G} | \mathfrak{F}}} \sim b \,[\rs((a, Z), \mathfrak{F})]$ and
\begin{equation}\label{eq-fc11-1}
(\mathfrak{W}_2(a,  Z), \bigcup \{\, \ap((e, Y), \bigcup \mathcal{V}) \mid (e, Y) \in \ds((a, Z), \mathfrak{F}) \,\})  
\vdash'  
\!(k_{a, Z, \pair{\mathfrak{G} | \mathfrak{F}}}, \rs((a, Z), \mathfrak{F})).
\end{equation}
In case $((a, Z), b) = ((\Delta, \emptyset), \Delta')$, we choose $k_{\Delta, \emptyset, \pair{\mathfrak{G} |\mathfrak{F}}} = \Delta'$.

\begin{claim}
For every $\mathfrak{J} \subseteq \mathfrak{A}$ and each $a \in R(\pr_1(\mathfrak{J}))$ such that $(a, \mathfrak{J})$ is $\mathfrak{A}$-maximal there is some $t_{a, \mathfrak{J}} \in A'$ with the next two properties:
\begin{align} \label{eq-fc11-9}
(\mathfrak{W}_2&(a,  \bigcup \pr_{2,1}(\mathfrak{J})), \bigcup \{\, \ap((e, U), \bigcup \mathcal{V}) \mid (e, U) \in \pr_1(\mathfrak{J}) \,\}) \vdash'  \notag \\
&\{ t_{a, \mathfrak{J}}, \Delta' \} \cup \pr_2(\mathfrak{J})   \cup   \{\, t_{c, \mathfrak{K}} \mid \mathfrak{K} \subsetneqq \mathfrak{J} \wedge c \sim a \,[\bigcup \pr_{2,1}(\mathfrak{K})] \wedge (c, \mathfrak{K})
\notag \\ 
&\text{$\mathfrak{A}$-maximal} \,\} \cup \{\, k_{c, Z, \pair{\mathfrak{G} | \mathfrak{F}}} \mid 
(\exists L \subseteq \{ 1, \ldots, m \}) (\exists b \in A')  \{\, (e_\mu, U_\mu) \mid  \notag \\
& \mu \in L \,\} \subseteq \pr_1(\mathfrak{J}) \wedge Z = \bigcup\nolimits_{\mu \in L} T_\mu  \wedge c \sim a \,[Z] \wedge \pair{\mathfrak{G} | \mathfrak{F}} \in \pair{\mathcal{G} | \mathcal{F}}  \wedge   \mbox{} \notag  \\
&((c, Z), b) \in \mathfrak{G}  \wedge b \sim k_{c, Z, \pair{\mathfrak{G} | \mathfrak{F}}} \,[\rs((c, Z), \mathfrak{F})]  \,\},
\end{align}
\begin{align}\label{eq-fc11-10}
& \Con'(t_{a, \mathfrak{J}}) \ni \{ \Delta' \} \cup \pr_2(\mathfrak{J}) \cup 
 \{\, t_{c, \mathfrak{K}} \mid \mathfrak{K} \subsetneqq \mathfrak{J} \wedge  c \sim a \,[\bigcup \pr_{2,1}(\mathfrak{K})] \wedge \mbox{} \notag \\
&\hspace{2em}   \text{$(c, \mathfrak{K}) $ $\mathfrak{A}$-maximal} \,\}   \cup  \{\, k_{c, Z, \pair{\mathfrak{G} | \mathfrak{F}}} \mid (\exists L \subseteq \{ 1, \ldots, m \}) (\exists b \in A') \notag \\
&\hspace{2em}  \{\, (e_\mu, U_\mu) \mid  \mu \in L \,\} \subseteq \pr_1(\mathfrak{J})  
\wedge Z = \bigcup\nolimits_{\mu \in L} T_\mu \wedge c \sim a \,[Z] \wedge \mbox{} \notag \\
&\hspace{2em} \pair{\mathfrak{G} | \mathfrak{F}} \in \pair{\mathcal{G} | \mathcal{F}}  \wedge  ((c, Z), b) \in \mathfrak{G}  \wedge  b \sim k_{c, Z, \pair{\mathfrak{G} | \mathfrak{F}}} \,[\rs((c, Z), \mathfrak{F})]  \,\}.
\end{align}
\end{claim}

With respect to the last set on the left hand side note that if $((c, Z), b) \in \mathfrak{G}$, then there is some $\mathfrak{L} \subseteq \mathfrak{F}$ so that $(c, \mathfrak{L})$ is $\mathfrak{F}$-maximal and $Z= \bigcup \pr_{2,1}(\mathfrak{L})$. It follows that $\| \mathfrak{L} \| \le \| \mathfrak{J} \|$.

The claim is shown by induction on the cardinality $\kappa$ of the subset $\mathfrak{J}$. The case $\kappa = 0$ is obvious. Set $t_{a, \emptyset} = \Delta'$, if $(a, \emptyset)$ is $\mathfrak{A}$-maximal. All other sets on the right hand side in Statement~(\ref{eq-fc11-9}) are empty in this case, except the last one, which is the singleton $\{ \Delta' \}$, if for some $ \pair{\mathfrak{G} | \mathfrak{F}} \in \pair{\mathcal{G} | \mathcal{F}}$, $((\Delta, \emptyset), \Delta') \in \mathfrak{G}$. 

Assume next that the claim holds for all subsets $\mathfrak{K} \subseteq \mathfrak{A}$ of cardinality $\kappa$ and let $\mathfrak{J} \subseteq \mathfrak{A}$ with $\| \mathfrak{J} \| = \kappa + 1$. Then 
\begin{align*} 
&(\mathfrak{W}_2(a,  \bigcup \pr_{2,1}(\mathfrak{J})), \bigcup \{\, \ap((e, U), \bigcup \mathcal{V}) \mid (e, U) \in \pr_1(\mathfrak{J}) \,\}) \vdash'  \\
&\hspace{.6em}  \{ \Delta' \}  \cup \pr_2(\mathfrak{J}) \cup   \{\, t_{c, \mathfrak{K}} \mid \mathfrak{K} \subsetneqq \mathfrak{J} \wedge c \sim a \,[\bigcup \pr_{2,1}(\mathfrak{K})] \wedge \text{$(c, \mathfrak{K})$ $\mathfrak{A}$-maximal} \,\} \cup \mbox{} \\
&\hspace{.6em}  \{\, k_{c, Z, \pair{\mathfrak{G} | \mathfrak{F}}} \mid  (\exists b \in A')   (\exists L \subseteq \{ 1, \ldots, m \}) \{\, (e_\mu, U_\mu) \mid \mu \in L \,\} \subseteq \pr_1(\mathfrak{J}) \wedge \mbox{} \\
&\hspace{.6em} Z = \bigcup\nolimits_{\mu \in L} T_\mu \wedge c \sim a \,[Z]   \wedge \pair{\mathfrak{G} | \mathfrak{F}} \in \pair{\mathcal{G} | \mathcal{F}} \wedge ((c, Z), b) \in \mathfrak{G}  \wedge \mbox{} \\
&\hspace{.6em}  b \sim k_{c, Z, \pair{\mathfrak{G} | \mathfrak{F}}} \,[\rs((c, Z), \mathfrak{F})]  \,\}, 
 \end{align*}
because of Statement~(\ref{eq-k-nu-2}), Axiom~\ref{dn-fctsp}(\ref{dn-fctsp-3-2}), and the induction hypothesis. By Condition~\ref{dn-infsys}(\ref{dn-infsys-11}) there is hence some $t_{a, \mathfrak{J}} \in A'$ so that Properties~(\ref{eq-fc11-9}) and (\ref{eq-fc11-10}) hold.

With help of the claim we can now define $\mathfrak{B}$. Let $\mathfrak{B}^{(0)}$ be as in Definition~\ref{dn-ass} and for $\kappa \ge 1$ set
\begin{equation*}
\widehat{\mathfrak{B}}^{(\kappa)} = 
 \{\, ((a, \bigcup \pr_{2,1}(\mathfrak{J})), t_{a,\mathfrak{J}}) \mid \mathfrak{J} \subseteq \mathfrak{A} \wedge \| \mathfrak{J} \| = \kappa \wedge  \mbox{} \\
 a \in R(\pr_1(\mathfrak{J})) \wedge \text{$(a, \mathfrak{J})$ $\mathfrak{A}$-maximal} \,\}.
\end{equation*}
Obviously, $\mathfrak{B} \in \ac(\mathfrak{A})$. Moreover, $(\pair{\mathfrak{W} | \mathfrak{V}}, \pair{\mathcal{W} | \mathcal{V}}) \vdash_\rightarrow \pair{\mathfrak{B} | \mathfrak{A}}$: Condition~\ref{dn-fctsp}(\ref{dn-fctsp-3-1}) follows with Statement~(\ref{eq-k-nu-2}), and Condition~\ref{dn-fctsp}(\ref{dn-fctsp-3-2}) is a consequence of Statement~(\ref{eq-fc11-9}). It remains to show that 
\[
\pair{\mathcal{G} | \mathcal{F}} \in \Con_\rightarrow(\pair{\mathfrak{B} | \mathfrak{A}}).
\]

For Condition~\ref{dn-fctsp}(\ref{dn-fctsp-2-1}) let $\pair{\mathfrak{G} | \mathfrak{F}} \in \pair{\mathcal{G} | \mathcal{F}}$ and $(a, S) \in \Con$. Then there is some $M \subseteq \{ 1, \ldots, m \}$ such that 
\[
\Cl((a, S), \mathfrak{F}) = \{\, (a_\mu, T_\mu) \mid \mu \in M \,\}.
\]
In addition, there are $e \in A$ and $b \in A'$ with 
$a \sim e \,[\bigcup_{\mu \in M} T_\mu)]
$
 so that $((e, \bigcup_{\mu \in M} T_\mu), b) = \mathfrak{G}(a, S)$. Moreover, 
 \[
 b \sim k_{e, \bigcup\nolimits_{\mu \in M} T_\mu, \pair{\mathfrak{G} | \mathfrak{F}}} \,[\rs((e, \bigcup\nolimits_{\mu \in M} T_\mu), \mathfrak{F})],
 \]
  that is 
\[
\mathfrak{G}_2(a, S) \sim k_{e, \bigcup\nolimits_{\mu \in M} T_\mu, \pair{\mathfrak{G} | \mathfrak{F}}} \,[\ap((a, S), \mathfrak{F})].
\]
Let $\mathfrak{J} \subseteq \mathfrak{A}$ such that 
\[
\Cl((a, S), \pr_1(\mathfrak{A})) = \pr_1(\mathfrak{J}).
\]
Then $(a, \mathfrak{J})$ is $\mathfrak{A}$-maximal. Hence, there is some $c \in A$ such that $a \sim c \,[\bigcup \pr_{2,1}(\mathfrak{J})]$ and $((c, \bigcup \pr_{2,1}(\mathfrak{J})), t_{c, \mathfrak{J}}) = \mathfrak{B}(a, S)$. Because of Statement~(\ref{eq-fc11-5}) we moreover have that $\{\, (e_\mu, U_\mu) \mid \mu \in M \,\} \subseteq \pr_1(\mathfrak{J})$. Thus, $\{ k_{e, \bigcup\nolimits_{\mu \in M} T_\mu, \pair{\mathfrak{G} | \mathfrak{F}}} \} \in \Con'(\mathfrak{B}_2(a, S))$, by Statement~(\ref{eq-fc11-10}).

For Condition~\ref{dn-fctsp}(\ref{dn-fctsp-2-2}), finally, let $1 \le \nu \le m$. Note that by Statement~(\ref{eq-fc11-5}),
\[
Z_\nu \subseteq \ap((a_\nu, T_\nu),  \mathfrak{A}).
\]
With Lemma~\ref{lem-fc4} it follows from what we have just seen that $\{ \pair{\mathfrak{B} | \mathfrak{A}} \} \in \Con_\rightarrow(\pair{\mathfrak{W} | \mathfrak{V}})$. Thus, $\ap((a_\nu, T_\nu), \mathfrak{A}) \in \Con'(\mathfrak{W}_2(a_\nu, T_\nu))$. As a consequence of Statement~(\ref{eq-fc11-3}) we obtain that 
\begin{equation*}\label{eq-fc11-cond2}
(\mathfrak{W}_2(a_\nu, T_\nu), \ap((a_\nu, T_\nu),  \mathfrak{A})) \vdash' d_\nu,
\end{equation*}
from which we gain with Lemma~\ref{lem-extvdash} that for all $(i, X) \in \Con$,
\[
(\mathfrak{W}_2(i, X), \ap((i, X),  \mathfrak{A})) \vdash' \ap((i, X), \bigcup \mathcal{F}).
\]
Hence,
\[
(\mathfrak{B}_2(i, X), \ap((i, X),  \mathfrak{A})) \vdash' \ap((i, X), \bigcup \mathcal{F}),
\]
because of Lemmas~\ref{lem-extwit}  and \ref{lem-witeqprop}(\ref{lem-witeqprop-1}). Condition~\ref{dn-fctsp}(\ref{dn-fctsp-2-2}) now follows with Lemma~\ref{lem-witeqprop}(\ref{lem-witeqprop-3}). 
\end{proof}

\begin{lemma}\label{lem-fc7}
$(A \rightarrow A', \Con_\rightarrow, \vdash_\rightarrow, \Delta_\rightarrow)$ satisfies Condition~\ref{dn-infsys}(\ref{dn-infsys-7}).
\end{lemma}
\begin{proof}
Let $\{\pair{\mathfrak{B} | \mathfrak{A}}\} \in \Con_\rightarrow(\pair{\mathfrak{W} | \mathfrak{V}})$ and $\pair{\mathcal{G} | \mathcal{F}} \in \Con_\rightarrow(\pair{\mathfrak{B} | \mathfrak{A}})$. We have to show that $\pair{\mathcal{G} | \mathcal{F}} \in \Con_\rightarrow(\pair{\mathfrak{W} | \mathfrak{V}})$.

Let $(a, S) \in \Con$. As $\{\pair{\mathfrak{B} | \mathfrak{A}}\} \in \Con_\rightarrow(\pair{\mathfrak{W} | \mathfrak{V}})$, there is some $k_{\mathfrak{A}} \in A'$ so that $\{ k_{\mathfrak{A}} \} \in \Con'(\mathfrak{W}_2(a, S))$ and $k_{\mathfrak{A}} \sim \mathfrak{B}_2(a, S) \,[\ap((a, S), \mathfrak{A})]$. Now, using that  $\pair{\mathcal{G} | \mathcal{F}} \in \Con_\rightarrow(\pair{\mathfrak{B} | \mathfrak{A}})$, it follows that $k_{\mathfrak{A}} \sim \mathfrak{B}_2(a, S) \,[\ap((a, S), \bigcup \mathcal{F})]$, from which we obtain with Lemma~\ref{lem-witeqprop}(\ref{lem-witeqprop-2}) that for every $\pair{\mathfrak{G} | \mathfrak{F}} \in \pair{\mathcal{G} | \mathcal{F}}$, $k_{\mathfrak{A}} \sim \mathfrak{B}_2(a, S) \,[\ap((a, S), \mathfrak{F})]$.

Moreover, there is some $k_{\mathfrak{F}} \in A'$ so that $k_{\mathfrak{F}} \sim \mathfrak{G}_2(a, S) \,[\ap((a, S), \mathfrak{F})]$ and $\{k_{\mathfrak{F}} \} \in \Con'(\mathfrak{B}_2(a, S))$. It follows that 
\[
\mathfrak{G}_2(a, S) \sim \mathfrak{B}_2(a, S) \,[\ap((a, S), \mathfrak{F})].
\]
 This shows that 
 \[
 \mathfrak{G}_2(a, S) \sim k_{\mathfrak{A}} \,[\ap((a, S), \mathfrak{A})].
 \]
  As said, $\{ k_{\mathfrak{A}} \} \in \Con'(\mathfrak{W}_2(a, S))$. 

Thus, the first of the two conditions in Def.~\ref{dn-fctsp}(\ref{dn-fctsp-2}) holds. For the second one let $r \in A'$ with $r \sim \mathfrak{W}_2(a, S) \,[\ap((a, S), \mathfrak{V})]$. Then 
\[
r \sim \mathfrak{W}_2(a, S) \,[\ap((a, S), \mathfrak{A})],
\]
 as $\{\pair{\mathfrak{B} | \mathfrak{A}}\} \in \Con_\rightarrow(\pair{\mathfrak{W} | \mathfrak{V}})$. Since $\{ k_{\mathfrak{A}} \} \in \Con'(\mathfrak{W}_2(a, S))$, it follows that 
 \[
 r \sim k_{\mathfrak{A}} \,[\ap((a, S), \mathfrak{A})],
 \]
  and because $k_{\mathfrak{A}} \sim \mathfrak{B}_2(a, S) \,[\ap((a, S), \mathfrak{A})]$, we obtain that $r \sim \mathfrak{B}_2(a, S) \,[\ap((a, S), \mathfrak{A})]$. Moreover, 
 \[
 \mathfrak{B}_2(a, S) \sim \mathfrak{W}_2(a, S) \,[\ap((a, S), \mathfrak{A})].
 \]
  Now, note that for any $s \in A'$ with $s \sim \mathfrak{B}_2(a, S) \,[\ap((a, S), \mathfrak{A})]$ we have that 
\[
s \sim \mathfrak{B}_2(a, S) \,[\ap((a, S), \bigcup \mathcal{F})].
\]
 As consequence we gain that $r \sim \mathfrak{W}_2(a, S) \,[\ap((a, S), \bigcup \mathcal{F})]$.
\end{proof}

\begin{lemma}\label{lem-fc8}
$(A \rightarrow A', \Con_\rightarrow, \vdash_\rightarrow, \Delta_\rightarrow)$ satisfies Condition~\ref{dn-infsys}(\ref{dn-infsys-8}).
\end{lemma}
\begin{proof}
Let $\{\pair{\mathfrak{B} | \mathfrak{A}}\} \in \Con_\rightarrow(\pair{\mathfrak{W} | \mathfrak{V}})$,   $\pair{\mathcal{G} | \mathcal{F}} \in \Con_\rightarrow(\pair{\mathfrak{B} | \mathfrak{A}})$ and $\pair{\mathfrak{Z} | \mathfrak{Y}} \in A \rightarrow A'$ so that $(\pair{\mathfrak{B} | \mathfrak{A}}, \pair{\mathcal{G} | \mathcal{F}}) \vdash_\rightarrow \pair{\mathfrak{Z} | \mathfrak{Y}}$. We must prove that $(\pair{\mathfrak{W} | \mathfrak{V}}, \pair{\mathcal{G} | \mathcal{F}}) \vdash_\rightarrow \pair{\mathfrak{Z} | \mathfrak{Y}}$. 

Let $((c, T), b) \in \mathfrak{Y}$ and choose $k_{\mathfrak{A}}$ according to Condition~\ref{dn-fctsp}(\ref{dn-fctsp-2-1}). Then $\mathfrak{B}_2(c, T) \sim k_{\mathfrak{A}} \,[\ap((c, T), \mathfrak{A})]$ and $\{ k_{\mathfrak{A}} \} \in \Con'(\mathfrak{W}_2(c, T))$.
    Moreover, we have that 
 \[
 (\mathfrak{B}_2(c, T), \ap((c, T), \bigcup \mathcal{F})) \vdash' b.
 \]
  With Lemma~\ref{lem-witeqprop}(\ref{lem-witeqprop-1}) it follows that also $(k_{\mathfrak{A}}, \ap((c, T), \mathfrak{A})) \vdash' b$. Hence, 
  \[
  (\mathfrak{W}_2(c, T), \ap((c, T), \mathfrak{A})) \vdash' b,
  \]
   by Axiom~\ref{dn-infsys}(\ref{dn-infsys-8}). The other condition follows similarly.
\end{proof}

The verification of Condition~\ref{dn-infsys}(\ref{dn-infsys-9}) proceeds in an analogous way. 

\begin{proposition}\label{prop-fct-alg}
Let $(A, \Con, \vdash, \Delta)$ and $(A', \Con', \vdash', \Delta')$ be information systems with witnesses satisfying Condition~(\ref{alg}). Then $(A \rightarrow A', \Con_\rightarrow, \vdash_\rightarrow, \Delta_\rightarrow)$ satisfies (\ref{alg}) as well.
\end{proposition}
\begin{proof}
As follows from a slight modification of the construction in Lemma~\ref{lem-fc10}, Condition~(\ref{salg}) is satisfied: Because of Condition~(\ref{alg}) $(e_\nu, U_\nu) \in \Con$ and $(k_\nu, Z_\nu) \in \Con'$ can be chosen as reflexive, for each $1 \le \nu \le m$. Moreover, as a consequence of Statement~(\ref{eq-k-nu}), $\{ k_\nu \} \in \Con'(\mathfrak{W}_2(e_\nu, U_\nu))$. Therefore, we have that $(\mathfrak{W}_2(e_\nu, U_\nu), Z_\nu) \vdash' (k_\nu, Z_\nu)$. By definition of $\pair{\mathcal{E} | \mathcal{D}}$, $Z_\nu \subseteq \ap((e_\nu, U_\nu), \bigcup \mathcal{D})$. Hence,
\[
(\mathfrak{W}_2(e_\nu, U_\nu), \ap((e_\nu, U_\nu), \bigcup \mathcal{D})) \vdash' \{k_\nu \} \cup Z_\nu,
\]
which implies that $(\pair{\mathfrak{W} | \mathfrak{V}}, \pair{\mathcal{E} | \mathcal{D}}) \vdash_\rightarrow \pair{\mathcal{E} | \mathcal{D}}$.
\end{proof}

Condition~(\ref{bc}) as well is inherited to the function space.

\begin{proposition}\label{prop-fct-bc}
Let $(A, \Con, \vdash, \Delta)$ and $(A', \Con', \vdash', \Delta')$ be information systems with witnesses. If $(A', \Con', \vdash', \Delta')$ satisfies Condition~(\ref{bc}), so does $(A \rightarrow A', \Con_\rightarrow, \vdash_\rightarrow, \Delta_\rightarrow)$.
\end{proposition}
\begin{proof}
Let $\pair{\mathfrak{B} | \mathfrak{A}}, \pair{\mathfrak{W} | \mathfrak{V}} \in A \rightarrow A'$ and $\pair{\mathcal{G} | \mathcal{F}} \in \Con_\rightarrow(\pair{\mathfrak{B} | \mathfrak{A}}) \cap \Con_\rightarrow(\pair{\mathfrak{W} | \mathfrak{V}})$. We need to show that for any $\pair{\mathfrak{I} | \mathfrak{H}} \in A \rightarrow A'$, 
\begin{equation}\label{eq-fct-bc}
(\pair{\mathfrak{B} | \mathfrak{A}}, \pair{\mathcal{G} | \mathcal{F}}) \vdash_\rightarrow \pair{\mathfrak{I} | \mathfrak{H}}
\Longleftrightarrow 
(\pair{\mathfrak{W} | \mathfrak{V}}, \pair{\mathcal{G} | \mathcal{F}}) \vdash_\rightarrow \pair{\mathfrak{I} | \mathfrak{H}}.
\end{equation}

As a consequence of Axiom~\ref{dn-fctsp}(\ref{dn-fctsp-2-2}) we have that for any $(a, S) \in \Con$, 
\[
\ap((a, S), \bigcup \mathcal{F}) \in \Con'(\mathfrak{B}_2(a, S)) \cap \Con'(\mathfrak{W}_2(a, S)).
\]
 Then also 
\begin{equation*}
\bigcup \{\, \ap((e, Y), \bigcup \mathcal{F}) \mid (e, Y) \in \ds((a, S), \mathfrak{H}) \,\} \in  \\ \Con'(\mathfrak{B}_2(a, S)) \cap \Con'(\mathfrak{W}_2(a, S)). 
\end{equation*}
Because of Condition~(\ref{bc}) it follows for all $((e, Y), f) \in \mathfrak{H}$ that
\[
(\mathfrak{B}_2(e, Y), \ap((e, Y), \bigcup \mathcal{F})) \vdash' f
\Longleftrightarrow 
(\mathfrak{W}_2(e, Y), \ap((e, Y), \bigcup \mathcal{F})) \vdash' f,
\]
and for all $((c, Z), b) \in \mathfrak{I})$ and $k \in A'$ that
%\pagebreak
\begin{equation*}
(\mathfrak{B}_2(c, Z), \bigcup \{\, \ap((e, Y), \bigcup \mathcal{F}) \mid (e, Y) \in \ds((c, Z), \mathfrak{H}) \,\}) \vdash' \\
(k, \rs((c, Z), \mathfrak{H}))
\end{equation*}
\centerline{$\Updownarrow$}
 \begin{equation*}
(\mathfrak{W}_2(c, Z), \bigcup \{\, \ap((e, Y), \bigcup \mathcal{F}) \mid (e, Y) \in \ds((c, Z), \mathfrak{H}) \,\}) \vdash'  \\
(k, \rs((c, Z), \mathfrak{H})),
\end{equation*}
which implies the Statement~(\ref{eq-fct-bc}).
\end{proof}

Let $f \in |A \rightarrow A'|$. Then $f \subseteq A \rightarrow A'$, i.e., $f$ is a set of pairs in $\mathcal{P}(\Con \times A') \times \mathcal{P}_f(\Con \times A')$ such that the first component is an associate of the second. We will now show that the states of $A \rightarrow A'$ correspond to the approximable mappings between $A$ and $A'$ in a one-to-one way.

\begin{lemma}\label{lem-state-approx}
For $f \in |A \rightarrow A'|$, let $\am(f) = \bigcup \pr_2(f)$. Then $\am(f) : A \Vdash A'$.
\end{lemma}
\begin{proof}
In order to show that $\am(f)$ is an approximable mapping we need to verify Conditions~(\ref{dn-am})(\ref{dn-am-6}-\ref{dn-am-7}).

(\ref{dn-am-6}) is a consequence of the fact that because of Axiom~\ref{dn-infsys}(\ref{dn-infsys-3}) and Condition~\ref{dn-st}(\ref{dn-st-2}), $\Delta_\rightarrow$ is contained in every state.

(\ref{dn-am-2}) Let $i \in A$, $X, X' \in \Con(i)$, and $b \in A'$ so that $X \subseteq X'$ and $(i, X)\am(f)b$. We need to show that $(i, X')\am(f)b$. 

Since $(i, X)\am(f) b$, there is some $\fcbael{B}{A} \in f$ with $((i, X), b) \in \mathfrak{A}$. Because of Condition~\ref{dn-st}(\ref{dn-st-3}) it follows that $(\fcbael{W}{V}, \fcbaset{W}{V}) \vdash_\rightarrow \fcbael{B}{A}$, for some $\fcbael{W}{V} \in f$ and $\fcbaset{W}{V} \fsubset f$ with $\fcbaset{W}{V} \in \Con_\rightarrow(\fcbael{W}{V})$. Thus, $(\mathfrak{W}_2(i, X), \ap((i, X), \bigcup\mathcal{V})) \vdash' b$.  By Lemma~\ref{lem-jextprop}(\ref{lem-jextprop-2}), we have that $\{ \mathfrak{W}_2(i, X) \} \in \Con'(\mathfrak{W}_2(i, X'))$. Moreover, $\ap((i, X), \bigcup\mathcal{V}) \subseteq \ap((i, X'), \bigcup\mathcal{V})$. Therefore,
\begin{equation}\label{eq-sa-1+}
(\mathfrak{W}_2(i, X'), \ap((i, X'), \bigcup\mathcal{V})) \vdash' b.
\end{equation}
Because of  Condition~\ref{dn-infsys}(\ref{dn-infsys-3}) we also obtain that 
\[
(\mathfrak{W}_2(i, X'), \ap((i, X'), \bigcup\mathcal{V})) \vdash' \Delta'.
\]
 Consequently, by Axiom~\ref{dn-infsys}(\ref{dn-infsys-11}), there is some $e \in A'$ with 
 \[(\mathfrak{W}_2(i, X'), \ap((i, X'), \bigcup\mathcal{V})) \vdash' e
 \]
  and $\{ b, \Delta' \} \in \Con'(e)$.

Now, set $\mathfrak{D} = \{ ((i, X'), b) \}$ and let $\mathfrak{E} \in \ac(\mathfrak{D}$ be as in Lemma~\ref{lem-singext}. Then Statement~(\ref{eq-sa-1+}) means that $(\fcbael{W}{V}, \fcbaset{W}{V}) \vdash_\rightarrow \fcbael{E}{D}$: for Condition~\ref{dn-fctsp}(\ref{dn-fctsp-3-2}) choose $k = e$ or $k = \Delta'$, respectively. By Condition~\ref{dn-st}(\ref{dn-st-2}) it follows that $\pair{\mathfrak{E} | \mathfrak{D}} \in f$. Hence $(i, X') \am(f) b$.

(\ref{dn-am-3}) and (\ref{dn-am-4})  follow similary.

(\ref{dn-am-1}) Let $(i, X) \in \Con$, $(k, Y) \in \Con'$ and $b \in A'$ so that $(i, X)\am(f) (k, Y)$ and $(k, Y) \vdash' b$. We must show that $(i, X)\am(f) b$.

Since $(i, X)\am(f) (k, Y)$, we have again that for all $c \in \{ k \} \cup Y$ there are tokens $\pair{\mathfrak{B}_c | \mathfrak{A}_c} \in f$ so that $((i, X), c) \in \mathfrak{A}_c$. Let $\pair{\mathcal{G} | \mathcal{F}} = \{\, \pair{\mathfrak{B}_c | \mathfrak{A})_c} \mid c \in \{ k \} \cup Y \, \}$. Because of Condition~(\ref{st}) there exist $\pair{\mathfrak{W} | \mathfrak{V}} \in f$ and $\pair{\mathcal{W} | \mathcal{V}} \fsubset f$ so that $(\pair{\mathfrak{W} | \mathfrak{V}}, \pair{\mathcal{W} | \mathcal{V}}) \vdash_\rightarrow \pair{\mathcal{G} | \mathcal{F}}$. It follows that
\[
(\mathfrak{W}_2(i, X), \ap((i, X), \bigcup \mathcal{V})) \vdash' (k, Y).
\]
Since $(k, Y) \vdash' b$, we have that $(\mathfrak{W}_2(i, X), \ap((i, X), \bigcup \mathcal{V})) \vdash' b$, from which it follows as above that  $(i, X) \am(f) b$.

It remains to verify Conditions~\ref{dn-am}(\ref{dn-am-5-1}, \ref{dn-am-5-2}, \ref{dn-am-8}, \ref{dn-am-7}). Because of Lemmas~\ref{pn-amleft}, \ref{pn-amright} and~\ref{pn-amint} it suffices to verify Requirement~(\ref{eq-amint}) instead.

 Let $(i, X) \in \Con$ and $F \fsubset A'$ with $(i, X)\am(f) F$. We have to show that there are $(c, U) \in \Con$ and $(e, V) \in \Con'$ so that $(i, X) \vdash (c, U)$, $(c, U)\am(f) (e, V)$ and $(e,V) \vdash' F$.

As we have already seen, there are $\pair{\mathfrak{W} | \mathfrak{V}} \in f$ and $\pair{\mathcal{W} | \mathcal{V}} \fsubset f$ with 
\[
(\mathfrak{W}_2(i, X), \ap((i, X), \bigcup \mathcal{V})) \vdash' F.
\]
Moreover, by the Global Interpolation Property~(\ref{eq-gip}), there exists $(c, U) \in \Con$ such that $(i, X) \vdash (c, U)$ and $(c, U) \vdash \ucl((i, X), \pr_1(\bigcup\mathcal{V}))$.  It follows with Lemma~\ref{lem-jextprop}(\ref{lem-jextprop-3}) that also $(\mathfrak{W}_2(c,  U), \ap((c, U), \bigcup \mathcal{V})) \vdash' F$. Apply the Global Interpolation Property again to obtain some $(e, V) \in \Con'$ so that $(\mathfrak{W}_2(c,  U), \ap((c, U), \bigcup \mathcal{V})) \vdash' (e, V)$ and $(e, V) \vdash' F$. Let $V \cup \{ e \} = \{ b_1, \ldots, b_n \}$, and for $1 \le \nu \le n$, set $\mathfrak{D}_\nu = \{ ((c, U), b_\nu) \}$. Moreover, let $\mathfrak{E}_\nu \in \ac(\mathfrak{D}_\nu)$ be as in Lemma~\ref{lem-singext}. Then $(\pair{\mathfrak{W} | \mathfrak{V}}, \pair{\mathcal{W} | \mathcal{V}}) \vdash_\rightarrow \pair{\mathfrak{E}_\nu | \mathfrak{D}_\nu}$. It follows that $\pair{\mathfrak{E}_\nu | \mathfrak{D}_\nu} \in f$ and hence that $(c, U)\am(f)(e, V)$.
\end{proof}

\begin{lemma}\label{lem-approx-state} 
For $\apmap{H}{A}{A'}$, let $\St(H)$ be the set of all $\pair{\mathfrak{B} | \mathfrak{A}}\in A \rightarrow A'$ such that the following two conditions hold:
\begin{enumerate}
\item \label{lem-approx-state-1}
$\mathfrak{A} \subseteq H$

\item \label{lem-approx-state-2}
$(\forall ((a, Z), b) \in \mathfrak{B}) (\exists j \in A')\, (a, Z) H j \wedge j \sim b \,[\rs((a, Z), \mathfrak{A})].$

\end{enumerate}
Then $\St(H) \in |A \rightarrow A'|$.
\end{lemma}
\begin{proof}
In order to verify that $\St(H)$ is a state of $|A \rightarrow A'|$, we check Conditions~\ref{dn-st}(\ref{dn-st-2}) and (\ref{st}).

\ref{dn-st}(\ref{dn-st-2}) Let $\pair{\mathfrak{W} | \mathfrak{V}}, \pair{\mathfrak{B} | \mathfrak{A}} \in \Con_\rightarrow$ and $\pair{\mathcal{W} | \mathcal{V}} \fsubset \Con_\rightarrow(\pair{\mathfrak{W} | \mathfrak{V}})$ such that $\{ \pair{\mathfrak{W} | \mathfrak{V}} \} \cup \pair{\mathcal{W} | \mathcal{V}} \subseteq \St(H)$ and $(\pair{\mathfrak{W} | \mathfrak{V}}, \pair{\mathcal{W} | \mathcal{V}}) \vdash_\rightarrow \pair{\mathfrak{B} | \mathfrak{A}}$. We must show that  $\pair{\mathfrak{B} | \mathfrak{A}} \in \St(H)$.

For the first requirement let $((c, U), d) \in \mathfrak{A}$. Then we have that 
\[
(\mathfrak{W}_2(c, U), \ap((c, U), \bigcup\mathcal{V})) \vdash' d.
\]

By assumption, $\pair{\mathfrak{W} | \mathfrak{V}} \in \St(H)$. Thus, there is some $j \in A'$ such that $(c, U) H j$ and $j \sim \mathfrak{W}_2(c, U) \,[\ap((c, U), \mathfrak{V})]$. Since $\pair{\mathcal{W} | \mathcal{V}} \in \Con_\rightarrow(\pair{\mathfrak{W} | \mathfrak{V}})$, we conclude that  $j \sim \mathfrak{W}_2(c, U) \,[\ap((c, U), \bigcup\mathcal{V})]$. With Lemma~\ref{lem-witeqprop}(\ref{lem-witeqprop-1}) we thus obtain that also
\[
(j, \ap((c, U), \bigcup\mathcal{V})) \vdash' d.
\]

Now, let $b \in \ap((c, U), \bigcup\mathcal{V})$. Then there is $(e, Y) \in \Con$ so that $((e, Y), b) \in \bigcup\mathcal{V}$ and for some $e' \in A$ with $e \sim e' \,[Y]$, $(c, U) \vdash (e', Y)$. As $\bigcup \mathcal{V} \subseteq H$, it follows with Lemma~\ref{lem-witeqam} that $(e', Y) H b$ and hence with Lemma~\ref{lem-amstrong3} that $(c, U) H b$. This shows that 
\[
(c, U) H (j, \ap((c, U), \bigcup\mathcal{V})).
\]
Consequently, by Condition~\ref{dn-am}(\ref{dn-am-1}) $(c, U) H d$.

For the second requirement let $((a, Z), b) \in \mathfrak{B}$. By assumption, $(\pair{\mathfrak{W} | \mathfrak{V}}, \pair{\mathcal{W} | \mathcal{V}}) \vdash_\rightarrow \pair{\mathfrak{B} | \mathfrak{A}}$. Hence there is some $k \in A'$ such that $b \sim k \,[\rs((a, Z), \mathfrak{A})]$ and
\[
(\mathfrak{W}_2(a, Z), \ap((a, Z), \bigcup\mathcal{V})) \vdash' (k, \rs((a, Z), \mathfrak{A})),
\]
from which it follows as above that $(a, Z) H k$.

(\ref{st}) Let $\pair{\mathcal{G} | \mathcal{F}} \fsubset \St(H)$. We need to show that there are $\pair{\mathfrak{W} | \mathfrak{V}} \in \St(H)$ and $\pair{\mathcal{W} | \mathcal{V}} \fsubset \St(H)$ with $\pair{\mathcal{W} | \mathcal{V}} \in \Con_\rightarrow(\pair{\mathfrak{W} | \mathfrak{V}})$ so that $(\pair{\mathfrak{W} | \mathfrak{V}}, \pair{\mathcal{W} | \mathcal{V}}) \vdash_\rightarrow \pair{\mathcal{G} | \mathcal{F}}$. The construction proceeds as in Lemma~\ref{lem-fc11}.  Instead of Condition~\ref{dn-fctsp}(\ref{dn-fctsp-3-2}), however, we now make use of Requirement~\ref{lem-approx-state}(\ref{lem-approx-state-2}).

Let $\bigcup\mathcal{F} = \{\, ((c_\nu, T_\nu), e_\nu) \mid 1 \le \nu \le n \,\}$. Then $(c_\nu, T_\nu) H e_\nu$. By Rule~(\ref{eq-amint}) there are $(i_\nu, Y_\nu) \in \Con$ and $(j_\nu, Z_\nu) \in \Con'$ with
\begin{gather}
(c_\nu, T_\nu) \vdash (i_\nu, Y_\nu) \label{eq-amst-1}\\
(i_\nu, Y_\nu) H (j_\nu, Z_\nu) \label{eq-amst-2} \\
(j_\nu, Z_\nu) \vdash' e_\nu. \label{eq-amst-3}
\end{gather}
Set
\[
\mathfrak{V} = \{\, ((i_\nu, Y_\nu), d) \mid d \in Z_\nu \wedge 1 \le \nu \le n \,\}.
\]
Then $\mathfrak{V} \subseteq H$. We will next construct $\mathfrak{W} \in \ac(\mathfrak{V})$ such that $\pair{\mathfrak{W} | \mathfrak{V}} \in \St(H)$ and $(\pair{\mathfrak{W} \ \mathfrak{V}}, \{ \pair{\mathfrak{W} \ \mathfrak{V}} \}) \vdash_\rightarrow \pair{\mathcal{G} | \mathcal{F}}$. If $\bigcup\mathcal{F}$ is empty, set $\pair{\mathfrak{W} | \mathfrak{V}} = \Delta_\rightarrow$. 

Let us now assume that $\bigcup\mathcal{F}$ is non-empty and that, for each $J \subseteq \{ 1, \ldots, n \}$, $R(\{\, (i_\nu, Y_\nu) \mid \nu \in J \,\})$ is a system of representatives of $W(\{\,(i_\nu, Y_\nu) \mid \nu \in J\,\})$ with respect to $\sim [\bigcup\nolimits_{\nu \in J} Y_\nu]$ so that $R(\emptyset) = \{ \Delta \}$ und $R(\{ (i_\nu, Y_\nu)\}) = \{ i_\nu\}$, for $1 \le \nu \le n$. By Condition~\ref{lem-approx-state}(\ref{lem-approx-state-2}) there is some $j_{a, Z, \pair{\mathfrak{G} |\mathfrak{F}}} \in A'$, for each $\pair{\mathfrak{G} | \mathfrak{F}} \in \pair{\mathcal{G} | \mathcal{F}}$ and $((a, Z), b) \in  \mathfrak{G}$, with $j_{a, Z, \pair{\mathfrak{G} | \mathfrak{F}}} \sim b \,[\rs((a, Z), \mathfrak{F})]$ and $(a, Z) H j_{a, Z, \pair{\mathfrak{G} | \mathfrak{F}}}$. In case  $((c, Z), b) = ((\Delta, \emptyset), \Delta')$, choose $j_{\Delta, \emptyset, \pair{\mathfrak{G} |\mathfrak{F}}} = \Delta'$.
 
 \begin{claim}\, \label{cl-amcl}
 For every $\mathfrak{J} \subseteq \mathfrak{V}$ and each $a \in R(\pr_1(\mathfrak{J}))$ so that $(a, \mathfrak{J})$ is $\mathfrak{V}$-maximal there is some $t_{a, \mathfrak{J}} \in A'$ with the following two properties:
\begin{align}\label{eq-amst-4}
\{ \Delta' &\} \cup \bigcup \{\, \{ k_\nu \} \cup Z_\nu \mid 1 \le \nu \le n \wedge (i_\nu, Y_\nu) \in \pr_1(\mathfrak{J}) \,\} \cup \mbox{} \notag \\
&\{\, t_{c, \mathfrak{K}} \mid \mathfrak{K} \subsetneqq \mathfrak{J} \wedge c \sim a \,[\bigcup \pr_{2, 1}(\mathfrak{K})] \wedge \text{$(c, \mathfrak{K})$ $\mathfrak{V}$-maximal} \,\} \cup \{\, j_{c, Z, \pair{\mathfrak{G} | \mathfrak{F}}} \mid \notag \\
&  (\exists L \subseteq \{ 1, \ldots, n \})(\exists b \in A') \{\, (i_\mu, Y_\mu) \mid \mu \in L \,\} \subseteq \pr_1(\mathfrak{J}) \wedge Z = \notag \\
&  \bigcup\nolimits_{\mu \in L} T_\mu \wedge c \sim a \,[Z] \wedge \pair{\mathfrak{G} | \mathfrak{F}} \in \pair{\mathcal{G} | \mathcal{F}} \wedge ((c, Z), b) \in \mathfrak{G} \wedge \mbox{} \notag \\
&b \sim j_{c, Z, \pair{\mathfrak{G} | \mathfrak{F}}} \,[\rs((c, Z), \mathfrak{F})] \,\} \in \Con'(t_{a, \mathfrak{J}}),
\end{align}
\begin{align}\label{eq-amst-5}
(a, &\bigcup \pr_{2, 1}(\mathfrak{J})) H (t_{a, \mathfrak{J}}, \{ \Delta' \} \cup \bigcup \{\, \{ k_\nu \} \cup Z_\nu \mid 1 \le \nu \le n \wedge (i_\nu, Y_\nu) \in \notag \\
& \pr_1(\mathfrak{J}) \,\} \cup 
\{\, t_{c, \mathfrak{K}} \mid \mathfrak{K} \subsetneqq \mathfrak{J} \wedge c \sim a \,[\bigcup \pr_{2, 1}(\mathfrak{K})] \wedge \text{$(c, \mathfrak{K})$ $\mathfrak{V}$-maximal} \,\} \cup \mbox{}  \notag \\
& \{\, j_{c, Z, \pair{\mathfrak{G} | \mathfrak{F}}} \mid (\exists L \subseteq \{ 1, \ldots, n \})(\exists b \in A') \{\, (i_\mu, Y_\mu) \mid \mu \in L \,\} \subseteq \pr_1(\mathfrak{J}) \wedge \mbox{} \notag \\
& Z = \bigcup\nolimits_{\mu \in L} T_\mu \wedge c \sim a \,[Z] \wedge \pair{\mathfrak{G} | \mathfrak{F}} \in \pair{\mathcal{G} | \mathcal{F}} \wedge ((c, Z), b) \in \mathfrak{G} \wedge \mbox{} \notag \\
& b \sim j_{c, Z, \pair{\mathfrak{G} | \mathfrak{F}}} \,[\rs((c, Z), \mathfrak{F})] \,\}).
\end{align}
\end{claim}

With respect to the last set on the right hand side note that if $((c, Z), b) \in \mathfrak{G}$, then there is some $\mathfrak{L} \subseteq \mathfrak{F}$ so that $(c, \mathfrak{L})$ is $\mathfrak{F}$-maximal and $Z= \bigcup \pr_{2,1}(\mathfrak{L})$. It follows that $\| \mathfrak{L} \| \le \| \mathfrak{J} \|$.

The claim is shown by induction on the cardinality $\kappa$ of the subset $\mathfrak{J}$. The case $\kappa = 0$ is obvious. Set $t_{a, \emptyset} = \Delta'$, if $(a, \emptyset)$ is $\mathfrak{V}$-maximal. All other sets on the left hand side in Statement~(\ref{eq-amst-4}) are empty in this case, except the last one, which is the singleton $\{ \Delta' \}$, if for some $ \pair{\mathfrak{G} | \mathfrak{F}} \in \pair{\mathcal{G} | \mathcal{F}}$, $((\Delta, \emptyset), \Delta') \in \mathfrak{G}$. 

Assume next that the claim holds for all subsets $\mathfrak{K} \subseteq \mathfrak{V}$ of cardinality $\kappa$ and let $\mathfrak{J} \subseteq \mathfrak{V}$ with $\| \mathfrak{J} \| = \kappa + 1$. Then 
\begin{align*}
(a, &\bigcup \pr_{2, 1}(\mathfrak{J})) H \{ \Delta' \} \cup \bigcup \{\, \{ k_\nu \} \cup Z_\nu \mid 1 \le \nu \le n \wedge (i_\nu, Y_\nu) \in \pr_1(\mathfrak{J}) \,\} \cup \mbox{} \notag \\
&\{\, t_{c, \mathfrak{K}} \mid \mathfrak{K} \subsetneqq \mathfrak{J} \wedge c \sim a \,[\bigcup \pr_{2, 1}(\mathfrak{K})] \wedge \text{$(c, \mathfrak{K})$ $\mathfrak{V}$-maximal} \,\} \cup \{\, j_{c, Z, \pair{\mathfrak{G} | \mathfrak{F}}} \mid  \notag \\
& (\exists L \subseteq \{ 1, \ldots, n \})(\exists b \in A') \{\, (i_\mu, Y_\mu) \mid \mu \in L \,\} \subseteq \pr_1(\mathfrak{J}) \wedge Z = \notag \\
& \bigcup\nolimits_{\mu \in L} T_\mu \wedge c \sim a \,[Z] \wedge \pair{\mathfrak{G} | \mathfrak{F}} \in \pair{\mathcal{G} | \mathcal{F}} \wedge ((c, Z), b) \in \mathfrak{G} \wedge \mbox{} \notag \\
 &b \sim j_{c, Z, \pair{\mathfrak{G} | \mathfrak{F}}} \,[\rs((c, Z), \mathfrak{F})] \,\},
\end{align*}
because of Statement~(\ref{eq-amst-2}), Condition~\ref{lem-approx-state}(\ref{lem-approx-state-2}), and the induction hypothesis. By Condition~\ref{dn-am}(\ref{dn-am-8}) there is hence some $t_{a, \mathfrak{J}} \in A'$ so that Properties~(\ref{eq-amst-4}) and (\ref{eq-amst-5}) hold.

With help of Claim~\ref{cl-amcl} we can now define $\mathfrak{W}$. Let $\mathfrak{W}^{(0)}$ be as in Definition~\ref{dn-ass} and for $\kappa \ge 1$ set
\begin{equation*}
\widehat{\mathfrak{W}}^{(\kappa)} = \{\, ((a, \bigcup \pr_{2, 1}(\mathfrak{J}), t_{a, \mathfrak{J}}) \mid \mathfrak{J} \subseteq \mathfrak{V} \wedge \mbox{} \\
 \| \mathfrak{J} \| = \kappa \wedge a \in R(\pr_1(\mathfrak{J})) \wedge \text{$(a, \mathfrak{J})$ $\mathfrak{V}$-maximal} \,\}. 
\end{equation*}
Obviously, $\mathfrak{W} \in \ac(\mathfrak{V})$. 

\begin{claim}\,
$\pair{\mathfrak{W} | \mathfrak{V}} \in St(H)$.
\end{claim}

By definition, $\mathfrak{V} \subseteq H$. For Requirement~\ref{lem-approx-state}(\ref{lem-approx-state-2}) let $(a, Z), b) \in \mathfrak{W}$. Then there is some subset $\mathfrak{J} \subseteq \mathfrak{V}$ so that $(a, \mathfrak{J})$ is $\mathfrak{V}$-maximal, $a \in R(\pr_1(\mathfrak{J}))$, $Z = \bigcup \pr_{2,1}(\mathfrak{J})$ and $b = t_{a, \mathfrak{J}}$. Hence, $\pr_2(\mathfrak{J}) \in \Con'(b)$ and $(a, Z) H b$. Note that $\pr_2(\mathfrak{J}) = \rs((a, Z), \mathfrak{V})]$. Thus, $b \sim b \,[\rs((a, Z), \mathfrak{J})]$.

\begin{claim}\,
$(\pair{\mathfrak{W} | \mathfrak{V}}, \{ \pair{\mathfrak{W} | \mathfrak{V}} \}) \vdash_\rightarrow \pair{\mathcal{G} | \mathcal{F}}$.
\end{claim}

Let $1 \le \nu \le n$, Then $Z_\nu \subseteq \ap((c_\nu, T_\nu), \mathfrak{V})$ and $\{ k_\nu \} \in \Con'(\mathfrak{W}_2(c_\nu, T_\nu))$. As a consequence of Statement~(\ref{eq-amst-3}) we therefore have that $(\mathfrak{W}_2(c_\nu, T_\nu), \ap((c_\nu, T_\nu), \mathfrak{V})) \vdash' e_\nu$, which means that Condition~\ref{dn-fctsp}(\ref{dn-fctsp-3-1}) holds.

For Condition~\ref{dn-fctsp}(\ref{dn-fctsp-3-2}), let $\pair{\mathfrak{G} | \mathfrak{F}} \in \pair{\mathcal{G} | \mathcal{F}}$ and $((a, Z), b) \in \mathfrak{G}$. As a consequence of Condition~\ref{dn-fctsp}(\ref{dn-fctsp-3-1}) we have that 
$
(\mathfrak{W}_2(a, Z), \ap((a, Z), \mathfrak{V})) \vdash' \rs((a, Z), \mathfrak{F}).
$
Because of Condition~\ref{dn-infsys}(\ref{dn-infsys-11}) there is thus some $k \in A'$ so that $(\mathfrak{W}_2(a, Z), \ap((a, Z), \mathfrak{V})) \vdash' k$ and  
\begin{equation} \label{eq-amst-6}
\rs((a, Z), \mathfrak{F}) \in \Con'(k).
\end{equation} 
It follows that 
\begin{equation}\label{eq-amst-7}
\{ k \} \in \Con'(\mathfrak{W}_2(a, Z)).
\end{equation}

Let $\mathfrak{J} = \{\, ((i, Y), d) \in \mathfrak{V} \mid (i, Y) \in \Cl((a, Z), \pr_1(\mathfrak{V})) \,\}$ and $\mathfrak{W}(a, Z) = ((e, U), t)$. Then $(e, \mathfrak{J})$ is $\mathfrak{V}$-maximal, $e \sim a \,[U]$, $U = \bigcup\pr_{2,1}(\mathfrak{J})$, and $\mathfrak{W}_2(a, Z) = t_{e, \mathfrak{J}}$. Since $((a, Z), b) \in \mathfrak{G}$, there is some $\mathfrak{K} \subseteq \mathfrak{F}$ such that $Z = \bigcup \pr_{2,1}(\mathfrak{K})$ and $\{\, (i_\mu, Y_\mu) \mid 1 \le \mu \le n \wedge (c_\mu, T_\mu) \in \pr_1(\mathfrak{K}) \,\} \subseteq \pr_1(\mathfrak{J})$. It follows that 
\begin{equation}\label{eq-amst-8}
\{ j_{a, Z, \pair{\mathfrak{G} | \mathfrak{F}}} \} \in \Con'(\mathfrak{W}_2(a, Z)).
\end{equation}
Moreover, 
\begin{equation}\label{eq-amst-9}
b \sim  j_{a, Z, \pair{\mathfrak{G} | \mathfrak{F}}} \,[\rs((a, Z), \mathfrak{F})],
\end{equation}
since $\pair{\mathfrak{G} | \mathfrak{F}} \in \St(H)$, by assumption. As a consequence of Statements~(\ref{eq-amst-6}--\ref{eq-amst-9}) we obtain that $b \sim k \,[\rs((a, Z), \mathfrak{F})]$, as was be shown.
\end{proof}

In the next two lemmas we verify that the operators $\am$ and $\St$ between $| A \rightarrow A' |$ and the set of all approximable mappings between $A$ and $A'$ are inverse to each other.

\begin{lemma}\label{lem-amtost}
For $f \in | A \rightarrow A' |$, $\St(\am(f)) = f$.
\end{lemma}
\begin{proof}
Let $\pair{\mathfrak{B} | \mathfrak{A}} \in f$. Then $\mathfrak{A} \subseteq \bigcup \pr_2(f)$. Thus, $\mathfrak{A} \subseteq \am(f)$, i.e., Condition~\ref{lem-approx-state}(\ref{lem-approx-state-1}) holds.

For Condition~\ref{lem-approx-state}(\ref{lem-approx-state-2}), let $((a, Z), b) \in \mathfrak{B}$. Since $\pair{\mathfrak{B} | \mathfrak{A}} \in f$, there are $\pair{\mathfrak{W} | \mathfrak{V}} \in f$ and $\pair{\mathcal{W} | \mathcal{V}} \fsubset f$ with $\pair{\mathcal{W} | \mathcal{V}} \in \Con_\rightarrow(\pair{\mathfrak{W} | \mathfrak{V}})$ and $(\pair{\mathfrak{W} | \mathfrak{V}}, \pair{\mathcal{W} | \mathcal{V}}) \vdash_\rightarrow \pair{\mathfrak{B} | \mathfrak{A}}$. It follows by Condition~\ref{dn-fctsp}(\ref{dn-fctsp-3-2}) that there is some $k \in A'$ such that 
\begin{equation}\label{eq-amtost-1}
k \sim b \,[\rs((a, Z), \mathfrak{A})] \,\,\text{and}\,\, (\mathfrak{W}_2(a, Z), \ap((a, Z), \bigcup \mathcal{V})) \vdash' (k, \rs((a, Z), \mathfrak{A})).
\end{equation}
Let $((i, X), c) \in \bigcup \mathcal{V}$ such that $(i, X) \in \Cl((a, Z), \pr_1(\bigcup \mathcal{V}))$. Then $(i, X) \am(f) c$, as $\mathcal{V} \subseteq \pr_2(f)$ and hence $\bigcup \mathcal{V} \subseteq \bigcup\pr_2(\mathcal{V}) = \am(f)$. So, we have that 
\begin{equation}\label{eq-amtost-2}
(a, Z) \am(f) \ap((a, Z), \bigcup \mathcal{V}).
\end{equation}
By Condition~\ref{dn-am}(\ref{dn-am-8}) there is now some $j \in A'$ such that 
\begin{equation}\label{eq-amtost-3}
(a, Z) \am(f) j
\end{equation}
 and $\ap((a, Z), \bigcup\mathcal{V}) \in \Con'(j)$. Furthermore, with Lemma~\ref{pn-amleft}, we obtain some some $(c, U) \in \Con$ such that $(a, Z) \vdash (c, U)$ and $(c, U) \am(f) j$. It follows that there is some $\fcbael{E}{D} \in f$ with $((c, U), j) \in \mathfrak{D}$. Moreover, since $f$ is a state, there exists $\fcbael{G}{F} \in f$ with $\{ \fcbael{W}{V}, \fcbael{E}{D} \} \Con_\rightarrow(\fcbael{G}{F})$. Because of Condition~\ref{dn-fctsp}(\ref{dn-fctsp-2-1}) there are thus $e_{\fcbael{W}{V}}, e_{\fcbael{E}{D}} \in A'$ so that
\begin{gather}
\mathfrak{W}_2(a, Z) \sim e_{\fcbael{W}{V}} \,[\ap((a, Z), \mathfrak{V})] \,\,\text{and}\,\, \{ e_{\fcbael{W}{V}} \} \in \Con'(\mathfrak{G}_2(a, Z)) \label{eq-amtost-3-1} \\
\mathfrak{E}_2(a, Z) \sim e_{\fcbael{E}{D}} \,[\ap((a, Z), \mathfrak{D}] \,\,\text{and}\,\, \{ e_{\fcbael{E}{D}} \} \in \Con'(\mathfrak{G}_2(a, Z)) \label{eq-amtost-3-2}
\end{gather}

Note that by construction $j \in \ap((a, Z), \mathfrak{D})$. With Lemma~\ref{lem-witeqprop}(\ref{lem-witeqprop-2}) we hence obtain from the first of Statements~(\ref{eq-amtost-3-2}) that $\mathfrak{E}_2(a, Z) \sim e_{\fcbael{E}{D}} \,[\{ j \}]$. Since $\ap((a, Z), \mathfrak \bigcup\mathcal{V}) \in \Con'(j)$, it follows that $\mathfrak{E}_2(a, Z) \sim e_{\fcbael{E}{D}} \,[\ap((a, Z), \bigcup\mathcal{V})]$.

With Axiom~\ref{dn-fctsp}(\ref{dn-fctsp-2-2}) we obtain from the first of Statements~(\ref{eq-amtost-3-1}) that also 
\[
\mathfrak{W}_2(a, Z) \sim e_{\fcbael{W}{V}} \,[\ap((a, Z), \bigcup\mathcal{V})].
\]
  So, since $\{ e_{\fcbael{E}{D}} \} , \{ e_{\fcbael{W}{V}} \} \in \Con'(\mathfrak{G}_2(a, Z))$, we have that 
  \[
  \mathfrak{E}_2(a, Z) \sim \mathfrak{W}_2(a, Z) \, [\ap((a, Z), \bigcup \mathcal{V})].
  \]
    Finally, as $\{ j \} \in \Con'(\mathfrak{E}_2(a, Z))$ by Axiom~\ref{dn-ass}(\ref{dn-ass-3-1}), we find that 
 \[
    j \sim \mathfrak{W}_2(a, Z) \, \ap((a, Z),[\bigcup\mathcal{V})].
 \]
    
Because of Statement~(\ref{eq-amtost-1}) it follows  that also
\begin{equation}\label{eq-amtost-4}
(j, \ap((a, Z), \bigcup \mathcal{V})) \vdash' (k, \rs((a, Z), \mathfrak{A})).
\end{equation}
Statements~(\ref{eq-amtost-2}, \ref{eq-amtost-3}, \ref{eq-amtost-4}) now imply that $(a, Z) \am(f) k$. Moreover, if follows with Statement~(\ref{eq-amtost-4}) that $\{ k \} \in \Con'(j)$. With Statement~(\ref{eq-amtost-1}) we thus have that $j \sim b \,[\rs((a, Z), \mathfrak{A})]$. This shows that $\pair{\mathfrak{B} | \mathfrak{A}} \in \St(\am(f))$.

For the converse inclusion let $\pair{\mathfrak{B} | \mathfrak{A}} \in \St(\am(f))$. Then $\mathfrak{A} \subseteq \am(f)$. 
Let $\mathfrak{A} = \{\, ((c_\nu, T_\nu), e_\nu) \mid 1 \le \nu \le n \,\}$. As $((c_\nu, T_\nu), e_\nu) \in \am(f)$, there is some $\pair{\mathfrak{E}^{(\nu)} | \mathfrak{D}^{(\nu)}} \in f$ with $((c_\nu, T_\nu), e_\nu) \in \mathfrak{D}^{(\nu)}$, for each $1 \le \nu \le n$. Because of Condition~(\ref{st}) there are furthermore some $\pair{\mathfrak{W} | \mathfrak{V}} \in f$ and a subset $\pair{\mathcal{W} | \mathcal{V}} \fsubset f$ with $(\pair{\mathfrak{W} | \mathfrak{V}}, \pair{\mathcal{W} | \mathcal{V}}) \vdash_\rightarrow \{\, \pair{\mathfrak{E}^{(\nu)} | \mathfrak{D}^{(\nu)}} \mid 1 \le \nu \le n \,\}$. Hence, we have for every $1 \le \nu \le n$ that 
\begin{equation}\label{eq-amtost-5}
(\mathfrak{W}_2(c_\nu, T_\nu), \ap((c_\nu, T_\nu), \bigcup \mathcal{V})) \vdash' e_\nu.
\end{equation}
Let $((a, Z), b) \in \mathfrak{B}$. Then it follows that
\[
(\mathfrak{W}_2(a, Z), \bigcup \{\, \ap((c, T), \bigcup \mathcal{V}) \mid (c, T) \in \ds((a, Z), \mathfrak{A}) \,\}) \vdash' \rs((a, Z), \mathfrak{A}).
\]
Because of Condition~\ref{dn-infsys}(\ref{dn-infsys-11}) there is  thus some $k \in A'$ with $\rs((a, Z), \mathfrak{A}) \in \Con'(k)$ and
\begin{equation*}
(\mathfrak{W}_2(a, Z), \bigcup \{\, \ap((c, T), \bigcup \mathcal{V}) \mid (c, T) \in \ds((a, Z), \mathfrak{A}) \,\}) \vdash' 
(k, \rs((a, Z), \mathfrak{A})).
\end{equation*}
Hence, $\{ k \} \in \Con'(\mathfrak{W}_2(a, Z))$.  By Requirement~\ref{lem-approx-state}(\ref{lem-approx-state-2}) we furthermore obtain some $j \in A'$ with $j \sim b \,[\rs((a, Z), \mathfrak{A})]$ and $(a, Z) \am(f) j$. As above, it follows that 
\[
(\mathfrak{W}_2(a, Z), \ap((a, Z), \bigcup \mathcal{V})) \vdash' j.
\]
Thus, $\{ j \} \in \Con'(\mathfrak{W}_2(a, Z))$ as well. Consequently, $b \sim k \,[\rs((a, Z), \mathfrak{A})]$. With Statement~(\ref{eq-amtost-5}) we now obtain that 
\[
(\pair{\mathfrak{W} | \mathfrak{V}}, \pair{\mathcal{W} | \mathcal{V}}) \vdash_\rightarrow \pair{\mathfrak{B} | \mathfrak{A}}.
\]
Therefore, $\pair{\mathfrak{B} | \mathfrak{A}} \in f$.
\end{proof}

 \begin{lemma}\label{lem-sttoam}
For $\apmap{H}{A}{A'}$, $\am(\St(H)) = H$.
\end{lemma}
\begin{proof}
Let $((a, S), b) \in \am(\St(H))$. Then $((a, S), b) \in \bigcup \pr_2(\St(H))$. Thus, there is some $\pair{\mathfrak{B} | \mathfrak{A}} \in \St(H)$ with $((a, S), b) \in \mathfrak{A}$. Since $\mathfrak{A} \subseteq H$, it follows that $((a, S), b) \in H$.

For the converse implication, let $((c, T), e) \in H$. As $(\Delta, \emptyset) H \Delta'$, we have that 
\[
(c, T) H \{ \Delta', e \}.
\]
  Hence, there is some $d \in A'$ with $(c, T) H d$ and $\{ \Delta', e \} \in \Con'(d)$. Set $\mathfrak{D} = \{ ((c, T), e) \}$ and let $\mathfrak{E} \in \ac(\mathfrak{D})$ be as in Lemma~\ref{lem-singext}. Then $\mathfrak{D} \subseteq H$. Next, let $((a, Z), b) \in \mathfrak{E}$. Then either $((a, Z), b) = ((c, T), d)$ or $((a, Z), b) = ((\Delta, \emptyset), \Delta')$. In the latter case $T$ is non-empty. Therefore $\rs((a, Z), \mathfrak{D}) = \emptyset$. Choose $j = \Delta'$ in this case. In the other case $\rs((a, Z), \mathfrak{D}) = \{ d \}$. Now, choose $j = d$. Then we have in both cases that $j \sim b \,[\rs((a, Z), \mathfrak{D})]$ and $(a, Z) H j$. This shows that $\pair{\mathfrak{E} | \mathfrak{D}} \in \St(H)$. Thus, $\mathfrak{D} \subseteq \bigcup \pr_2(\St(H))$, which means that $((c, T), e) \in \am(\St(H))$.
\end{proof}

\begin{proposition}\label{pn-arrowst-am}
Let $(A, \Con, \vdash, \Delta)$ and $(A', \Con', \vdash', \Delta')$ be information systems with witnesses. Then the states of $A \rightarrow A'$ are in a one-to-one correspondence to the approximable mappings between $A$ and $A'$:
\begin{enumerate}
\item \label{pn-arrowst-am-1}
$\set{\bigcup\pr_2(f)}{f \in | A \rightarrow A'|}$ is the set of all approximable mappings between $A$ and $A'$.

\item \label{pn-arrowst-am-2}
$| A \rightarrow A' |$ is the collection of all sets $\set{\fcbael{B}{A} \in A \rightarrow A'}{\mathfrak{A} \subseteq H \wedge (\forall ((a, Z), b) \in \mathfrak{B})  (\exists j \in A') (a, Z) H j \wedge j \sim b \,[\rs((a, Z), \mathfrak{A})]}$, where $H$ is an approximable mapping $H$ between $A$ and $A'$.

\end{enumerate}
\end{proposition}

In Ref.~\cite{sp21} it was shown how approximable mappings between two information systems with witnesses $A$ and $A'$ as well as  Scott continuous functions from $|A|$ to $|A'|$ correspond to each other. As we will show next, this correspondence establishes an isomorphism between the domains $|A \rightarrow A'|$ and $[|A| \rightarrow |A'|]$.

For $\apmap{G}{A}{A'}$ and $x \in |A|$  let 
\[
\mathcal{L}(G)(x) = \set{b \in A'}{(\exists (i, X) \in \Con) \{ i \} \cup X \subseteq x \wedge (i, X) G b}.
\]
Then $\mathcal{L}(G) \in [|A| \rightarrow |A'|]$. Since for $f \in |A \rightarrow A'|$,  $\apmap{\am(f)}{A}{A'}$, it follows that $\mathcal{L}(\am(f)) \in [|A| \rightarrow |A'|]$. Set
\[
\fct(f) = \mathcal{L}(\am(f)).
\]
Then we have for $x \in |A|$ that
\begin{align*}
&\fct(f)(x)\\
&\hspace{3mm}= \set{b \in A'}{(\exists (i, X) \in \Con) \{ i \} \cup X \subseteq x \wedge (i, X) \am(f) b} \\
&\hspace{3mm}= \set{b \in A'}{(\exists (i, X) \in \Con) (\exists \fcbael{B}{A} \in f) \{ i \} \cup X \subseteq x \wedge 
((i, X), b) \in \mathfrak{A}} \\
&\hspace{3mm}= \set{b \in A'}{(\exists (i, X) \in \Con) (\exists \fcbael{B}{A}, \fcbael{W}{V} \in A \rightarrow A') \\ 
& \hspace{1cm} (\exists \fcbaset{W}{V} \in \Con_\rightarrow(\fcbael{W}{V})) 
 \{ \fcbael{W}{V} \} \cup  \fcbaset{W}{V} \subseteq f  \wedge \{ i \} \cup X \subseteq x \wedge \mbox{} \\
& \hspace{1cm} (\fcbael{W}{V}, \fcbaset{W}{V}) \vdash_\rightarrow \fcbael{B}{A} \wedge ((i, X), b) \in \mathfrak{A}} \\
&\hspace{3mm}= \set{b \in A'}{(\exists (i, X) \in \Con)  (\exists \fcbael{B}{A}, \fcbael{W}{V} \in A \rightarrow A') \\ 
& \hspace{1cm} (\exists \fcbaset{W}{V} \in \Con_\rightarrow(\fcbael{W}{V})) 
\{ \fcbael{W}{V} \} \cup  \fcbaset{W}{V} \subseteq f  \wedge \{ i \} \cup X \subseteq x \wedge \mbox{} \\
& \hspace{1cm}  (\mathfrak{B}_2(i, X), \ap((i, X), \bigcup \mathcal{V})) \vdash' b},
\end{align*}
where the last equality follows as in the second part of the proof of Lemma~\ref{lem-sttoam}. As is now easily seen,  $\fct \in [|A \rightarrow A'| \rightarrow [|A| \rightarrow |A'|]]$.

Conversely, let $g \in [|A| \rightarrow |A'|]$ and for $(i, X) \in \Con$ and $b \in A'$ define
\[
(i, X) H^g b \Longleftrightarrow b \in g([X]_i).
\]
Then $\apmap{H^g}{A}{A'}$. Set
\[
\st(g) = \St(H^g).
\]
By Lemma~\ref{lem-approx-state}, $\st(g) \in |A \rightarrow A'|$.  

\begin{lemma}\label{lem-st-cont}
The function $\fun{\st}{[|A| \rightarrow |A'|]}{|A \rightarrow A'|}$ is Scott continuous.
\end{lemma}
\begin{proof}
Obviously, $\st$ is monotone. Let $\mathbb{G} \subseteq [|A| \rightarrow |A'|]$ be directed and $\fcbael{B}{A} \in \st(\bigsqcup \mathbb{G})$. Then $\mathfrak{A} \subseteq H^{\bigsqcup\mathbb{G}}$. Since $\mathfrak{A}$ is finite and $\mathbb{G}$ directed, we gain some $g \in \mathbb{G}$ so that $\mathfrak{A} \subseteq H^g$. 

Now, let $((a, Z), b) \in \mathfrak{B}$. Then $(a, Z) H^g \rs((a, Z), \mathfrak{A})$. Because of Condition~\ref{dn-am}(\ref{dn-am-7}) there is hence some $j \in A'$ with $(a, Z) H^g j$ and $\rs((a, Z), \mathfrak{A}) \in \Con'(j)$. It follows that $(a, Z) H^{\bigsqcup\mathbb{G}} j$ as well. Since $\fcbael{B}{A} \in \st(\bigsqcup \mathbb{G})$, there is also some $k \in A'$ with $(a, Z) H^{\bigsqcup\mathbb{G}} k$ and $k \sim b \,[\rs((a, Z), \mathfrak{A})]$. Again by Condition~\ref{dn-am}(\ref{dn-am-7}) there is some $r \in A'$ with $(a, Z) H^{\bigsqcup\mathbb{G}} r$ and $\{ j, k \} \in \Con'(r)$. Consequently, $j \sim b \,[\rs((a, Z), \mathfrak{A})]$, which shows that $\fcbael{B}{A} \in \st(g)$. The converse inclusion follows by monotonicity.
\end{proof}

As shown in Ref.~\cite{sp21}, $\mathcal{L}(H^g) = g$ and $H^{\mathcal{L}(G)} = G$. With Lemmas~\ref{lem-amtost} and \ref{lem-sttoam} we therefore obtain that the two functions $\fct$ and $\st$ are inverse to each other.

\begin{proposition}\label{pn-domiso}
Let $(A, \Con, \vdash, \Delta)$ and $(A', \Con', \vdash', \Delta')$ be information systems with witnesses. Then the domains $|A \rightarrow A'|$ and $[|A| \rightarrow |A'|]$ are isomorphic.
\end{proposition}

\section{Cartesian closure}\label{sec-cc}

As was shown in Ref.~\cite{sp21}, the category $\mathbf{ISW}$, as well as its full subcategories $\mathbf{aISW}$, $\mathbf{bcISW}$ and $\mathbf{abcISW}$, possess a terminal object and are closed under taking finite products.

The one-point information system with witnesses $T = (\{ \Delta \}, \Con_T, \vdash_T, \Delta)$ with $\Con_T = \{ (\Delta, \emptyset), (\Delta, \{ \Delta \}))$ and \mbox{$\vdash_T$} = $\Con_T \times \{\Delta \}$ is a terminal object.

For $\nu = 1, 2$, let $(A_\nu, \Con_\nu, \vdash_\nu, \Delta_\nu)$ be an information system with witnesses. Set $A_\times = A_1 \times A_2$, $\Delta_\times = (\Delta_1, \Delta_2)$,
\begin{equation*}
\Con_\times = \{\, ((i, j), X) \in A_\times \times \mathcal{P}_f(A_\times) \mid \\
 \pr_1(X) \in \Con_1(i) \wedge \pr_2(X) \in \Con_2(j) \,\},
\end{equation*}
and for $((i, j), X) \in \Con_\times$ and $(a_1, a_2) \in A_\times$ define
\[
((i, j), X) \vdash_\times (a_1, a_2) \Longleftrightarrow (i, \pr_1(X)) \vdash_1 a_1 \wedge (j, \pr_2(X)) \vdash_2 a_2.
\]
Moreover, let the relations $\Pr_\nu \subseteq \Con_\times \times A_\nu$, with $\nu = 1, 2$, be given by
\[
((i_1, i_2), X) \Pr\nolimits_\nu a_\nu \Longleftrightarrow (i_\nu, \pr_\nu(X)) \vdash_\nu a_\nu.
\]
Then $(A_\times, \Pr_1, \Pr_2)$ is the categorical product of $A_1$ and $A_2$.

The aim of this section is to show for information systems with witnesses $A$ and $A'$ that $| A \rightarrow A' |$ is the exponent of $A$ and $A'$ in the category $\mathbf{ISW}$.
For $a \in A$, $\pair{\mathfrak{B} | \mathfrak{A}} \in A \rightarrow A'$, $\mathbb{Z} \in \Con_{(A \rightarrow A') \times A}(\pair{\mathfrak{B} | \mathfrak{A}}, a)$, and $b \in A'$, let
\[
((\pair{\mathfrak{B} | \mathfrak{A}}, a), \mathbb{Z}) \Ev b  \Longleftrightarrow
(\mathfrak{B}_2(a, \pr_2(\mathbb{Z})), \ap((a, \pr_2(\mathbb{Z})), \bigcup \pr_{2,1}(\mathbb{Z}))) \vdash' b.
\]

\begin{lemma}\label{lem-evalam} 
$\apmap{\Ev}{(A \rightarrow A') \times A}{A'}$.
\end{lemma}
\begin{proof}
We have to verify the conditions in Def.~\ref{dn-am}.
In case of Condition~\ref{dn-am}(\ref{dn-am-6}) we have to verify that 
\[
((\Delta_\rightarrow, \Delta), \emptyset) \Ev \Delta',
\]
 i.e., we have to check whether $(\Delta', \ap((\Delta, \emptyset), \emptyset)) \vdash' \Delta'$, which holds by Axiom~\ref{dn-infsys}(\ref{dn-infsys-3}). Note that $\ap((\Delta, \emptyset), \emptyset)$ is empty.

For Condition~\ref{dn-am}(\ref{dn-am-2}) let $\mathbb{Z}, \mathbb{Z}' \in \Con_{(A \rightarrow A') \times A}(\fcbael{B}{A}, a))$ with $\mathbb{Z} \subseteq \mathbb{Z}'$. Then 
\[
\Cl((a, \pr_2(\mathbb{Z}), \mathfrak{A})) \subseteq \Cl((a, \pr_2(\mathbb{Z}')), \mathfrak{A}).\]
 Thus, 
\[
\{ \mathfrak{B}_2(a, \pr_2(\mathbb{Z})) \} \in \Con'(\mathfrak{B}_2(a, \pr_2(\mathbb{Z}'))).
\]
 Moreover, 
 \[
 \ap((a, \pr_2(\mathbb{Z})), \bigcup \pr_{2,1}(\mathbb{Z})) \subseteq \ap((a, \pr_2(\mathbb{Z})), \bigcup \pr_{2,1}(\mathbb{Z})).
 \]
 Therefore, the condition follows with Axioms~\ref{dn-infsys}(\ref{dn-infsys-5}, \ref{dn-infsys-7}).

For Condition~\ref{dn-am}(\ref{dn-am-3}) assume that 
\[
((\pair{\mathfrak{B} | \mathfrak{A}}, a), \mathbb{Z}) \vdash_{(A \rightarrow A') \times A} \mathbb{T}
\]
 and $((\pair{\mathfrak{B} | \mathfrak{A}}, a), \mathbb{T}) \Ev b$. Then we have that
\begin{gather}
(\pair{\mathfrak{B} | \mathfrak{A}}, \pr_1(\mathbb{Z})) \vdash_\rightarrow \pr_1(\mathbb{T}) \label{eq-evalam-1} \\
(a, \pr_2(\mathbb{Z})) \vdash \pr_2(\mathbb{T}) \label{eq-evalam-2} \\
(\mathfrak{B}_2(a, \pr_2(\mathbb{T})), \ap((a, \pr_2(\mathbb{T})), \bigcup \pr_{2,1}(\mathbb{T}))) \vdash' b. \label{eq-evalam-3}
\end{gather}
As a consequence of Statement~(\ref{eq-evalam-2}), 
\[
\{ \mathfrak{B}_2(a, \pr_2(\mathbb{T})) \} \in \Con(\mathfrak{B}_2(a, \pr_2(\mathbb{Z})))
\]
 and hence
\[
(\mathfrak{B}_2(a, \pr_2(\mathbb{Z})), \ap((a, \pr_2(\mathbb{T})), \bigcup \pr_{2,1}(\mathbb{T}))) \vdash' b.
\]
Because of Statements~(\ref{eq-evalam-1}) and (\ref{eq-evalam-2}) it moreover follows that 
\[
\ap((a, \pr_2(\mathbb{T})), \bigcup \pr_{2,1}(\mathbb{T})) \subseteq \ap((a, \pr_2(\mathbb{Z})), \bigcup \pr_{2,1}(\mathbb{Z})).
\]
Therefore, $(\mathfrak{B}_2(a, \pr_2(\mathbb{Z})), \ap((a, \pr_2(\mathbb{Z})), \bigcup \pr_{2,1}(\mathbb{Z}))) \vdash' b$, i.e., 
\[((\pair{\mathfrak{B} | \mathfrak{A}}, a), \mathbb{Z}) \Ev b.
\]

For Conditions~\ref{dn-am}(\ref{dn-am-5-1}, \ref{dn-am-5-2}, \ref{dn-am-8}, \ref{dn-am-7}) we make use of Lemmas~\ref{pn-amleft}, \ref{pn-amright} and \ref{pn-amint}, and verify Requirement~(\ref{eq-amint}) instead. Suppose to this end that $((\fcbael{B}{A}, a), \mathbb{Z}) \Ev F$. Then
\[
(\mathfrak{B}_2(a, \pr_2(\mathbb{Z})), \ap((a, \pr_2(\mathbb{Z})), \bigcup \pr_{2,1}(\mathbb{Z}))) \vdash' F.
\]
Because of the Global Interpolation Property there is some $(j, Y) \in \Con'$ so that
\begin{gather}
(\mathfrak{B}_2(a, \pr_2(\mathbb{Z})), \ap((a, \pr_2(\mathbb{Z})), \bigcup \pr_{2,1}(\mathbb{Z}))) \vdash' (j, Y) \label{eq-evalam-6} \\
(j, Y) \vdash F. \label{eq-evalam-7}
\end{gather}
By definition we moreover have that 
\[
(a, \pr_2(\mathbb{Z})) \vdash \ucl((a, \pr_2(\mathbb{Z})), \bigcup \pr_{2,1}(\mathbb{Z})).
\]
Hence, there exists $(i, X) \in \Con$ with
\begin{gather}
(a, \pr_2(\mathbb{Z})) \vdash (i, X) \label{eq-evalam-8} \\
(i, X) \vdash \ucl((a, \pr_2(\mathbb{Z})), \bigcup \pr_{2,1}(\mathbb{Z})). \label{eq-evalam-9}
\end{gather}
It follows that $\Cl((i, X), \pr_1(\mathfrak{A})) \subseteq \Cl((a, \pr_2(\mathbb{Z})), \pr_1(\mathfrak{A}))$ and thus that $\{ \mathfrak{B}_2(i, X) \} \in \Con'(\mathfrak{B}_2(a, \pr_2(\mathbb{Z})))$. Moreover, $\ap((i, X), \bigcup \pr_{2,1}(\mathbb{Z})) = \ap((a, \pr_2(\mathbb{Z})), \bigcup \pr_{2,1}(\mathbb{Z}))$. So, we gain that
\begin{equation}\label{eq-evalam-10}
(\mathfrak{B}_2(i, X), \ap((i, X), \bigcup \pr_{2,1}(\mathbb{Z}))) \vdash' Y \cup \{ j \}.
\end{equation}
By Axiom~\ref{dn-infsys}(\ref{dn-infsys-11}) there is thus some $\bar{\jmath} \in A'$ such that 
\[
(\mathfrak{B}_2(i, X), \ap((i, X), \bigcup \pr_{2,1}(\mathbb{Z}))) \vdash'\bar{\jmath}
\]
 and $Y \cup \{ j, \Delta' \} \in \Con'(\bar{\jmath})$.

Let $\mathfrak{D}= \set{((i, X), e)}{e \in Y \cup \{ j \}}$ and $\mathfrak{E} \in \ac(\mathfrak{D})$ be as in Lemma~\ref{lem-singext}. Then we have for $((d, V), b) \in \mathfrak{E}$ that either $((d, V), b) = ((i, X), \bar{\jmath})$, $\ds((d, V), \mathfrak{D}) = \{ (i, X) \}$ and $\rs((d, V), \mathfrak{D}) = Y \cup \{ j \}$, or $((d, V), b) = ((\Delta, \emptyset), \Delta')$, $X$ is not empty, and $\ds((d, V), \mathfrak{D})$ and $\rs((d, V), \mathfrak{D}) $ are both empty. It follows that in either case
\begin{equation*}
(\mathfrak{B}_2(d, V), \bigcup \set{\ap((c, U), \bigcup \pr_{2,1}(\mathbb{Z}))}{(c, U) \in \ds((d, V), \mathfrak{D})}) \vdash'  
(b, \rs((c, U), \mathfrak{D})).
\end{equation*}
With Statement~(\ref{eq-evalam-10}) we hence obtain that $(\fcbael{B}{A}, \pr_1(\mathbb{Z})) \vdash_\rightarrow \fcbael{E}{D}$. By the Global Interpolation Property there are now $\fcbael{W}{V} \in A \rightarrow A'$ and $\fcbaset{W}{V} \in \Con_\rightarrow(\fcbael{W}{V})$ such that 
\[
(\fcbael{B}{A}, \pr_1(\mathbb{Z})) \vdash_\rightarrow (\fcbael{W}{V}, \fcbaset{W}{V}) \,\,\text{and}\,\, (\fcbael{W}{V}, \fcbaset{W}{V}) \vdash_\rightarrow \fcbael{E}{D}.
\]
In particular, we have that
\[
(\mathfrak{W}_2(i, X), \ap((i, X), \bigcup \mathcal{V})) \vdash' (j, Y).
\]

Set $\mathbb{T} = \fcbaset{W}{V} \times X$. Then 
\begin{gather*}
\mathbb{T} \in \Con_{(A \rightarrow A') \times A}(\fcbael{W}{V}, i), \\
((\fcbael{B}{A}, a), \mathbb{Z}) \vdash_{(A \rightarrow A') \times A} ((\fcbael{W}{V}, i), \mathbb{T}), \\
((\fcbael{W}{V}, i), \mathbb{T}) \Ev (j, Y) 
\intertext{and} 
(j, Y) \vdash' F.
\end{gather*}

 Condition~\ref{dn-am}(\ref{dn-am-1}) is a consequence of Lemma~\ref{lem-strong6}.

For Condition~\ref{dn-am}(\ref{dn-am-4}) let 
\[
\{ (\pair{\mathfrak{B} | \mathfrak{A}}, j) \} \in \Con_{(A \rightarrow A') \times A}(\pair{\mathfrak{W} | \mathfrak{V}}, i), 
\]
$\mathbb{Z} \in \Con_{(A \rightarrow A') \times A}(\pair{\mathfrak{B} | \mathfrak{A}}, j)$, and 
$((\pair{\mathfrak{B} | \mathfrak{A}}, j), \mathbb{Z}) \Ev b$.
Then we have that
\[
(\mathfrak{B}_2(j, \pr_2(\mathbb{Z})), \ap((j, \pr_2(\mathbb{Z})), \bigcup \pr_{2,1}(\mathbb{Z}))) \vdash' b,
\]
$\pr_1(\mathbb{Z}) \in \Con_\rightarrow(\pair{\mathfrak{B} | \mathfrak{A}})$, 
$\pr_2(\mathbb{Z}) \in \Con(j)$, 
$\{ \pair{\mathfrak{B} | \mathfrak{A}} \} \in \Con_\rightarrow(\pair{\mathfrak{W} | \mathfrak{V}})$, and
$\{ j \} \in \Con(i)$, from which it follows that $\pr_1(\mathbb{Z}) \in \Con_\rightarrow(\pair{\mathfrak{W} | \mathfrak{V}})$ and $\pr_2(\mathbb{Z}) \in \Con(i)$. Thus, $\mathbb{Z} \in \Con_{(A \rightarrow A') \times A}(\pair{\mathfrak{W} | \mathfrak{V}}, i)$.

We have to show that 
\begin{equation}\label{eq-evalam-4}
(\mathfrak{W}_2(i, \pr_2(\mathbb{Z})), \ap((i, \pr_2(\mathbb{Z})), \bigcup \pr_{2,1}(\mathbb{Z}))) \vdash' b.
\end{equation}
Note that because of Axioms~\ref{dn-infsys}(\ref{dn-infsys-8}, \ref{dn-infsys-9}), 
\[
\ap((j, \pr_2(\mathbb{Z})), \bigcup \pr_{2,1}(\mathbb{Z})) = \ap((i, \pr_2(\mathbb{Z})), \bigcup \pr_{2,1}(\mathbb{Z})).
\]
 Similarly,  $\Cl((j, \pr_2(\mathbb{Z})), \pr_1(\mathfrak{A})) = \Cl((i, \pr_2(\mathbb{Z})), \pr_1(\mathfrak{A}))$ which implies that 
\[
\mathfrak{B}_2(j, \pr_2(\mathbb{Z})) = \mathfrak{B}_2(i, \pr_2(\mathbb{Z})).
\]
Thus,
\[
(\mathfrak{B}_2(i, \pr_2(\mathbb{Z})), \ap((i, \pr_2(\mathbb{Z})), \bigcup \pr_{2,1}(\mathbb{Z}))) \vdash' b.
\]
Since $\fcbael{B}{A} \in \Con_\rightarrow(\fcbael{W}{V})$, there is some $ r \in A'$ with
\begin{equation}\label{eq-evalam-5}
r \sim \mathfrak{B}_2(i, \pr_2(\mathbb{Z})) \,[\ap((i, \pr_2(\mathbb{Z})), \mathfrak{A})]\,\,\text{and}\,\, 
\{ r \} \in \Con'(\mathfrak{W}_2(i, \pr_2(\mathbb{Z}))).
\end{equation}
By assumption $\pr_1(\mathbb{Z}) \in \Con_\rightarrow(\fcbael{B}{A})$. With Axiom~\ref{dn-fctsp}(\ref{dn-fctsp-2-2}) it follows that 
\[
r \sim \mathfrak{B}_2(i, \pr_2(\mathbb{Z})) \,[\ap((i, \pr_2(\mathbb{Z})), \bigcup \pr_{2,1}(\mathbb{Z}))]
\]
as well. Therefore also
\[
(r, \ap((i, \pr_2(\mathbb{Z})), \bigcup \pr_{2,1}(\mathbb{Z}))) \vdash' b.
\]
With Statements~(\ref{eq-evalam-5}), Statement~(\ref{eq-evalam-4}) is now a consequence.
\end{proof}

Let $(A'', \Con'', \vdash'', \Delta'')$ be a further information system with witnesses. 

\begin{definition}\label{dn-lambda}
For $\apmap{H}{A \times A'}{A''}$, $(a, S) \in \Con$ and $\fcbael{B}{A} \in A' \rightarrow A''$ define 
\[
(a, S) \Lambda(H) \fcbael{B}{A}
\]
 if
\begin{enumerate}
\item \label{dn-lambda-1}
$(\forall ((i, X), e) \in \mathfrak{A}) ((a, i), S \times X) H e$

\item \label{dn-lambda-2}
$(\forall ((c, T), b) \in \mathfrak{B}) (\exists j \in A'') ((a, c), S \times T) H j \wedge j \sim b \,[\rs((c, T), \mathfrak{A})]$.

\end{enumerate}
\end{definition}

\begin{lemma}\label{lem-lambam}
$\apmap{\Lambda(H)}{A}{A' \rightarrow A''}$.
\end{lemma}
\begin{proof}
Again we have to verify the conditions in Def.~\ref{dn-am}.

For Condition~\ref{dn-am}(\ref{dn-am-6}) we need to show that $(\Delta, \emptyset) \Lambda(H) \Delta_\rightarrow$. Condition~\ref{dn-lambda}(\ref{dn-lambda-1}) holds vacuously. In case of Condition~\ref{dn-lambda}(\ref{dn-lambda-2}) it only needs  be shown that $((\Delta, \Delta'), \emptyset) H \Delta''$, which is valid by Axiom~\ref{dn-am}(\ref{dn-am-6}). The remaining condition follows with Lemma~\ref{lem-witeqprop}(\ref{lem-witeqprop-0}).

Condition~\ref{dn-am}(\ref{dn-am-2}) is obviously satisfied, as it holds for $H$.

For Condition~\ref{dn-am}(\ref{dn-am-3}) assume that $(a, S) \vdash S'$. Moreover, let $(i, X) \in \Con'$ and $c \in A''$ such that $((a, i), S' \times X) H c$. By Lemma~\ref{pn-amleft} there is some $((\bar{a}, \bar{\imath}), U) \in \Con_{A \times A'}$ such that 
\[
((a, i), S' \times X) \vdash_{A \times A'} ((\bar{a}, \bar{\imath}), U)
\]
 and $((\bar{a}, \bar{\imath}), U) H c$. It follows that $(a, S') \vdash (\bar{a}, \pr_1(U))$. Hence, $(a, S) \vdash (\bar{a}, \pr_1(U))$ and thus, $((a, i), S \times X) \vdash_{A \times A'} ((\bar{a}, \bar{\imath}), U)$. Therefore, $((a, i), S \times X) H c$. It is now an easy consequence that from  $(a, S') \Lambda(H) \fcbael{E}{D}$ one obtains  $(a, S) \Lambda(H) \fcbael{E}{D}$.

As in the preceding proof, we verify Requirement~(\ref{eq-amint}) instead of  Conditions~\ref{dn-am}(\ref{dn-am-5-1}, \ref{dn-am-5-2}, \ref{dn-am-8}, \ref{dn-am-7}). Suppose that 
\[
(a, S) \Lambda(H) \fcbaset{G}{F}.
\]
We have to show that there exists $(\bar{\imath}, \overline{U}) \in \Con$ and $(\fcbael{B}{A}, \fcbaset{W}{V}) \in \Con_{A' \rightarrow A''}$ so that
\begin{gather*}
(a, S) \vdash (\bar{\imath}, \overline{U}) \label{eq-lambam-7b} \\
(\bar{\imath}, \overline{U}) \Lambda(H) (\fcbael{B}{A}, \fcbaset{W}{V}) \label{eq-lambam-8} \\
(\fcbael{B}{A}, \fcbaset{W}{V}) \vdash_{A' \rightarrow A''} \fcbaset{G}{F}. \label{eq-lambam-9}
\end{gather*}

If $\bigcup\mathcal{F}$ is empty, let $(\bar{\imath}, \overline{U}) = (\Delta, \emptyset)$, $\fcbael{B}{A} = \Delta_{A' \rightarrow A''}$ and $\fcbaset{W}{V} = \{ \Delta_{A' \rightarrow A''} \}$. Otherwise, assume that $\bigcup \mathcal{F} = \set{((i_\nu, X_\nu), e_\nu)}{1 \le \nu \le n}$. Then we have for all $1 \le \nu \le n$ that $((a, i_\nu), S \times X_\nu) H e_\nu$. Thus, for each such $\nu$ there are $((\bar{a}_\nu, \bar{\imath}_\nu), M_\nu) \in \Con_{A \times A'}$ and $(\hat{\jmath}_\nu, \widehat{N}_\nu) \in \Con''$ with
\begin{gather}
((a, i_\nu), S \times X_\nu) \vdash_{A \times A'} ((\bar{a}_\nu, \bar{\imath}_\nu), M_\nu) \label{eq-lambam-10} \\
((\bar{a}_\nu, \bar{\imath}_\nu), M_\nu) H (\hat{\jmath}_\nu, \widehat{N}_\nu) \label{eq-lambam-11} \\
(\hat{\jmath}_\nu, \widehat{N}_\nu) \vdash'' e_\nu \label{eq-lambam-12} 
\end{gather}
It follows that
\begin{equation*}\label{eq-lambam-13}
(a, S) \vdash \bigcup_{\nu = 1}^n \{ \bar{a}_\nu \} \cup \pr_1(M_\nu).
\end{equation*}
Thus, there is $(\bar{a}, \overline{M}) \in \Con$ with $(a, S) \vdash (\bar{a}, \overline{M})$ and $(\bar{a}, \overline{M}) \vdash \bigcup_{\nu = 1}^n \{ \bar{a}_\nu \} \cup \pr_1(M_\nu)$. Moreover, $((a, i_\nu), S \times X_\nu) \vdash_{A \times A'} (\bar{a}, \bar{\imath}_\nu)$.

By definition of product information systems we have that 
\[
\pr_1(M_\nu) \times \pr_2(M_\nu) \in \Con_{A \times A'}(\bar{a}, \bar{\imath}_\nu).
\]
 Since $M_\nu \subseteq \pr_1(M_\nu) \times \pr_2(M_\nu)$, it follows with Statement~(\ref{eq-lambam-11}) that 
\begin{equation}\label{eq-lambam-13+}
((\bar{a}, \bar{\imath}_\nu), \pr_1(M_\nu) \times \pr_2(M_\nu)) H (\hat{\jmath}_\nu, \widehat{N}_\nu).
\end{equation}

By Statement~(\ref{eq-lambam-10}) we have that $(i_\nu, X_\nu) \vdash' (\bar{\imath}_\nu, \pr_2(M_\nu))$. Thus, there is some $(\hat{\imath}_\nu, \widehat{M}_\nu) \in \Con'$ with $(i_\nu, X_\nu) \vdash' (\hat{\imath}_\nu, \widehat{M}_\nu)$ and $(\hat{\imath}_\nu, \widehat{M}_\nu) \vdash' (\bar{\imath}_\nu, \pr_2(M_\nu))$. Here is what we have obtained so far:
\begin{gather}
((a, i_\nu), S \times X_\nu) \vdash_{A \times A'} ((\bar{a}, \hat{\imath}_\nu), \overline{M} \times \widehat{M}_\nu) \label{eq-lambam-14} \\
((\bar{a}, \hat{\imath}_\nu), \overline{M} \times \widehat{M}_\nu) \vdash_{A \times A'} ((\bar{a}, \hat{\imath}_\nu), \pr_1(M_\nu) \times \pr_2(M_\nu)) \label{eq-lambam-15} 
\end{gather}
Since $\{ (\bar{a}, \bar{\imath}_\nu) \} \in \Con_{A \times A'}(\bar{a}, \hat{\imath}_\nu)$, it follows with Statement~(\ref{eq-lambam-13+}) that 
\[
((\bar{a}, \hat{\imath}_\nu), \pr_1(M_\nu) \times \pr_2(M_\nu)) H (\hat{\jmath}_\nu, \widehat{N}_\nu)
\]
 and hence because of Statement~(\ref{eq-lambam-15}) that
\begin{equation}\label{eq-lambam-16}
((\bar{a}, \hat{\imath}_\nu), \overline{M} \times \widehat{M}_\nu) H (\hat{\jmath}_\nu, \widehat{N}_\nu).
\end{equation}

Set 
\[
\mathfrak{A} = \set{((\hat{\imath}_\nu, \widehat{M}_\nu), \hat{d})}{\hat{d} \in \widehat{N}_\nu \wedge 1 \le \nu \le n}.
\]
Next, we will construct $\mathfrak{B} \in \ac(\mathfrak{A})$ so that $(\bar{a}, \overline{M}) \Lambda(H) (\fcbael{B}{A}, \{ \fcbael{B}{A} \})$ and 
\[
(\fcbael{B}{A}, \{ \fcbael{B}{A} \}) \vdash_{A' \rightarrow A''} \fcbaset{G}{F}.
\]

For each $J \subseteq \{ 1, \ldots, n \}$, let $R(\set{(\hat{\imath}_\nu, \widehat{M}_\nu)}{\nu \in J})$ be a system of representatives of $W(\set{(\hat{\imath}_\nu, \widehat{M}_\nu)}{\nu \in J})$ with respect to $\sim [\bigcup\nolimits_{\nu \in J} \widehat{M}_\nu]$ so that $R(\emptyset) = \{ \Delta \}$ und $R(\{ \hat{\imath}_\nu, \widehat{M}_\nu) \}) = \{ \hat{\imath}_\nu \}$, for $1 \le \nu \le n$. By \ref{dn-lambda}(\ref{dn-lambda-2}) there is some $j_{(e, Z), (c, T), \pair{\mathfrak{G} |\mathfrak{F}}} \in A''$, for each $(e, Z) \in \Con$, $\pair{\mathfrak{G} | \mathfrak{F}} \in \pair{\mathcal{G} | \mathcal{F}}$ and $((c, T), b) \in  \mathfrak{G}$, with 
\[
j_{(e, Z), (c, T), \pair{\mathfrak{G} | \mathfrak{F}}} \sim b \,[\rs((c, T), \mathfrak{F})]
\]
 and $((e, c), Z \times T) H j_{(e, Z), (c, T), \pair{\mathfrak{G} | \mathfrak{F}}}$. In case that $((c, T), b) = ((\Delta', \emptyset), \Delta'')$, choose 
 \[
 j_{(e, Z), (\Delta', \emptyset), \pair{\mathfrak{G} |\mathfrak{F}}} = \Delta''.
 \]

\begin{claim}\, \label{cl-lambam16+}
 For every $\mathfrak{J} \subseteq \mathfrak{A}$ and each $e \in R(\pr_1(\mathfrak{J}))$ so that $(e, \mathfrak{J})$ is $\mathfrak{A}$-maximal there is some $t_{e, \mathfrak{J}} \in A''$ with the following two properties:
\begin{align}\label{eq-lambam-17}
\{ \Delta'' \} &\cup \bigcup \{\, \{ \hat{\jmath}_\nu \} \cup \widehat{N}_\nu \mid 1 \le \nu \le n \wedge (\hat{\imath}_\nu, \widehat{M}_\nu) \in \pr_1(\mathfrak{J}) \,\} \cup \{\, t_{c, \mathfrak{K}} \mid  \notag \\
& \mathfrak{K} \subsetneqq \mathfrak{J} \wedge c \sim e \,[\bigcup \pr_{2, 1}(\mathfrak{K})] \wedge \text{$(c, \mathfrak{K})$ $\mathfrak{A}$-maximal} \,\} \cup \{\, j_{(\bar{a}, \overline{M}), (c, T), \pair{\mathfrak{G} | \mathfrak{F}}} \mid  \notag \\
&  (\exists L \subseteq \{ 1, \ldots, n \})(\exists b \in A'') \{\, (\hat{\imath}_\mu, \widehat{M}_\mu) \mid \mu \in L \,\} \subseteq \pr_1(\mathfrak{J}) \wedge \mbox{} \notag \\ 
&T = \bigcup\nolimits_{\mu \in L} X_\mu \wedge c \sim e \,[T] \wedge \pair{\mathfrak{G} | \mathfrak{F}} \in \pair{\mathcal{G} | \mathcal{F}} \wedge ((c, T), b) \in \mathfrak{G} \wedge \mbox{} \notag \\
&b \sim j_{(\bar{a}, \overline{M}), (c, T), \pair{\mathfrak{G} | \mathfrak{F}}} \,[\rs((c, T), \mathfrak{F})] \,\} \in \Con'(t_{e, \mathfrak{J}}),
\end{align}
\begin{align}\label{eq-lambam-18}
((\bar{a}, &e),\, \overline{M} \times \bigcup \pr_{2, 1}(\mathfrak{J})) H (t_{e, \mathfrak{J}}, \{ \Delta'' \} \cup \bigcup \{\, \{ \hat{\jmath}_\nu \} \cup \widehat{N}_\nu \mid 1 \le \nu \le n \wedge \mbox{} \notag \\
&(\hat{\imath}_\nu, \widehat{M}_\nu) \in \pr_1(\mathfrak{J}) \,\} \cup \{\, t_{c, \mathfrak{K}} \mid \mathfrak{K} \subsetneqq \mathfrak{J} \wedge c \sim e \,[\bigcup \pr_{2, 1}(\mathfrak{K})] \wedge \mbox{} \notag \\
&  \text{$(c, \mathfrak{K})$ $\mathfrak{A}$-maximal} \,\} \cup  \{\, j_{(\bar{a}, \overline{M}), (c, T), \pair{\mathfrak{G} | \mathfrak{F}}} \mid (\exists L \subseteq \{ 1, \ldots, n \})(\exists b \in A'') \notag  \\
& T = \bigcup\nolimits_{\mu \in L} X_\mu \wedge c \sim e \,[T] \wedge   \{\, (\hat{\imath}_\mu, \widehat{M}_\mu) \mid \mu \in L \,\} \subseteq \pr_1(\mathfrak{J}) \wedge \pair{\mathfrak{G} | \mathfrak{F}} \in   \notag \\
& \pair{\mathcal{G} | \mathcal{F}} \wedge  ((c, T), b) \in \mathfrak{G} \wedge b \sim j_{(\bar{a}, \overline{M}), (c, T), \pair{\mathfrak{G} | \mathfrak{F}}} \,[\rs((c, T), \mathfrak{F})] \,\}).
\end{align}
\end{claim}

With respect to the last set on the right hand side note that if $((c, T), b) \in \mathfrak{G}$, then there is some $\mathfrak{L} \subseteq \mathfrak{F}$ so that $(c, \mathfrak{L})$ is $\mathfrak{F}$-maximal and $T= \bigcup \pr_{2,1}(\mathfrak{L})$. It follows that $\| \mathfrak{L} \| \le \| \mathfrak{J} \|$.

The claim is demonstrated by induction on the cardinality $\kappa$ of $\mathfrak{J}$. The case $\kappa=0$ is obvious. Set $t_{e, \emptyset} = \Delta''$, if $(e, \emptyset)$ is $\mathfrak{A}$-maximal. All other sets on the left hand side in Statement~(\ref{eq-lambam-17}) are empty in this case, except the last one, which is the singleton $\{ \Delta'' \}$, if for some $\fcbael{G}{F} \in \fcbaset{G}{F}$, $((\Delta', \emptyset), \Delta'') \in \mathfrak{G}$.

Assume next that the claim holds for all $\mathfrak{K} \subseteq \mathfrak{A}$ of cardinality $\kappa$ and let $\mathfrak{J} \subseteq \mathfrak{A}$ of cardinality $\kappa + 1$. Then 
\begin{align*}
((\bar{a},\, &e), \overline{M} \times \bigcup \pr_{2, 1}(\mathfrak{J})) H \{ \Delta'' \} \cup \bigcup \{\, \{ \hat{\jmath}_\nu \} \cup \widehat{N}_\nu \mid 1 \le \nu \le n \wedge (\hat{\imath}_\nu, \widehat{M}_\nu) \in  \\
&\pr_1(\mathfrak{J}) \,\} \cup \{\, t_{c, \mathfrak{K}} \mid \mathfrak{K} \subsetneqq \mathfrak{J} \wedge c \sim e \,[\bigcup \pr_{2, 1}(\mathfrak{K})] \wedge \text{$(c, \mathfrak{K})$ $\mathfrak{A}$-maximal} \,\} \cup \mbox{} \\
&\{\, j_{(\bar{a}, \overline{M}), (c, T), \pair{\mathfrak{G} | \mathfrak{F}}} \mid    (\exists L \subseteq \{ 1, \ldots, n \})(\exists b \in A'')\, T = \bigcup\nolimits_{\mu \in L} X_\mu \wedge  \mbox{}   \\
&c \sim e \,[T] \wedge   \{\, (\hat{\imath}_\mu, \widehat{M}_\mu) \mid \mu \in L \,\} \subseteq \pr_1(\mathfrak{J}) \wedge  \pair{\mathfrak{G} | \mathfrak{F}} \in \pair{\mathcal{G} | \mathcal{F}} \wedge \mbox{} \\
&((c, T), b) \in \mathfrak{G} \wedge b \sim j_{(\bar{a}, \overline{M}), (c, T), \pair{\mathfrak{G} | \mathfrak{F}}} \,[\rs((c, T), \mathfrak{F})] \,\},
\end{align*}
because of Statement~(\ref{eq-lambam-16}), Axiom~\ref{dn-lambda}(\ref{dn-lambda-2}) and the induction hypothesis. By Condition~\ref{dn-am}(\ref{dn-am-7}) there is hence some $t_{e, \mathfrak{J}} \in A''$ so that Properties~(\ref{eq-lambam-17}) and (\ref{eq-lambam-18}) hold.

With help of Claim~\ref{cl-lambam16+} we can now define $\mathfrak{B}$. Let $\mathfrak{B}^{(0)}$ be as in Definition~\ref{dn-ass} and for $\kappa \ge 1$ set
\begin{equation*}
\mathfrak{B}^{(\kappa)} = \set{((e, \pr_{2,1}(\mathfrak{J})), t_{e, \mathfrak{J}})}{\mathfrak{J} \subseteq \mathfrak{A} \wedge \| \mathfrak{J} \| = \kappa \wedge \mbox{} \\
 e \in R(\pr_1(\mathfrak{J})) \wedge \text{$(e, \mathfrak{J})$ $\mathfrak{A}$-maximal}}.
\end{equation*}
Obviously, $\mathfrak{B} \in \ac(\mathfrak{A})$.

\begin{claim}\, \label{cl-lambam-19}
$(\bar{a}, \overline{M}) \Lambda(H) \fcbael{B}{A}$.
\end{claim}

Because of Statement~(\ref{eq-lambam-16}), Condition~\ref{dn-lambda}(\ref{dn-lambda-1}) is satisfied. For Condition~\ref{dn-lambda}(\ref{dn-lambda-2}) let 
\[
((e, T), b) \in \mathfrak{B}.
\]
 Then there is some $\mathfrak{J} \subseteq \mathfrak{A}$ such that $(e, \mathfrak{J})$ is $\mathfrak{A}$-maximal, $T= \bigcup \pr_{2,1}(\mathfrak{J})$ and $b = t_{e, \mathfrak{J}}$. As follows from Properties~(\ref{eq-lambam-17}) and (\ref{eq-lambam-18}), $\pr_2(\mathfrak{J}) \in \Con''(t_{e, \mathfrak{J}})$ and $((\bar{a}, e), \overline{M} \times T) H t_{e, \mathfrak{J}}$. Note that $\pr_2(\mathfrak{J}) = \rs((e, T), \mathfrak{A})$ and choose $j = t_{e, \mathfrak{J}}$.

\begin{claim}\, \label{cl-lambam-19+}
$(\fcbael{B}{A},  \{ \fcbael{B}{A} \}) \vdash_{A' \rightarrow A''} \fcbaset{G}{F}$.
\end{claim}

Let $\fcbael{G}{F} \in \fcbaset{G}{F}$ and $((i_\nu, X_\nu), e_\nu) \in \mathfrak{F}$. Then $(\hat{\jmath}_\nu, \widehat{N}_\nu) \vdash'' e_\nu$. Moreover, $\widehat{N}_\nu \subseteq \ap((i_\nu, X_\nu), \mathfrak{A})$, and by construction of $\mathfrak{B}$, $\{ \hat{\jmath}_\nu \} \in \Con''(\mathfrak{B}_2(i_\nu, X_\nu))$. Thus,
\begin{equation}\label{eq-lambam-20}
(\mathfrak{B}_2(i_\nu, X_\nu), \ap((i_\nu, X_\nu), \mathfrak{A})) \vdash'' e_\nu,
\end{equation}
 which means that Condition~\ref{dn-fctsp}(\ref{dn-fctsp-3-1}) holds. For Condition~\ref{dn-fctsp}(\ref{dn-fctsp-3-2}) let $(c, T), b) \in \mathfrak{G}$ and $L \subseteq \{ 1, \ldots, n \}$ such that $\ds((c, T), \mathfrak{F}) = \set{(i_\nu, X_\nu)}{\nu \in L}$. Then it follows with Statement~(\ref{eq-lambam-20}) that 
\[
 (\mathfrak{B}_2(c, T), \bigcup \set{\ap((i_\nu, X_\nu), \mathfrak{A})}{\nu \in L}) \vdash'' \set{e_\nu}{\nu \in L}.
\] 
 Note that $\set{e_\nu}{\nu \in L} = \rs((c, T), \mathfrak{F})$. By Axiom~\ref{dn-infsys}(\ref{dn-infsys-11}) there is thus some $k \in A''$ with $\set{e_\nu}{\nu \in L} \in \Con''(k)$ and
 \[
  (\mathfrak{B}_2(c, T), \bigcup \set{\ap((i_\nu, X_\nu), \mathfrak{A})}{\nu \in L}) \vdash'' k.
\]
Therefore, $\{ k \} \in \Con''(\mathfrak{B}_2(c, T))$. 

By the definition of $\mathfrak{B}(c, T)$ there is some $\mathfrak{J} \subseteq \mathfrak{A}$ such that $(\mathfrak{B}_{1,1}(c, T), \mathfrak{J})$ is $\mathfrak{A}$-maximal and $\mathfrak{B}_{1,2}(c, T) = \bigcup \pr_{2,1}(\mathfrak{J})$. Moreover, $\set{((\hat{\imath}_\nu, \widehat{M}_\nu), d)}{d \in \widehat{N}_\nu \wedge \nu \in L} \subseteq \mathfrak{J}$ and $\mathfrak{B}_2(c, T) = t_{\mathfrak{B}_{1,1}(c, T), \mathfrak{J}}$. It follows that 
\[
j_{(\bar{a}, \overline{M}), (c, T), \fcbael{G}{F}} \in \Con''(\mathfrak{B}_2(c, T)).
\]
 Since in addition, $j_{(\bar{a}, \overline{M}), (c, T), \fcbael{G}{F}} \sim b \,[\rs((c, T), \mathfrak{F})]$, we obtain that also $k \sim b \,[\rs((c, T), \mathfrak{F})]$.

For Condition~\ref{dn-am}(\ref{dn-am-1}) assume that 
\begin{gather}
(a, S) \Lambda(H) (\fcbael{B}{A}, \fcbaset{W}{V}) \label{eq-lambam-1} \\
(\fcbael{B}{A}, \fcbaset{W}{V}) \vdash_\rightarrow \fcbael{E}{D}. \label{eq-lambam-2} 
\end{gather}
We need to show that $(a, S) \Lambda(H) \fcbael{E}{D}$.

Let to this end $((j, Y), d) \in \mathfrak{D}$. As a consequence of Statement~(\ref{eq-lambam-2}) we obtain that 
\begin{equation}\label{eq-lambam-3}
(\mathfrak{B}_2(j, Y), \ap((j, Y), \bigcup\mathcal{V})) \vdash'' d.
\end{equation}
Moreover, because of Statement~(\ref{eq-lambam-1}), we have for $((i, X), e) \in \mathfrak{A} \cup \bigcup\mathcal{V}$ that
\begin{equation}\label{eq-lambam-4}
((a, i), S \times X) H e.
\end{equation}

Now, suppose that $(i, X) \in \Cl((j, Y), \pr_1(\mathfrak{A} \cup \bigcup\mathcal{V}))$. Then $i \sim j \,[X]$ and hence $(a, i) \sim (a, j) \,[S \times X]$. With Lemma~\ref{lem-witeqam} it therefore follows with Statement~(\ref{eq-lambam-4}) that $((a, j), S \times X) H e$. Note that $X \subseteq \ucl((j, Y), \pr_1(\mathfrak{A} \cup \bigcup\mathcal{V}))$ and $\ucl((j, Y), \pr_1(\mathfrak{A} \cup \bigcup\mathcal{V})) \in \Con'(j)$. Thus,
\[
((a, j), S \times \ucl((j, Y), \pr_1(\mathfrak{A} \cup \bigcup\mathcal{V}))) H e,
\]
by Axiom~\ref{dn-am}(\ref{dn-am-2}). Since $(j, Y) \vdash' \ucl((j, Y), \pr_1(\mathfrak{A} \cup \bigcup\mathcal{V}))$, it ensues that $((a, j), S \times Y) H e$, which shows that
\begin{equation*}\label{eq-lambam-5}
((a, j), S \times Y) H \ap((j, Y), \mathfrak{A} \cup \bigcup \mathcal{V}).
\end{equation*}
Because of Condition~\ref{dn-am}(\ref{dn-am-7}) there is hence some $r \in A''$ with 
\begin{equation}\label{eq-lambam-6}
((a, j), S \times Y) H (r, \ap((j, Y), \mathfrak{A} \cup \bigcup \mathcal{V})).
\end{equation}
Furthermore, as a consequence of Axiom~\ref{dn-lambda}(\ref{dn-lambda-2}) and Statement~(\ref{eq-lambam-1}), there is some $\bar{r} \in A''$ with $((a, j), S \times Y) H (\bar{r}, \ap((j, Y), \mathfrak{A}))$ and 
\[
\bar{r} \sim \mathfrak{B}_2(j, Y) \,[\ap((j, Y), \mathfrak{A})].
\]
 It follows that for some $\hat{r} \in A''$, $((a, j), S \times Y) H \hat{r}$ and $\{ r, \bar{r} \} \in \Con''(\hat{r})$. Thus, $r \sim \mathfrak{B}_2(j, Y) \,[\ap((j, Y), \mathfrak{A})]$. As $\fcbaset{W}{V} \in \Con_\rightarrow(\fcbael{B}{A})$, we obtain with Axiom~\ref{dn-fctsp}(\ref{dn-fctsp-2-2}) that $r \sim \mathfrak{B}_2(j, Y) \,[\ap((j, Y), \bigcup\mathcal{V})]$. With Statement~(\ref{eq-lambam-3}) we therefore have that
\begin{equation*}\label{eq-lambam-7a}
(r, \ap((j, Y), \bigcup\mathcal{V})) \vdash'' d,
\end{equation*}
from which we gain with Statement~(\ref{eq-lambam-6}) that
\[
((a, j), S \times Y) H d.
\]

Next, let $((c, T), e) \in \mathfrak{E}$. Then it follows that
\[
((a, c), S \times T) H \rs((c, T), \mathfrak{D}).
\]
Because of Statement~(\ref{eq-lambam-2}) and Axiom~\ref{dn-fctsp}(\ref{dn-fctsp-3-2}) there is some $k \in A''$ so that 
\[
e \sim k \,[\rs((c, T), \mathfrak{D})]
\] and 
\begin{equation*}
(\mathfrak{B}_2(c, T), \bigcup \set{\ap((j, Y), \bigcup \mathcal{V})}{(j, Y) \in \ds((c, T), \mathfrak{D})}) \vdash'' 
 (k, \rs((c, T), \mathfrak{D})).
\end{equation*}
As in the first part of the proof this implies that $((a, c), S \times T) H k$. Thus, $(a, S) \Lambda(H) \fcbael{E}{D}$.

For Condition~\ref{dn-am}(\ref{dn-am-4}) let $\{ a \} \in \Con(b)$, $S \in \Con(a)$, $(i, X) \in \Con'$ and $c \in A''$ so that $((a, i), S \times X) H c$. We have that $\{ (a, i) \} \in \Con_{A \times A'}(b, i)$ and $S \times X \in \Con_{A \times A'}(a, i)$. Hence, $S \times X \in \Con_{A \times A'}(b, i)$ and $((b, i), S \times X) H c$. Therefore, if $(a, S) \Lambda(H) \fcbael{E}{D}$ then also $(b, S) \Lambda(H) \fcbael{E}{D}$.
\end{proof}

\begin{lemma}\label{lem-lamev1}
$(\Lambda(H) \times \Id_{A'}) \circ \Ev = H$. for all $\apmap{H}{A\times A'}{A''}$.
\end{lemma}
\begin{proof}
Let $((a, b), U) \in \Con_{A \times A'}$ and $c \in A''$ with $((a, b), U) ((\Lambda(H) \times \Id_{A'}) \circ \Ev) c$. Because of Lemma~\ref{pn-amleft} there exists $((a', b'), U') \in \Con_{A \times A'}$ such that
\begin{gather}
((a, b), U) \vdash_{A \times A'} ((a', b'), U') \label{eq-lamev1-1} \\
((a', b'), U') ((\Lambda(H) \times \Id\nolimits_{A'}) \circ \Ev) c. \label{eq-lamev1-2}
\end{gather}
In addition there is some $((\fcbael{B}{A}, \bar{b}), \overline{U}) \in \Con_{(A' \rightarrow A'') \times A'}$ with
\begin{gather}
((a', b'), U') (\Lambda(H) \times \Id\nolimits_{A'}) ((\fcbael{B}{A}, \bar{b}), \overline{U}) \label{eq-lamev1-3} \\
((\fcbael{B}{A}, \bar{b}), \overline{U}) \Ev c \label{eq-lamev1-4},
\end{gather}
where we obtain with Statement~(\ref{eq-lamev1-3})  that
\begin{gather}
(a', \pr_1(U')) \Lambda(H) (\fcbael{B}{A}, \pr_1(\overline{U})) \label{eq-lamev1-5} \\
(b', \pr_2(U')) \vdash' (\bar{b}, \pr_2(\overline{U})). \label{eq-lamev1-6}
\end{gather}
As a consequence of Statement~(\ref{eq-lamev1-5}) we gain that for $((i, X), e) \in \bigcup \pr_{2,1}(\overline{U})$,
\begin{equation}\label{eq-lamev1-7}
((a', i), \pr_1(U') \times X) H e,
\end{equation}
which implies that
\begin{equation}\label{lem-lamev1-8}
((a', \bar{b}), \pr_1(U') \times \pr_2(\overline{U})) H \ap((\bar{b}, \pr_2(\overline{U})), \bigcup \pr_{2,1}(\overline{U})).
\end{equation}
With Statements~(\ref{eq-lamev1-1}) and (\ref{eq-lamev1-6}) we therefore have that
\begin{equation}\label{eq-lamev1-9}
((a, b), U) H \ap((\bar{b}, \pr_2(\overline{U})), \bigcup \pr_{2,1}(\overline{U})).
\end{equation}
By Condition~\ref{dn-lambda}(\ref{dn-lambda-2}) there is now some $j \in A''$ so that 
\begin{gather}
((a, b), U) H (j, \ap((\bar{b}, \pr_2(\overline{U})), \bigcup \pr_{2,1}(\overline{U}))) \label{eq-lamev1-10} \\
j \sim \mathfrak{B}_2(\bar{b}, \pr_2(\overline{U})) \,[\ap((\bar{b}, \pr_2(\overline{U})), \bigcup \pr_{2,1}(\overline{U}))]. \label{lamev1-11}
\end{gather}
Since by Statement~(\ref{eq-lamev1-4})
\[
(\mathfrak{B}_2(\bar{b}, \pr_2(\overline{U})), \ap((\bar{b}, \pr_2(\overline{U})), \bigcup \pr_{2,1}(\overline{U}))) \vdash'' c,
\]
it follows that also
\[
(j, \ap((\bar{b}, \pr_2(\overline{U})), \bigcup \pr_{2,1}(\overline{U}))) \vdash'' c.
\]
Together with Statement~(\ref{eq-lamev1-10}) we thus obtain that $((a, b), U) Hc$.

Next assume conversely that $((a, b), U) Hc$. Because of Lemmas~\ref{pn-amleft}, \ref{pn-amright} and \ref{pn-amint} there exist $((a', b'), U')$, $((a'', b''), U'') \in \Con_{A \times A'}$ and $(j, V) \in \Con''$ such that
\begin{gather}
((a, b), U) \vdash_{A \times A'} ((a', b'), U') \label{eq-lamev1-12+} \\
((a', b'), U') \vdash_{A \times A'} ((a'', b''), U'') \\
((a'', b''), U'') H (j, V \cup \{ \Delta'' \}) \\
(j, V) \vdash'' c. \label{eq-lamev1-12}
\end{gather}
Set 
\[
\mathfrak{A} = \set{((b'', \pr_2(U'')), d)}{d \in V}
\]
 and let $\mathfrak{B} \in \ac(\mathfrak{A})$ be as in Lemma~\ref{lem-singext}. Then we have for $(i, X) \in \Con'$ that
\[
((a'', i), \pr_1(U'') \times X) H (\mathfrak{B}_2(i, X), \ap((i, X), \mathfrak{A})).
\]
This shows that $(a'', \pr_1(U'')) \Lambda(H) \fcbael{B}{A}$. Moreover, we obtain with Statement~(\ref{eq-lamev1-12+}) that
\begin{equation}\label{eq-lamev1-13}
((a, b), U) (\Lambda(H) \times \Id\nolimits_{A'}) ((\fcbael{B}{A}, b'), \{ \fcbael{B}{A} \} \times \pr_2(U')).
\end{equation}
Now, note that $\ap((b', \pr_2(U')), \mathfrak{A}) = V$ and $\mathfrak{B}_2(b', \pr_2(U')) = j$. Therefore, it follows with Statement~(\ref{eq-lamev1-12}) that
\[
(\mathfrak{B}_2(b', \pr_2(U')), \ap((b', \pr_2(U')), \mathfrak{A})) \vdash'' c.
\]
Let $\mathbb{Z} = \{ \fcbael{B}{A} \} \times \pr_2(U')$. Then $\mathbb{Z} \in \Con_{(A' \rightarrow A'') \times A'}(\fcbael{B}{A}, b')$. 
In addition, we have that $((\fcbael{B}{A}, b'), \mathbb{Z}) \Ev c$. With Statement~(\ref{eq-lamev1-13}) we now gain that 
\[
((a, b), U) ((\Lambda(H) \times \Id\nolimits_{A'}) \circ \Ev) c
\] 
as was to be shown.
\end{proof}

\begin{lemma}\label{lem-lamev2}
$\Lambda((G \times \Id_{A'}) \circ \Ev) = G$, for all $\apmap{G}{A}{A' \rightarrow A''}$.
\end{lemma}
\begin{proof}
Let $(a, V) \in \Con$ and $\fcbael{W}{V} \in A' \rightarrow A''$ so that
\begin{equation}\label{eq-lamev2-0}
(a, V) \Lambda((G \times \Id\nolimits_{A'}) \circ \Ev) \fcbael{W}{V}.
\end{equation}
Moreover, assume that $\mathfrak{V} = \set{((i_\nu, X_\nu), e_\nu)}{1 \le\nu \le n}$. Then we have for every $1 \le \nu \le n$ that
\[
((a, i_\nu), V \times X_\nu) ((G \times \Id\nolimits_{A'}) \circ \Ev) e_\nu.
\]
Thus, there are $((\pair{\mathfrak{B}^{(\nu)} | \mathfrak{A}^{(\nu)}}, b_\nu), \mathbb{Z}_\nu) \in \Con_{(A' \rightarrow A'') \times A'}$ with
\begin{gather}
((a, i_\nu), V \times X_\nu) (G \times \Id\nolimits_{A'}) ((\pair{\mathfrak{B}^{(\nu)} | \mathfrak{A}^{(\nu)}}, b_\nu), \mathbb{Z}_\nu)  \label{eq-lamev2-1} \\
((\pair{\mathfrak{B}^{(\nu)} | \mathfrak{A}^{(\nu)}}, b_\nu), \mathbb{Z}_\nu)  \Ev e_\nu. \label{eq-lamev2-2}
\end{gather}
So, we obtain:
\begin{gather}
(a, V) G (\pair{\mathfrak{B}^{(\nu)} | \mathfrak{A}^{(\nu)}}, \pr_1(\mathbb{Z}_\nu)) \label{eq-lamev2-3} \\
(i_\nu, X_\nu) \vdash' (b_\nu, \pr_2(\mathbb{Z}_\nu)) \label{eq-lamev2-4} \\
(\mathfrak{B}_2^{(\nu)}(b_\nu, \pr_2(\mathbb{Z}_\nu)), \ap((b_\nu, \pr_2(\mathbb{Z}_\nu)), \bigcup \pr_{2,1}(\mathbb{Z}_\nu)) \vdash'' e_\nu. \label{eq-lamev2-5}
\end{gather}
Because of Statement~(\ref{eq-lamev2-4}) we moreover have that
\begin{gather}
\{ \mathfrak{B}_2^{(\nu)}(b_\nu, \pr_2(\mathbb{Z}_\nu)) \} \in \Con'(\mathfrak{B}_2^{(\nu)}(i_\nu, X_\nu)) \label{eq-lamev2-6} \\
\ap((b_\nu, \pr_2(\mathbb{Z}_\nu)), \bigcup \pr_{2,1}(\mathbb{Z}_\nu)) \subseteq \ap((i_\nu, X_\nu), \bigcup \pr_{2,1}(\mathbb{Z}_\nu)) \label{eq-lamev2-7}
\end{gather}
and hence that
\begin{equation} \label{eq-lamev2-8}
(\mathfrak{B}_2^{(\nu)}(i_\nu, X_\nu), \ap((i_\nu, X_\nu), \bigcup \pr_{2,1}(\mathbb{Z}_\nu))) \vdash'' e_\nu.
\end{equation}
By Statement~(\ref{eq-lamev2-3}) and Lemma~\ref{pn-amright} there exists $(\fcbael{E}{D}, \fcbaset{E}{D}) \in \Con_{A' \rightarrow A''}$ with
\begin{gather}
(a, V) G (\fcbael{E}{D}, \fcbaset{E}{D}) \label{eq-lamev2-9} \\
(\fcbael{E}{D}, \fcbaset{E}{D}) \vdash_{A' \rightarrow A''} (\pair{\mathfrak{B}^{(\nu)} | \mathfrak{A}^{(\nu)}}, \pr_1(\mathbb{Z}_\nu)), \label{eq-lamev2-10}
\end{gather}
for all $1 \le \nu \le n$. Thus, $\{ \pair{\mathfrak{B}^{(\nu)} | \mathfrak{A}^{(\nu)}} \} \in \Con_{A' \rightarrow A''}(\fcbael{E}{D})$. By Condition~\ref{dn-fctsp}(\ref{dn-fctsp-2-1}) it follows that 
\[
\mathfrak{E}_2(i_\nu, X_\nu) \sim \mathfrak{B}_2^{(\nu)}(i_\nu, X_\nu) \,[\ap((i_\nu, X_\nu), \mathfrak{A}^{(\nu)})].
\] 
Since $\pr_1(\mathbb{Z}_\nu) \in \Con_{A' \rightarrow A''}(\pair{\mathfrak{B}^{(\nu)} | \mathfrak{A}^{(\nu)}})$, we hence obtain with Condition~\ref{dn-fctsp}(\ref{dn-fctsp-2-2}) that moreover
\[
\mathfrak{E}_2(i_\nu, X_\nu) \sim \mathfrak{B}_2^{(\nu)}(i_\nu, X_\nu) \,[\ap((i_\nu, X_\nu), \bigcup \pr_{2,1}(\mathbb{Z}_\nu))].
\]
As a consequence of Statement~(\ref{eq-lamev2-8}) we therefore gain that
\[
(\mathfrak{E}_2(i_\nu, X_\nu), \ap((i_\nu, X_\nu), \bigcup \pr_{2,1}(\mathbb{Z}_\nu))) \vdash'' e_\nu.
\]
Because it follows with Statement~(\ref{eq-lamev2-10}) that
\[
(\mathfrak{E}_2(i_\nu, X_\nu), \ap((i_\nu, X_\nu), \bigcup \mathcal{D})) \vdash'' \ap((i_\nu, X_\nu), \bigcup \pr_{2,1}(\mathbb{Z}_\nu)),
\]
we conclude that
\begin{equation*}\label{eq-lamev2-11}
(\mathfrak{E}_2(i_\nu, X_\nu), \ap((i_\nu, X_\nu), \bigcup \mathcal{D})) \vdash'' e_\nu,
\end{equation*}
from which we additionally gain that for $((t, S), c) \in \mathfrak{W}$,
\[
(\mathfrak{E}_2(t, S), \bigcup \set{\ap((d, U), \bigcup \mathcal{D})}{(d, U) \in \ds((t, S), \mathfrak{V})}) \vdash'' \rs((t, S), \mathfrak{V}).
\]
Because of Axiom~\ref{dn-infsys}(\ref{dn-infsys-11}) there is thus some $j \in A''$ so that
\begin{gather*}
\rs((t, S), \mathfrak{V}) \in \Con''(j) \\
(\mathfrak{E}_2(t, S), \bigcup \set{\ap((d, U), \bigcup \mathcal{D})}{(d, U) \in \ds((t, S), \mathfrak{V})}) \vdash'' j. \label{eq-lamev2-12+}
\end{gather*}
It follows that also
\[
(\mathfrak{E}_2(t, S), \ap((t, S), \bigcup \mathcal{D})) \vdash'' (j, \rs((t, S), \mathfrak{V})).
\]

By definition, $(t, S) \vdash' \ucl((t,S), \pr_1(\bigcup \mathcal{D} \cup \mathfrak{D}))$. Hence, there is some $(b, T) \in \Con'$ with
\begin{gather*}
(t, S) \vdash' (b, T) \\
(b, T) \vdash' \ucl((t,S), \pr_1(\bigcup \mathcal{D} \cup \mathfrak{D})).
\end{gather*}
Then 
\[
\ap((b, T), \bigcup \mathcal{D}) = \ap((t, S), \bigcup \mathcal{D})\,\,\text{and}\,\,\mathfrak{E}_2(b, T) = \mathfrak{E}_2(t, S).
\]
Thus, we also have that
\begin{equation*}\label{eq-lamev2-12}
(\mathfrak{E}_2(b, t), \ap((b, T), \bigcup \mathcal{D})) \vdash'' (j, \rs((t, S), \mathfrak{V})).
\end{equation*}

So far we have shown that
\begin{gather*}
((a, t), V \times S) (G \times \Id\nolimits_{A'}) ((\fcbael{E}{D}, b), \fcbaset{E}{D} \times T) \\
((\fcbael{E}{D}, b), \fcbaset{E}{D} \times T) \Ev (j, \rs((t, S), \mathfrak{V})),
\end{gather*}
from which we obtain that
\[
((a, t), V \times S) ((G \times \Id\nolimits_{A'}) \circ \Ev) (j, \rs((t, S), \mathfrak{V})).
\]
Because of Statement~(\ref{eq-lamev2-0}) and Condition~\ref{dn-lambda}(\ref{dn-lambda-2}) there is also some $\bar{\jmath} \in A''$ so that 
\[
((a, t), V \times S) ((G \times \Id\nolimits_{A'}) \circ \Ev) \bar{\jmath} \,\,\text{and}\,\, \bar{\jmath} \sim c \,[\rs((t, S), \mathfrak{V})].
\]
With Condition~\ref{dn-am}(\ref{dn-am-7}) there then some $k \in A''$ such that $\{ j, \bar{\jmath} \} \in \Con''(k)$ and $((a, t), V \times S) ((G \times \Id\nolimits_{A'}) \circ \Ev) k$. It follows that also $j \sim c \,[\rs((t, S), \mathfrak{V})]$. So, we have shown that
\[
(\fcbael{E}{D}, \fcbaset{E}{D}) \vdash_{A' \rightarrow A''} \fcbael{W}{V},
\]
from which we obtain with Statement~(\ref{eq-lamev2-9}) that $(a, V) G \fcbael{W}{V}$.

Next, conversely, let $(a, V) \in \Con$ and $\fcbael{W}{V} \in A' \rightarrow A''$ with
\[
(a, V) G \fcbael{W}{V}.
\]
Then there is some $(\fcbael{B}{A}, \fcbaset{G}{F}) \in \Con_{A' \rightarrow A''}$ with
\begin{gather}
(a, V) G (\fcbael{B}{A}, \fcbaset{G}{F}) \\
(\fcbael{B}{A}, \fcbaset{G}{F}) \vdash_{A' \rightarrow A''} \fcbael{W}{V}. \label{eq-lamev2-13}
\end{gather}
Let $\mathfrak{V} = \set{(i_\nu, X_\nu), e_\nu)}{1 \le \nu \le m}$. Then we obtain for $1 \le \nu \le m$ that
\[
(\mathfrak{B}_2(i_\nu, X_\nu), \ap((i_\nu, X_\nu), \bigcup \mathcal{F})) \vdash'' e_\nu.
\]
Moreover, there are $(b_\nu, Y_\nu) \in \Con'$, for $1 \le \nu \le m$, with
\begin{gather*}
(i_\nu, X_\nu) \vdash' (b_\nu, Y_\nu) \\
(b_\nu, Y_\nu) \vdash' \ucl((i_\nu, X_\nu), \pr_1(\bigcup \mathcal{F} \cup \mathcal{A})).
\end{gather*}
It follows that
\[
\ap((i_\nu, X_\nu), \bigcup\mathcal{F}) = \ap((b_\nu, Y_\nu), \bigcup\mathcal{F}) \,\,\text{and}\,\, 
\mathfrak{B}_2(i_\nu, X_\nu) = \mathfrak{B}_2(b_\nu, Y_\nu).
\]
Thus,
\[
(\mathfrak{B}_2(b_\nu, Y_\nu), \ap((b_\nu, Y_\nu), \bigcup \mathcal{F})) \vdash'' e_\nu.
\]
Set $\mathbb{Z}_\nu = \fcbaset{G}{F} \times Y_\nu$. Then $\mathbb{Z}_\nu \in \Con_{(A' \rightarrow A'') \times A'}(\fcbael{B}{A}, b_\nu)$ and
\begin{gather*}
((a, i_\nu), V \times X_\nu) (G \times \Id\nolimits_{A'}) ((\fcbael{B}{A}, b_\nu), \mathbb{Z}_\nu) \\
((\fcbael{B}{A}, b_\nu), \mathbb{Z}_\nu) \Ev e_\nu,
\end{gather*}
which means we have that
\[
((a, i_\nu), V \times X_\nu) ((G \times \Id\nolimits_{A'}) \circ \Ev) e_\nu.
\]
It remains to verify Condition~\ref{dn-lambda}(\ref{dn-lambda-2}). Let to this end $((c, T), d) \in \mathfrak{W}$. Because of Statement~(\ref{eq-lamev2-13}) and Condition~\ref{dn-fctsp}(\ref{dn-fctsp-3-2}) there is some $k \in A''$ with 
\begin{gather}
k \sim d \,[\rs((c, T), \mathfrak{V})] \label{eq-lamev2-14} \\
(\mathfrak{B}_2(c, T), \bigcup \set{\ap((e, U), \bigcup \mathcal{F})}{(e, U) \in \ds((c, T), \mathfrak{V})}) \vdash'' k.
\end{gather}
It follows that also
\[
(\mathfrak{B}_2(c, T), \ap((c, T), \bigcup \mathcal{F})) \vdash'' (k, \rs((c, T), \mathfrak{V})).
\]

Let $(\bar{b}, \overline{Y}) \in \Con'$ such that
\begin{gather*}
(c, T) \vdash' (\bar{b}, \overline{Y}) \\
(\bar{b}, \overline{Y}) \vdash' \ucl((c, T), \pr_1(\bigcup \mathcal{F} \cup \mathfrak{A})).
\end{gather*}
Then 
\[
\ap((c, T), \bigcup \mathcal{F}) = \ap((\bar{b}, \overline{Y}), \bigcup \mathcal{F}) \,\,\text{and}\,\, \mathfrak{B}_2(c, T) = \mathfrak{B}_2(\bar{b}, \overline{Y})
\]
and hence
\[
(\mathfrak{B}_2(\bar{b}, \overline{Y}), \ap((\bar{b}, \overline{Y}), \bigcup \mathcal{F})) \vdash'' (k, \rs((c, T), \mathfrak{V})).
\]
Set $\mathbb{Z} = \fcbaset{G}{F} \times \overline{Y}$. Then $\mathbb{Z} \in \Con_{(A' \rightarrow A'') \times A'}(\fcbael{B}{A}, \bar{b}))$ and
\begin{gather*}
((a, c), V \times T) (G \times \Id\nolimits_{A'})((\fcbael{B}{A}, \bar{b}), \mathbb{Z}) \\
((\fcbael{B}{A}, \bar{b}), \mathbb{Z}) \Ev (k, \rs((c, T), \mathfrak{V})),
\end{gather*}
which implies that
\[
((a, c), V \times T) ((G \times \Id\nolimits_{A'}) \circ \Ev) k.
\]
Because of Statement~(\ref{eq-lamev2-14}) this shows that 
\[
(a, V) \Lambda((G \times \Id\nolimits_{A'}) \circ \Ev) \fcbael{W}{V}. \qedhere
\]
\end{proof}

\begin{proposition}\label{pn-exp}
Let $A$ and $A'$ be information system with witnesses. Then $(A \rightarrow A', \Ev)$ is their exponent in $\mathbf{ISW}$.
\end{proposition}

With Propositions~\ref{prop-fct-alg} and \ref{prop-fct-bc} we moreover have that if both $A$ and $A'$ satisfy Condition~(\ref{alg}), (\ref{bc}), or both of them, then $(A \rightarrow A', \Ev)$ is their exponent in \textbf{aISW}, \textbf{bcISW}, and \textbf{abcISW}, respectively.

As we have already seen, \text{ISW} as well as \textbf{aISW}, \textbf{bcISW}, and \textbf{abcISW} contain a terminal object. Moreover, we have shown how to construct the categorical product of information systems with witnesses.

\begin{theorem}\label{thm-ccc}
The category \textbf{ISW} of information systems with witnesses and approximable mappings as well as its full subcategories \textbf{aISW}, \textbf{bcISW}, and \textbf{abcISW}, respectively, of information systems with witnesses satisfying Condition~(\ref{alg}), (\ref{bc}), or both of them are Cartesian closed.
\end{theorem}

\section{Final remarks}\label{sec-finrem}

This paper is a continuation of Ref.\ \cite{sp21} where information systems with witnesses were introduced as a logic-style representation of L-domains: the category of  these information systems with approximable mappings as morphisms was shown to be equivalent to the category of L-domains with Scott continuous functions. As demonstrated in Ref. \cite{ju89}, the latter category is one of the two Cartesian closed full subcategories of the category of continuous domains. It follows in particular that the category of information systems with witnesses  is Cartesian closed as well. In the present paper a direct construction of exponentiation in this category is presented. 

Logic-style representations of domains---in particular those forming a Cartesian closed category---allow incorporating the domains into proof assistants, as done in the Minlog system, developed by the Munich logic group (cf. Ref.\ \cite{min}).

Minlog  is an interactive proof system based on first order natural deduction calculus. It is intended to reason about higher-type computable functionals, using minimal rather than classical or intuitionistic logic. Minlog implements a \emph{theory of computable functionals}. The underlying semantics is the Scott-Ershov model of partial continuous functionals, with \emph{free algebras} as base types. These algebras are viewed as domains represented by Scott's information systems, whose tokens are constructor trees possibly involving the symbol $\ast$ (``no information") (Ref.\ \cite{beal}).

By using information systems with witnesses instead of Scott's information systems a larger class of data structures and computable functionals can be considered. Examples of topological hyperspaces are known that are L-domains with respect to superset inclusion, but are not bounded-complete (cf.\  e.g.\ Ref.\ \cite[p.\ 58]{ju89}).

\end{document}